\documentclass[10pt]{elsarticle}
\usepackage{amssymb,amsmath}
\usepackage{enumitem}
\usepackage{bbm}
\usepackage{upgreek}
\usepackage{mathtools}
\usepackage[margin=1in,papersize={8.5in,11in}]{geometry}
\usepackage[utf8]{inputenc}
\usepackage[english]{babel}
\usepackage{amsthm}
\usepackage{float}
\usepackage{amsfonts,mathrsfs}
\usepackage{bm}
\usepackage{makecell}
\usepackage{tikz}
\usetikzlibrary{matrix}

\usepackage{epstopdf}
\usepackage{subcaption}
\usepackage{tablefootnote}
\graphicspath{ {images/} }
\usepackage{comment}

\newtheorem{theorem}{Theorem}
\newtheorem{corollary}{Corollary}[theorem]

\newtheorem{prop}{Proposition}
\theoremstyle{definition}

\def\R{\mathbb{R}}
\def\H{\mathcal{H}}
\def\O{\mathcal{O}}

\def\W{\mathcal{W}}
\def\A{\mathcal{A}}
\def\B{\mathcal{B}}
\def\D{\mathcal{D}} 
\def\S{\mathcal{S}} 
\def\L{\mathcal{L}} 
\def\M{\mathcal{M}}
  
\def\N{\mathcal{N}}

\def\la{\langle} 
\def\ra{\rangle_{\rho}} 
 
\def\P{\mathcal{P}}

\def\Q{\mathcal{Q}}
\def\E{\mathcal{E}}

\def\G{\mathcal{G}}
\def\K{\mathcal{K}} 
\def\I{\mathcal{I}}

\def\Pi{\mathbf{P}}

\journal{arXiv}

\begin{document}
\begin{frontmatter}
\title{Generalized second fluctuation-dissipation theorem in the nonequilibrium steady state: Theory and applications}

\author[ucm]{Yuanran Zhu\corref{correspondingAuthor}}
\cortext[correspondingAuthor]{Corresponding author}
\ead{yzhu56@ucmerced.edu}
\author[msu]{Huan Lei}
\author[ucm]{Changho Kim}
\address[ucm]{Department of Applied Mathematics, University of California, Merced, CA 95343}
\address[msu]{Department of Computational Mathematics, Science \& Engineering and Department of Statistics \& Probability, Michigan State  University, MI 48824}

\begin{abstract}
In this paper, we derive a generalized second fluctuation-dissipation theorem (FDT) for stochastic dynamical systems in the steady state. The established theory is built upon the Mori-type generalized Langevin equation for stochastic dynamical systems and only uses the properties of the Kolmogorov operator. The new second FDT expresses  the memory kernel of the generalized Langevin equation as the correlation function of the fluctuation force plus an additional term. In particular, we show that for nonequilibrium states such as heat transport between two thermostats with different temperatures, the classical second FDT is valid even when the exact form of the steady state distribution is unknown. The obtained theoretical results  enable us to construct a data-driven nanoscale fluctuating heat conduction model based on the second FDT. We numerically verify that the new model of heat transfer yields better predictions than the Green-Kubo formula for systems far from the equilibrium.
\end{abstract}
\end{frontmatter}

\section{Introduction}
\label{sec:Intro}
The fluctuation-dissipation relations are one of the most important results in statistical mechanics. In 1966, Kubo, in his renowned paper \cite{kubo1966fluctuation}, proposed a general relationship that connects the response of a given system to the external perturbation and the internal thermal fluctuation of the system in the absence of disturbance. This groundbreaking result, together with many generalizations for open systems, is now categorized as the fluctuation-dissipation theorems (FDTs) which provide a far reaching generalization of Einstein's theory of Brownian motion and Nyquist's work on thermal noise in electrical resistors. In his original paper \cite{kubo1966fluctuation}, Kubo presented two results which characterize the fluctuation-dissipation relationship for arbitrary Hamiltonian systems. The aforementioned response result is called the {\em first} FDT, whereas another relationship which builds connections between the thermal fluctuation of an observable and the memory kernel for the corresponding generalized Langevin equation (GLE) is now known as the {\em second} FDT. It is not widely known that Kubo actually provided two derivation methods for these two relationships. The first method is based on a phenomenological linear GLE for a system observable $u(t)$:
\begin{align}\label{intro_pheno_GLE}
\frac{d}{dt}u(t)=-\int_0^tK(t-s)u(s)ds+\frac{1}{m}f(t)+\frac{1}{m}g(t),
\end{align}
where $K(t)$ is known as the memory kernel and $f(t)$ is the fluctuation force satisfying $\langle f(t)\rangle=0$ and $\langle u(s)f(t)\rangle=0$ for $t>s$. Here $\langle\cdot\rangle$ is the ensemble average in the thermal equilibrium. The system is assumed to be perturbed by a periodic external force $g(t)=g_0\cos(\omega t)$. For such a case, the first FDT is proved using the standard Fourier-Laplace transform and the second FDT is derived from the first FDT hence can be viewed as a corollary of it \cite{kubo1966fluctuation}.
The second method differs for the derivation of these two FDTs. In particular, the first FDT is derived using a perturbation method which is known as the linear response theory. During the derivation, the phenomenological GLE \eqref{intro_pheno_GLE} is {\em not} used. The second FDT, however, can be derived using Mori's equation \cite{Mori}, which can be interpreted as the \eqref{intro_pheno_GLE} in an operator form. In fact, since the evolution operator for the Hamiltonian system is given by $e^{t\L}$, the second FDT is a natural result of the skew-symmetry of the Liouville operator $\L$ with respect to the Hilbert space inner product $\langle\cdot,\cdot\rangle_{\rho_{eq}}$, where $\rho_{eq}$ is the equilibrium distribution. Detailed explanations can be found in Section \ref{sec:EMZE}.

Since the establishment of the FDTs in the 1960s, there has been a considerable amount of work on its verification, violation, generalization and applications. In fact, the FDT-related linear response relation has became the cornerstone of nonequilibrium statistical mechanics, which provides a powerful tool to study various transport phenomena for near-equilibrium systems \cite{snook2006langevin,morriss2013statistical}. Here we only mention some recent studies in active matter \cite{dal2019linear}, turbulence \cite{majda2005information,gritsun2008climate} and heat conduction \cite{kundu2009green,lepri1998anomalous,dhar2008heat}. We note that most of these studies are about the {\em first} FDT and relatively less attention has been paid to the {\em second} FDT. For most applications, the second FDT was required for the construction of the GLE when it was used as a reduced-order model for particular observables. Its generalization and verification has not yet been widely studied except for Mae's recent work \cite{maes2014second}. One of the reasons why this is the case is that, unlike the {\em first} FDT which is directly related to the Green-Kubo formula for transport theory, the study on the {\em second} FDT is difficult to translate into useful, experimentally verifiable results. In fact, Kubo's second derivation method reviewed in the previous paragraph clearly indicates the first and the second FDTs are actually quite different. Specifically, the second FDT is an intrinsic property for Hamiltonian systems in the thermodynamic equilibrium and its validity does not rely on perturbation arguments. 
This observation hints that one can possibly derive the second FDT or its generalized form from the operator-form GLE, i.e. the Mori-Zwanzig equation \cite{Mori,zwanzig1961memory}, for generalized stochastic dynamics in the nonequilibrium steady state. In addition, noting a gradually increasing interest in open systems and far-from-equilibrium phenomena, we anticipate that a classical or generalized second FDT for the nonequilibrium system will provide insights on the system dynamics. All these thoughts motivated the current research.

The main purpose of this paper is two-fold. First, we will follow Kubo-Mori's methodology to establish a generalized second FDT for stochastic dynamical systems driven by Gaussian white noise. Such stochastic models are widely used in coarse-graining modeling for molecules \cite{espanol1995hydrodynamics,espanol1995statistical,lei2016data,hudson2018coarse,Grogan_Lei_JCP_2020}, nonequilibrium heat conduction models \cite{eckmann1999non,lepri1998anomalous,kundu2009green} and many other open systems \cite{majda2005information,majda2010quantifying}. Secondly, the second FDT is shown to be valid in the nonequilibrium steady state, therefore can be applied to address the far-from-equilibrium transport problem.

This article is organized as follows. Section \ref{sec:short_sum} provides a short summary of the theoretical results obtained in this paper. Section \ref{sec:EMZE} briefly reviews the derivation of the GLE for stochastic systems using the Mori-Zwanzig equation. In Section \ref{sec:FDT}, as a comparison, we first review the derivation of the first FDT for stochastic systems via perturbation analysis, then we use the GLE to derive a generalized second FDT which holds in the nonequilibrium steady state. In Section \ref{sec:Application_to_specifics}, this newly established relation is applied to equilibrium and nonequilibrium systems commonly used in the statistical mechanics. In particular, we show the validity of the classical second FDT for the averaged heat flux in a heat conduction model. The theoretical results are verified numerically in Section \ref{sec:Numerical} via the molecular dynamics simulation of the heat conduction model. In addition, two reduced-order models are proposed and shown to faithfully characterize the thermal conduction for far-from-equilibrium systems. The main findings of the paper are summarized in Section \ref{sec:conclusion}. 
\section{The main theoretical results}
\label{sec:short_sum}
In this section, we give a short summary of the main theoretical results obtained in this paper and provide simple examples to explain them. In statistical physics, open systems normally refer to dynamical systems in contact with thermal reservoirs. These reservoirs are modeled by deterministic or stochastic thermostats. If a nonequilibrium condition is imposed, for instance one may choose thermostats with different temperatures and attach them to a Hamiltonian system, then the external force exerted by the thermostats will drive the system {\em out of} the thermal equilibrium. In this paper, we are mainly concerned with the nonequilibrium systems with stochastic thermostats. Mathematically, they can be described by a general stochastic differential equation (SDE):
\begin{align}\label{sum:sde}
    d\bm x(t)=\bm F(\bm x(t))+\sigma(\bm x(t))d\bm \W(t),\qquad \bm x(0)=\bm x_0\sim\rho_0(\bm x),
\end{align}
Stochastic analysis already told us \cite{mattingly2002ergodicity} after a finite transient time, SDE \eqref{sum:sde} satisfying some suitable conditions will converge to a unique steady state. For stochastic force $\sigma(\bm x(t))d\bm \W(t)$ in a nonequilibrium setting, such a steady state is normally termed as the nonequilibirum steady state (NESS). The primary goal of the current paper is to investigate whether the classical second FDT still holds in the NESS of SDEs. To this end, we follow the aforementioned Mori-Kubo's methodology and use the following GLE for the (open) stochastic system as the starting point of our analysis:
\begin{align}\label{intro_MZ}
   \frac{d}{dt}u(t)
=\Omega u(t)+\int_0^tK(t-s)u(s)ds+f(t),
\end{align}
where $u(t)=e^{t\K}u(0)$ is an arbitrary observable function of the stochastic system.  Using the properties of the Kolmogorov backward operator $\K$ (generator of the stochastic dynamics), we proved the following result regarding the validity of the second FDT:

\vspace{0.2cm}
(I) In GLE \eqref{intro_MZ}, the memory kernel $K(t)$ and the fluctuation force $f(t)$ satisfy a generalized second FDT:
\begin{align}\label{intro_FDT}
    K(t)=-\frac{\langle f(t),f(0)\rangle_{\rho}}{\langle u^2(0)\rangle_{\rho}}+\frac{\langle f(t),w(0)\rangle_{\rho}}{\langle u^2(0)\rangle_{\rho}},
\end{align}
where $\langle \cdot\rangle_{\rho}$ is the ensemble average in the NESS. In \eqref{intro_FDT}, the first term yields the {\em classical} second FDT. The additional term $\langle f(t),w(0)\rangle_{\rho}/\langle u^2(0)\rangle_{\rho}$ constitutes the generalized relation that the current paper studies. 
However, since $w(0)$ in \eqref{intro_FDT} depends on the steady state probability density $\rho$ (see its explicit expression \eqref{into_w(0)}), which is hard to obtain for most high-dimensional stochastic systems, this generalized relation cannot be used directly. To overcome this difficulty, we notice an important fact about the obtained second FDT \eqref{intro_FDT} which is the second main result of this paper: 

\vspace{0.2cm}
(II) The additional observable $w(0)$ can be explicitly written as 
\begin{align}\label{into_w(0)}
w(0)=\S u(0)=2\sigma^2\Delta u(0)+2\sigma^2\nabla(\ln\rho)\cdot\nabla,
\end{align}
where $\S$ is a second-order differential operator in the {\em non-degenerate} coordinate of the stochastic system. If in addition, observable $u(0)$ is a function of the {\em degenerate} coordinate, then $w(0)=\S u(0)=0$ and the {\em generalized} second FDT \eqref{intro_FDT} degenerates to the {\em classical} second FDT.

\vspace{0.2cm}
\noindent
We can use the Langevin dynamics as an example to explain the meaning of the degenerate coordinate and result (II). As it is well-known, the Gaussian white noise for a Langvein dynamics is only imposed in the momentum coordinate $p_i$. We shall call the {\em momentum coordinate} $p_i$ non-degenerate and the {\em position coordinate} $q_i$ degenerate. Bearing this in mind, the reason why the special case discussed in (II) has $w(0)=\S u(0)=0$ becomes obvious. Since the initial condition of the observable $u(0)=f(q_i(0))$, then we have
\begin{align*}
  w(0)=\S u(0)=\S f(q_i(0))=2\sigma^2\partial^2_{p_i}f(q_i(0)) +2\sigma^2\partial_{p_i}(\ln\rho)\partial_{p_i}f(q_i(0))=0. 
\end{align*}
This result can be physically interpreted as follows: For the linear GLE of an arbitrary (open) stochastic system, the validity of the classical second FDT of an observable $u(t)$ depends on whether $u(t)$ has {\em direct} interactions with the environment through the white noise. In reality, many nonequilibrium systems can be modeled by highly degenerate stochastic systems. Hence most observables in such nonequilibrium systems satisfy the {\em classical} second FDT. We will use this fact and the GLE \eqref{intro_MZ} to build effective reduced-order models for low-dimensional observables. In Section \ref{sec:Numerical}, we show the application of such an effective model in studying the far-from equilibrium transport phenomena, which are generally hard to handle using the linear response theory.

\section{Mori-type generalized Langevin equations for SDEs}
\label{sec:EMZE}
The starting point of our analysis is the Mori-Zwanzig (MZ) equation for stochastic differential equations (SDEs). Such stochastic system MZ equation was derived independently by different researchers using various methods \cite{espanol1995hydrodynamics,morita1980contraction,zhu2019generalized}. Here we briefly review the derivation used in \cite{,zhu2019generalized}; more detailed explanations can be found therein. Let us consider a $d$-dimensional SDE on a smooth manifold $M$
\begin{align}\label{eqn:sde}
d\bm x(t)=\bm F(\bm x(t))dt+\bm \sigma(\bm x(t))d\bm\W(t), \qquad \bm x(0)=\bm x_0\sim \rho_0(\bm x),
\end{align}
where $\bm F:M\rightarrow \R^d$ and 
$\bm \sigma: M\rightarrow \R^{d\times m}$ are 
smooth functions, $\bm \W(t)$ is 
an $m$-dimensional Wiener process with 
independent components, and $\bm x_0$ is a 
random initial state characterized in terms of 
a probability density function $\rho_0(\bm x)$. For the deterministic initial condition we have $\rho_0(\bm x)=\prod_{i=1}^d\delta(x_i-x_i(0))$. 
The solution of \eqref{eqn:sde} with different initial conditions 
form a $d$-dimensional stochastic flow
on the manifold $M$, which is known as 
Brownian flow \cite{kunita1997stochastic}.
The Markov property of \eqref{eqn:sde} allows us to define a composition operator $\M(t,0)$ that pushes forward in time the average of the observable $\bm u(t)=\bm u(\bm x(t))$ over the noise, i.e., 
\begin{align}
\mathbb{E}_{\bm \W(t)}[\bm u(\bm x(t))|\bm x_0]= 
\M(t,0)\bm u(\bm x_0)=e^{t\K}\bm u(\bm x_0).
\label{MarkovSemi}
\end{align}
Using It\^o's interpretation of the white noise, $\M(t,0)$ is known as a Markovian semigroup generated by the backward Kolmogorov operator
\cite{Risken,kloeden2013numerical}:
\begin{align}
\K(\bm x_0)&=\sum_{k=1}^dF_k(\bm x_0)\frac{\partial}{\partial x_{0k}}
+\frac{1}{2}\sum_{j=1}^m\sum_{i,k=1}^d\sigma_{ij}(\bm x_0)
\sigma_{kj}(\bm x_0)\frac{\partial^2}{\partial x_{0i}\partial x_{0k}}.
\label{KI}
\end{align}
From now on, with a slight abuse of notation, we will use $\bm u(\bm x(t))$ to represent its noise-averaged quantity \eqref{MarkovSemi}. In fact, without any further specification, all $\bm u(\bm x(t))$ that appear in the rest of this paper represent the noise-averaged quantity \eqref{MarkovSemi}. We also note \eqref{MarkovSemi} is only a conditional expectation. When the initial condition is random, \eqref{MarkovSemi} is still a {\em stochastic quantity} with initial randomness. We now introduce a projection operator $\P$ and its orthogonal operator $\Q=\I-\P$. By differentiating the well-known Dyson's identity
\begin{align*}
    e^{t\K}=e^{t\Q\K}+\int_0^te^{s\K}\P\K e^{(t-s)\Q\K}\Q\K ds,
\end{align*}
it is straightforward to obtain the following 
MZ equation governing 
the evolution of the noise-averaged observable 
\eqref{MarkovSemi}:
\begin{equation}
\frac{\partial}{\partial t}e^{t\K}\bm u(0)
=e^{t\K}\P\K\bm u(0)
+e^{t\Q\K}\Q\K\bm u(0)+\int_0^te^{s\K}\P\K
e^{(t-s)\Q\K}\Q\K\bm u(0)ds.\label{eqn:EMZ_full}
\end{equation}
The three terms at the right hand side of \eqref{eqn:EMZ_full} 
are called, respectively, streaming term, fluctuation 
(or noise) term, and memory term.  Applying the projection operator 
$\P$ to \eqref{eqn:EMZ_full}
yields the projected MZ equation
\begin{equation} 
\frac{\partial}{\partial t}\mathcal{P}e^{t\K}\bm u(0)
=\mathcal{P}e^{t\K}\mathcal{PK}\bm u(0)
+\int_0^t\P e^{s\K}\mathcal{PK}
e^{(t-s)\Q\K}\mathcal{QK}\bm u(0)ds.\label{eqn:EMZ_projected}
\end{equation}
Note that the MZ equation \eqref{eqn:EMZ_full} and its projected form \eqref{eqn:EMZ_projected} for stochastic systems 
have the same structure as the classical MZ 
equations for deterministic (autonomous) systems 
\cite{zhu2019generalized,zhu2018estimation,zhu2018faber}. 
However, the Liouville operator $\L$ is now replaced by 
the Kolmogorov operator $\K$. 
Let us consider the weighted Hilbert space 
$H=L^2(M,\rho)$ with inner product
\begin{equation}
\langle h,g\rangle_{\rho}=\int_{M} 
h(\bm x)g(\bm x)\rho(\bm x)d\bm x,
\qquad h,g\in H,
\label{ip}
\end{equation}
where $\rho$ is a positive weight 
function on $M$ which is often chosen to be a certain type of
probability densities. 
With this Hilbert space, we can introduce the Mori-type projection operator:
\begin{align}
\label{Mori_P}
\P h=\sum_{i,j=1}^N G^{-1}_{ij}
\langle u_i(0),h\rangle_{\rho}u_j(0),
\qquad h\in H,
\end{align}
where $G_{ij}=\langle u_i(0),u_j(0)\rangle_{\rho}$
and $u_i(0)=u_i(\bm x(0))$ ($i=1,...,N$) are
$N$ linearly independent functions. It is easy to check that $\P$ and its orthogonal $\Q=\I-\P$ are both symmetric projection operators in $L^2(M,\rho)$, i.e. $\P^*=\P=\P^2$, $\Q^*=\Q=\Q^2$, where $\P^*,\Q^*$ are the $L^2(M,\rho)$-adjoint operator of $\P,\Q$. With $\P$ defined as \eqref{Mori_P}, the memory integral of 
the MZ equation \eqref{eqn:EMZ_full} and its 
projected version \eqref{eqn:EMZ_projected} can be simplified to a convolution term. Therefore these two MZ equations can be rewritten as
\begin{align}
\frac{d\bm u(t)}{dt} &= \bm \Omega \bm u(t) +
\int_{0}^{t}\bm K(t-s)\bm u(s)ds+\bm f(t),\label{gle_full}\\
\frac{d}{dt}\P{\bm u}(t) &= \bm \Omega\P {\bm u}(t)+ 
\int_{0}^{t} \bm K(t-s) \P {\bm u}(s)\,ds,\label{gle_projected}
\end{align}
where $\bm u(t) =[u_1(t),\dots,u_N(t)]^T$ and 
\begin{subequations}
\begin{align}
		G_{ij} & = \langle u_{i}(0), u_{j}(0)\rangle_{\rho}
		\quad \text{(Gram matrix)},\label{gram}\\
		\Omega_{ij} &= \sum_{k=1}^N G^{-1}_{jk}
		\langle u_{k}(0), \K u_{i}(0)\rangle_{{\rho}}\quad 
		\text{(streaming matrix)},\label{streaming}\\
		K_{ij}(t) & =\sum_{k=1}^N G^{-1}_{jk}
		\langle u_{k}(0), \K e^{t\Q\K}\Q\K u_{i}(0)\rangle_{\rho}\quad 
		\text{(memory kernel)},\label{SFD}\\
		 f_i(t)& =e^{t\Q\K}\Q\K u_i(0) \quad 
		\text{(fluctuation term)}.\label{f}
	\end{align}
\end{subequations}
\noindent
Equations \eqref{gle_full}-\eqref{gle_projected} are known as the (linear) generalized Langevin equation (GLE) in statistical physics. The projection operator method provides a systematic way to derive such {\em closed} equations of motion from the first principle. Depending on the choice of the Hilbert space weight function $\rho$, the linear GLE \eqref{gle_full} and its projected form \eqref{gle_projected} yield evolution equations for different dynamical quantities. When considering SDE \eqref{eqn:sde} in the context of statistical physics, the most common setting of $\rho$ is $\rho=\rho_0=\rho_{S}$, where $\rho_S$ is the steady state distribution of the stochastic system satisfying the Fokker-Planck equation $\partial_t\rho_{S}=\K^*\rho_{S}=0$ (see details in Section \ref{sec:FDT}). For such a case, GLE \eqref{gle_full} yields the full dynamics of the noise-averaged quantity $\mathbb{E}_{\bm \W(t)}[\bm u(\bm x(t))|\bm x_0]$ and the projected GLE \eqref{gle_projected} yields the evolution equation of the steady state time-autocorrelation function $C_{ij}(t)$, which is defined as \cite{pavliotis2014stochastic,zhu2020hypoellipticity} 
\begin{align*}
C_{ij}(t):=\mathbb{E}_{\bm x_0}[\mathbb{E}_{\bm \W(t)}[u_i(t)u_j(0)|\bm x_0]]=\langle\M(t,0)u_i(0),u_j(0)\rangle_{\rho_0}
=\langle\M(t,0)u_i(0),u_j(0)\rangle_{\rho_S}.
\end{align*}
For deterministic Hamiltonian systems with the unitary evolution operator $e^{t\L}$, the Mori-Zwanizg equation is as \eqref{eqn:EMZ_full} but with the Kolmogorov operator $\K$ replaced by the Liouville operator $\L$. As we mentioned in the introduction, the second FDT the the thermal equilibirum holds naturally as a result of the skew-symmetry of the Liouville operator $\L$ with respect to $\langle\cdot,\cdot\rangle_{\rho_{eq}}$, where $\rho_{eq}$ is the equilibrium distribution \cite{snook2006langevin,zhu2018estimation}. In fact we have 
\begin{equation}
\begin{aligned}
		K_{ij}(t) =\sum_{k=1}^N G^{-1}_{jk}
		\langle u_{k}(0), \L e^{t\Q\L}\Q\L u_{i}(0)\rangle_{\rho_{eq}}
		&=-\sum_{k=1}^N G^{-1}_{jk}
		\langle \L u_{k}(0), e^{t\Q\L}\Q\L u_{i}(0)\rangle_{\rho_{eq}}\\
		&=-\sum_{k=1}^N G^{-1}_{jk}
		\langle \Q\L u_{k}(0), e^{t\Q\L}\Q\L u_{i}(0)\rangle_{\rho_{eq}}\\
	   &=-\sum_{k=1}^N G^{-1}_{jk}
		\langle f_k(0), f_{i}(t)\rangle_{\rho_{eq}},
	\end{aligned}
\end{equation}
where we used the idempotence of the projection operator, i.e. $\Q^2=\Q$ and the fact that Mori-type projection operator $\Q$ is symmetric in $L^2(M,\rho_{eq})$. For stochastic systems studied in this paper, since the Kolmogorov operator $\K$ is not skew-symmetric with respect to the inner product $\langle\cdot,\cdot\rangle_{\rho}$, the classical second FDT has to be generalized accordingly.
\section{Generalized fluctuation-dissipation theorem}\label{sec:FDT}
In this section, we first review the derivation of the generalized first FDT for stochastic systems and emphasize its difference from the second FDT. As a mathematical preparation, we need to assume that there exists a unique probability measure $d\mu=\rho d\bm x$ which is invariant under the Markovian semigroup generated by SDE \eqref{eqn:sde}. This implies the existence and uniqueness of a time-independent probability density $\rho$ that satisfies the Kolmogorov forward (Fokker-Planck) equation:
\begin{align}\label{Fokker-Planck}
\partial_t\rho=\K^*\rho=0,
\end{align}
where $\K^*$ is the $L^2(M)$-adjoint operator of $\K$. Throughout the paper, we further assume that $\rho$ is a smooth function which decays to $0$ when approaching boundary $\partial M$. For frequently used statistical physics models, these two assumptions are generally hard to prove since it involves the analysis of the degenerate elliptic operator $\K$. Some recent studies on hypoellipticity \cite{cuneo2018non,eckmann2000non} have shown the possibility to obtain affirmative answers using the rather complicated H\"ormander analysis. These theoretical results, together with numerical studies such as \cite{lepri2003thermal,lepri1998anomalous}, suggest that the uniqueness, smoothness and the decaying properties of $\rho$ at $\partial M$ are rather technical assumptions.

\subsection{Derivation of the first FDT}
The first fluctuation–dissipation theorem describes the linear response of a given dynamical system to external perturbations. This relationship was established by Kubo \cite{kubo1966fluctuation} for classical and quantum Hamiltonian systems and then has been extended to open systems by different researchers \cite{agarwal1972fluctuation,hairer2010simple,dal2019linear,gritsun2008climate}. As we briefly mentioned in the introduction, the derivation of the first FDT relies on the perturbation theory. To demonstrate this point, we consider a general perturbation to the stochastic system \eqref{eqn:sde}:
\begin{align}\label{eqn:sde_p}
d\bm x(t)=\bm F(\bm x(t))dt+\bm \sigma(\bm x(t))d\bm\W(t)+\delta\bm G(t)\cdot\tilde{\bm F}(\bm x(t))dt+\sqrt{\delta}\bm H(t)\cdot\tilde{\bm \sigma}(\bm x(t))d\bm \W(t), \qquad \bm x(0)\sim \rho,
\end{align}
where $0<\delta\ll 1$, $\delta\bm G(t)\cdot \tilde{\bm F}(\bm x(t)) dt$ and $\sqrt{\delta}\bm H(t)\cdot\tilde{\bm \sigma}(\bm x(t))d\bm \W(t)$ characterize respectively the dynamics of the pertubative frictional and fluctuating forces added to the stochastic system. We denote $\rho$ as the steady state distribution for the unperturbed stochastic system \eqref{eqn:sde} and $\rho_{\delta}$ as the one for the perturbed stochastic system \eqref{eqn:sde_p}. For the sake of simplicity, we assume that $\delta \bm G(t)$, $\sqrt{\delta} \bm H(t)$, the diffusion matrix $\bm \sigma(\bm x(t))$ and the perturbed diffusion matrix $\tilde{\bm\sigma}(\bm x(t))$ are all $d\times d$ diagonal matrices. As a consequence, the diffusion part (second-order derivatives) of the Kolmogorov backward operator corresponding to the perturbed system is of the diagonal form and we have
\begin{align*}
\K=\sum_{i=1}^d F_i(\bm x)\partial_{x_i}+\sigma_i^2(\bm x)\partial_{x_i}^2+\delta G_i(t)\tilde{F_i}(\bm x)\partial_{x_i}+\delta H^2_i(t)\tilde{\sigma}_i^2(\bm x)\partial_{x_i}^2,
\end{align*}
where $\sigma_i(\bm x), \tilde{\sigma}_i(\bm x)$, $\delta G_i(t)$ and $\delta H_i^2(t)$ are the $i$-th diagonal elements of the corresponding matrices. Using the standard perturbation theory, we obtain the following generalized first FDT: 
\begin{theorem}\label{thm1:1st-FDT}(Generalized-1st-FDT) Assuming that the steady state distribution $\rho=\rho(\bm x)$ is a smooth function of $\bm x$ and decays to 0 when approaching the boundary $\partial M$. For the perturbed SDE \eqref{eqn:sde_p}, the following generalized first FDT holds for state space observable $u=u(\bm x)$:
\begin{equation}\label{Generalized_1st_FDT}
\begin{aligned}
\la u(t)\rangle_{\rho_{\delta}}-\la u(t)\ra=-\sum_{i=1}^d\int_0^{t}&\left\la\frac{1}{\rho} [\partial_{x_i}\tilde{F}_i\rho]u(t-s)\right\ra\delta G_i(s)ds\\
&+\sum_{i=1}^d\int_0^t\left\la\left(\partial_{x_i}^2\tilde{\sigma}_i+\frac{2}{\rho}\partial_{x_i}\tilde{\sigma}_i\partial_{x_i}\rho+\frac{1}{\rho}\tilde{\sigma}_i\partial_{x_i}^2\rho\right)u(t-s)\right\ra\delta H_i^2(s)ds+O(\delta^2).
\end{aligned}
\end{equation}
If the perturbative forces $\delta \bm G(t)$ and $\sqrt{\delta} \bm H(t)$ are homogeneous, i.e. $\delta G_i(t)=\delta G(t)$ and $\sqrt{\delta} H_i(t)=\sqrt{\delta} H(t)$, then the generalized first FDT \eqref{Generalized_1st_FDT} has a vector-form representation:
\begin{equation}\label{Generalized_1st_FDT2}
\begin{aligned}
\la u(t)\rangle_{\rho_{\delta}}-\la u(t)\ra=-\int_0^{t}&\left\la\frac{1}{\rho} \nabla\cdot[\tilde{\bm F}\rho]u(t-s)\right\ra\delta G(s)ds\\
+
&\int_0^t\left\la\left(\nabla\cdot\overrightarrow{L_{\tilde{\bm \sigma}}}
+\frac{2}{\rho}\overrightarrow{L_{\tilde{\bm \sigma}}}\cdot\nabla\rho
+\frac{1}{\rho}\D_{\tilde{\bm \sigma}}\rho\right)u(t-s)\right\ra\delta H^2(s)ds+O(\delta^2),
\end{aligned}
\end{equation}
where the vector function $\overrightarrow{L_{\tilde{\bm \sigma}}}=[\partial_{x_i}\tilde{\sigma_i}]_{i=1}^d$ and $\D_{\tilde{\bm \sigma}}$ is a second-order differential operator defined as $\D_{\tilde{\bm \sigma}}=\sum_{i=1}^d\tilde{\sigma_i}\partial_{x_i}^2$. 
\end{theorem}
The proof of Theorem \ref{thm1:1st-FDT} follows a perturbation analysis similar to the one used in \cite{kubo1966fluctuation} or more recent papers such as \cite{gritsun2008climate,dal2019linear}. We also outline the procedure in \ref{sec:APP1-FDT_proof}. In the absence of the stochastic perturbation, i.e. $\sqrt{\delta} \bm H(t)=0$, for homogeneous perturbation $\delta G_i(t)=\delta G(t)$, we get the commonly used linear-response relation \cite{gritsun2008climate,dal2019linear}:
\begin{align}\label{Generalized_1st_FDT3}
\la u(t)\rangle_{\rho_{\delta}}=-\int_0^{t}\left\la\frac{1}{\rho} \nabla\cdot[\tilde{\bm F}\rho]u(t-s)\right\ra\delta G(s)ds.
\end{align}
We note that to apply formulas \eqref{Generalized_1st_FDT}-\eqref{Generalized_1st_FDT3}, it is required to know the steady state distribution $\rho$. For nonequilibrium systems, this is generally hard to obtain due to high dimensionality of the Fokker-Planck equation \eqref{Fokker-Planck}. Various approaches such as the Gaussian approximation \cite{gritsun2008climate} and the information theory method \cite{majda2005information} were proposed to address this issue. As we will see in the following section, the generalized second FDT also contains additional terms that involve $\rho$. However, for specifically chosen observables $u(\bm x)$, calculation of the density $\rho$ can be avoided.  
\subsection{Derivation of the second FDT}
The second flutucation-dissipation theorem gives the proportionality between the noise amplitude and the friction kernel in the GLE. For open systems, depending on the form of the GLE and the observable of interest, the explicit expression of the second FDT will be different. A well-organized review in this regard is given by Maes in \cite{maes2014second}. In this section, we provide a novel way to generalize the second FDT using the first-principle GLE derived in Section \ref{sec:EMZE}. The derivation only uses the mathematical properties of the Kolmogorov operator $\K$. Without loss of generality, we consider the SDE \eqref{eqn:sde} with the Kolmogorov (backward) operator \eqref{KI} of the form:
\begin{align}\label{K_ini}
\K=\L+\D=\sum_{i=1}^NF_i(\bm x)\partial_{x_i}+\sum_{i,j=1}^N\sigma_i\sigma_j\partial^2_{x_ix_j},
\end{align}
where $\sigma_i, \sigma_j$ are constants. Note that here $\L$ represents a general advection operator instead of the skew-symmetric Liouville operator. For such a stochastic system in the steady state $\rho$, by introducing the Mori-type projection operator \eqref{Mori_P} in weighted Hilbert space $L^2(M,\rho)$, we can prove the following generalized second FDT for the GLE \eqref{gle_full}-\eqref{gle_projected}. This is the main theoretical result of this paper. 
\begin{theorem}\label{thm:2nd-FDT}
(Generalized-2nd-FDT) For SDE with infinitesimal generator \eqref{K_ini}, if we assume that the steady state distribution $\rho$ is a smooth function of $\bm x$ and decays to $0$ when approaching the boundary $\partial M$, then the following generalized second FDT holds for any state space observable $u=u(\bm x)$ in GLE \eqref{gle_full}-\eqref{gle_projected}:
\begin{align}\label{general_2nd_FDT}
K(t)=-\frac{\langle f(0), f(t)\rangle_{\rho}}{\langle u^2(0)\rangle_{\rho}}+\frac{\langle w(0), f(t)\rangle_{\rho}}{\langle u^2(0)\rangle_{\rho}},
\end{align}
where $w(0)=\S u(0)$ and $\S$ is a second-order differential operator  defined as 
\begin{align}\label{W_exact_form}
\S=2\sum_{i,j=1}^N\sigma_i\sigma_j\partial^2_{x_ix_j}+\frac{1}{\rho}\sum_{i,j=1}^N\sigma_i\sigma_j\partial_{x_i}\rho\partial_{x_j}+\sigma_i\sigma_j\partial_{x_j}\rho\partial_{x_i}.
\end{align}
For a special case where $\K$ is a degenerate elliptic operator and the diffusion term is of the diagonal form $\D=\sigma^2\Delta$, where $\Delta$ is a $n$-dimensional Laplacian with $n\leq N$, then $\S$ is an $n$-dimensional operator and admits a simple form:
\begin{align}\label{W_simple_form}
\S=2\sigma^2\Delta+2\sigma^2\nabla(\ln\rho)\cdot\nabla.
\end{align}
\end{theorem}
\begin{proof}
For any operator $\O$ defined in weighted Hilbert space $L^2(M,\rho)$, where $\rho$ is the steady state distribution, we denote $\O^*_{\rho}$ as the adjoint operator of $\O$ in $L^2(M,\rho)$ with respect to inner product $\langle\cdot,\cdot\rangle_{\rho}$. Correspondingly, we denote $\O^*$ as the adjoint operator of the flat Hilbert space $\O$ in $L^2(M)$ with respect to the weight-less inner product $\langle\cdot,\cdot\rangle$. First we aim to find the $L^2(M,\rho)$ adjoint of the Kolmogorov operator $\K$ \eqref{K_ini}. Using the integration by parts formula and the fact that $\rho(\bm x)=0$ at $\partial M$, it is easy to obtain
\begin{align}
\L+\L^*_{\rho}&=-\nabla\cdot \bm F(\bm x)- \frac{1}{\rho}\L\rho. \label{formula:L+L*}
\end{align}
Similarly, for the second-order differential operator $\D$ in \eqref{K_ini}, using the integration by parts formula twice leads to
\begin{equation}\label{formula:D*}
\begin{aligned}
\D^*_{\rho}=\D+\frac{1}{\rho}\sum_{i,j=1}^N\sigma_i\sigma_j\partial_{x_i}\rho\partial_{x_j}+\sigma_i\sigma_j\partial_{x_j}\rho\partial_{x_i}+\frac{1}{\rho}\D\rho.
\end{aligned}
\end{equation}
Naturally we have 
\begin{equation}\label{formula:D*+D_general}
\begin{aligned}
\D+\D^*_{\rho}=2\D+\frac{1}{\rho}\sum_{i,j=1}^N\sigma_i\sigma_j\partial_{x_i}\rho\partial_{x_j}+\sigma_i\sigma_j\partial_{x_j}\rho\partial_{x_i}+\frac{1}{\rho}\D\rho=\S+\frac{1}{\rho}\D\rho,
\end{aligned}
\end{equation}
where operator $\S$ is defined as \eqref{W_exact_form}. For the special case where $\D=\sigma^2\Delta$, it is easy to see that $\S$ has the simple form \eqref{W_simple_form}. Summing up formulas \eqref{formula:L+L*} and \eqref{formula:D*+D_general}, we can get that 
\begin{align}\label{K_sum}
\K_{\rho}^*+\K=-\nabla\cdot\bm F(\bm x)-\frac{1}{\rho}\L\rho+\frac{1}{\rho}\D\rho+\S.
\end{align}
Using a similar procedure for operator $\L$ and $\D$ defined in the flat Hilbert space $L^2(M)$, we have
\begin{align}\label{formula:LD_flat}
\L+\L^*=-\nabla\cdot \bm F(\bm x), \qquad \D^*=\D.
\end{align}
On the other hand, since $\rho$ satisfies the steady state Fokker-Planck equation $\partial_t\rho=\K^*\rho=0$, we obtain an operator identity $\partial_t\rho=\K^*\rho=\L^*\rho+\D^*\rho=0$. Combining this with formula \eqref{formula:LD_flat} and noting that $-\nabla\cdot\bm F(\bm x)$ is a multiplication operator, we obtain 
\begin{equation}
\left.\begin{aligned}
\L\rho+\L^*\rho&=-\rho\nabla\cdot\bm F(\bm x)\\
\D^*\rho&=\D\rho
\end{aligned}\right\}
\Rightarrow
-\frac{1}{\rho}\L\rho+\frac{1}{\rho}\D\rho-\nabla\cdot\bm F(\bm x)=0.
\end{equation}
Substituting this relation into \eqref{K_sum} we get $\K^*_{\rho}=-\K+\S$. Now we note that symmetric projection operator $\Q$ satisfies $\Q^*=\Q=\Q^2$. For the formally defined GLE memory kernel \eqref{SFD}, we obtain the generalized second FDT:
\begin{align*}
\langle u^2(0)\rangle_{\rho}K(t)=\la u(0),\K e^{t\Q\K}u(0)\ra&=\la \K^*_{\rho}u(0),e^{t\Q\K}\Q\K u(0)\ra \\
&=-\la \Q\K u(0), e^{t\Q\K}\Q\K u(0)\ra+\la w(0),e^{t\Q\K}\Q\K u(0)\ra\\
&=-\la f(0), f(t)\ra+\la w(0), f(t)\ra.\
\end{align*}
\end{proof}
Theorem \ref{thm:2nd-FDT} can be readily generalized to the $N$-dimensional GLE \eqref{gle_full}-\eqref{gle_projected} and for the general Kolmogorov operator \eqref{KI}. When the dissipative term in the Kolmogorov operator \eqref{K_ini} is given by $\D=\sum_{i,j=1}^N\sigma_i(\bm x)\sigma_j(\bm x)\partial^2_{x_i,x_j}$, then we obtain $w(0)=\S u(0)$ with
\begin{align}\label{W_sigma}  \S=\sum_{i,j=1}^N2\sigma_i\sigma_j\partial^2_{x_ix_j}+\frac{1}{\rho}\sigma_i\sigma_j\partial_{x_i}\rho\partial_{x_j}+\frac{1}{\rho}\sigma_i\sigma_j\partial_{x_j}\rho\partial_{x_i}+\partial_{x_i}[\sigma_i\sigma_j]\partial_{x_j}+\partial_{x_j}[\sigma_i\sigma_j]\partial_{x_i},
\end{align}
where the shorthand notation $\sigma_i(\bm x)=\sigma_i$ and $\sigma_j(\bm x)=\sigma_j$ are used. The derivation of \eqref{W_sigma} follows directly from the proof of Theorem \ref{thm:2nd-FDT} and is given in \ref{sec:APP-2FDT_proof}. On the other hand, the result of Theorem \ref{thm:2nd-FDT} holds for an arbitrary Mori-type projection operator $\P$. Hence for the $N$-dimensional GLEs \eqref{gle_full}-\eqref{gle_projected}, we have 
\begin{align}\label{KK}
K_{ij}(t)=\sum_{k=1}^N G^{-1}_{jk}(-\langle f_{k}(0), f_i(t)\rangle_{\rho}
+\langle \S u_{k}(0), f_i(t)\rangle_{\rho}).
\end{align}
When one tries to apply the first and the second FDTs such as \eqref{Generalized_1st_FDT3} and \eqref{general_2nd_FDT}, it is normally necessary to know the exact form of the steady state distribution $\rho$, which, however, is generally hard to deduce for nonequilibrium systems such as turbulence \cite{majda2005information} and heat conduction models \cite{lepri1998anomalous}. This difficulty can be bypassed  because of the following corollary which is already announced in Section \ref{sec:short_sum}:
\begin{corollary}\label{corollary_2ndFDT0}
If the phase space observable $u=u(\bm p(t), \bm q(t))$ is only a function of the degenerate coordinates of the stochastic system, then $w(0)=\S u(0)=0$ and the GLE \eqref{gle_full}-\eqref{gle_projected} for $u(t)$ satisfies the classical second FDT. 
\end{corollary}
The proof is obvious since $\S$ is an operator in the {\em non-degenerate} coordinate. A typical example is the Langevin dynamics for a molecular dynamical system. Since the Gaussian white noise is only imposed in the momentum coordinate $p_i$, the position coordinate $q_i$ is therefore {\em degenerate}. Hence we have $w(0)=\S u(0)=\S f(q_i(0))=0$. 
For more general, nonequilibrium systems in the steady state, this result still holds which is somewhat surprising because the nonequilibrium steady state (NESS) measure is generally unknown! As we will see in the following section, many statistical physics models are generated by highly degenerate elliptic operator $\K$, i.e. the dissipative forces act on a small subset of the phase space coordinates. For such models, if the observable $u(\bm x)$ of interest is a function of the degenerate coordinate, we can simply avoid the calculation of $\rho$ in evaluating $w(0)$ and use the classical second FDT to build reduced-order models for the observable. An application of this fact in heat conduction problems is presented in Section \ref{sec:Numerical}. We also have the following result:
\begin{corollary}\label{corollary_2ndFDT}
For SDE with infinitesimal generator \eqref{KI}, if the linear GLEs \eqref{gle_full}-\eqref{gle_projected} for observable $\bm u(t)$ satisfying the classical second FDT with $\bm w(0)=0$, then $\Omega_{ij}=0$ in \eqref{gle_full}-\eqref{gle_projected}.
\end{corollary}
\begin{proof}
It is sufficient to prove the one-dimensional case. According to the definition of $\Omega_{ij}$ \eqref{streaming}, for one-dimensional GLE \eqref{gle_full}-\eqref{gle_projected}, we have 
\begin{align*}
\Omega=\frac{\langle \K u(0),u(0)\rangle_{\rho}}{\langle u^2(0)\rangle_{\rho}}=\frac{\langle u(0),\K_{\rho}^*u(0)\rangle_{\rho}}{\langle u^2(0)\rangle_{\rho}}=-\frac{\langle u(0),\K u(0)\rangle_{\rho}}{\langle u^2(0)\rangle_{\rho}}+\frac{\langle u(0),w(0)\rangle_{\rho}}{\langle u^2(0)\rangle_{\rho}}=-\frac{\langle u(0),\K u(0)\rangle_{\rho}}{\langle u^2(0)\rangle_{\rho}}=0.
\end{align*}
\end{proof}
\paragraph{Remark} When compared with previous methods, e.g. \cite{maes2014second}, in deriving the (generalized) second FDT for stochastic systems, our approach does not rely on weak perturbation assumptions therefore is generally applicable to arbitrary stochastic systems driven by white noise. Moreover, the derivation only uses the properties of the Kolmogorov operator without applying any {\em ad hoc} approximations. On the other hand, although the presented example is based on the Mori-type projection operator which leads to linear GLEs, the above theory also applies to nonlinear GLEs such as the Zwanzig's equation \cite{zwanzig1973nonlinear}, essentially because the Zwanzig-type projection operator is an infinite-rank operator which is similar to the Mori-type projection operator. The proof is rather technical therefore will be deferred to \ref{sec:APP_Zwanzig}.

\section{The generalized second FDT for specific systems} \label{sec:Application_to_specifics}
\subsection{The generalized second FDT for equilibrium systems}
The application of the generalized FDTs to equilibrium systems leads to explicit expressions of formula \eqref{Generalized_1st_FDT}, \eqref{Generalized_1st_FDT2} and \eqref{general_2nd_FDT} since the equilibrium distribution $\rho$ is given by the Gibbs-Boltzmann form $\rho=e^{-\beta\H}/Z$ for the canonical ensemble. In this section, we will derive such expression for some frequently used statistical mechanics models. Before we move onto analyzing stochastic systems, it is worth noticing that for equilibrium systems generated by deterministic forces such as the Nos\'e-Hoover thermostats, the classical second FDT holds as a result of the skew-symmetry of the Liouville operator $\L$.
\paragraph{Langevin dynamics}
The Langevin dynamics for a $d$-dimensional system of $N$ interacting particles is given by the following SDE in $\R^{2d\times N}$: 
\begin{align}\label{eqn:LE}
\begin{dcases}
d\bm q_i=\frac{1}{m_i}\bm p_idt\\
d\bm p_i=\sum_{i\neq j}^N\bm F_{i,j}^Cdt-\frac{\gamma}{m_i}\bm p_idt+\sigma d\bm \W_{i}(t)
\end{dcases},
\end{align}
where $m_i$ is the mass of each particle, 
$\sum_{i\neq j}^N\bm F_{i,j}^C$ is the total conservative force acting on particle $i$, and 
$\bm \W_{i}(t)$ is a $d$-dimensional Wiener process which satisfies $d\bm \W_j(t)=\bm \xi_i(t)dt$ with
\begin{align*}
\langle \xi_{ij}(t)\rangle=0,\qquad 
\langle \xi_{ij}(t)\xi_{i'j'}(s)\rangle=2\gamma k_{B}T\delta_{ii'}\delta_{jj'}\delta(t-s).
\end{align*}   
The parameter $\gamma$ is the friction coefficient which is related to $\sigma$ through the fluctuation-dissipation relation 
$\sigma=(2\gamma/\beta)^{1/2}$, where $\beta=1/k_BT$, $k_B$ is the Boltzmann constant and $T$ the temperature of the equilibrium system. 
The stochastic dynamical system 
\eqref{eqn:LE} is widely used in the mesoscopic modelling of liquids and gases. The Kolmogorov backward operator \eqref{KI} associated with 
the SDE \eqref{eqn:LE} is given by
\begin{align}\label{Kolmo:LE}
\K=\sum_{i=1}^N\frac{\bm p_i}{m_i}\cdot\partial_{\bm q_i}+
\sum_{i,j\neq i}^N\bm F_{i,j}^C\cdot\partial_{\bm p_i}-\sum_{i}^N
\frac{\gamma\bm p_i}{m_i}\cdot\partial_{\bm p_i}+\sum_i^N\frac{\gamma}{\beta}\partial_{\bm p_i}\cdot\partial_{\bm p_i},
\end{align}
where ``$\cdot$'' denotes the standard dot product. 
If the interaction potential $V(\bm q)$ is strictly positive 
at infinity then the Langevin equation \eqref{eqn:LE} 
admits a unique invariant Gibbs measure 
given by
\begin{equation}
\rho_{eq}(\bm p,\bm q)=\frac{1}{Z} e^{-\beta \H(\bm p,\bm q)},
\end{equation}
where 
\begin{equation}
\H(\bm p,\bm q)=\sum_{i=1}^N\frac{\|\bm p_i\|_2^2}{2m_i}+V(\bm q)
\end{equation}
is the Hamiltonian and $Z$ is the partition function. The formal expression of the Gibbs-Boltzmann distribution enables us to get the explicit expression of the additional term $w(0)$ in the generalized second FDT \eqref{general_2nd_FDT}. As an example, the Langevin dynamics \eqref{eqn:LE} is often used to study the self-diffusion of Brownian particles, for which a relevant physical observable is the tagged particle velocity $\bm v_j$. By choosing $u(0)=p_{jx}/m_j$ and using \eqref{general_2nd_FDT} and \eqref{Kolmo:LE}, we obtain
\begin{align}\label{w0_Langevin}
w(0)=\sum_{i=1}^N\frac{\gamma}{\beta}\partial_{\bm p_i}\cdot\partial_{\bm p_j}\frac{p_{jx}}{m_j}
+\sum_{i=1}^N\frac{\sigma^2}{\rho_{eq}}\partial_{\bm p_{i}}\rho_{eq}\cdot\partial_{\bm p_{i}}\frac{p_{jx}}{m_j}=-\frac{2\gamma}{m_j^2}p_{jx}.
\end{align}
This implies that the GLE \eqref{gle_full} for observable $p_{jx}$ can be rewritten as
\begin{align*}
    \frac{d}{dt}p_{jx}(t)=\Omega p_{jx}(t)-\int_0^t\left\langle f(0)+\frac{2\gamma}{m_j^2}p_{jx}(0),f(t-s)\right\rangle_{\rho_{eq}}p_{jx}(s)ds+f(t).
\end{align*}
\paragraph{Remark} Here we obtained a somewhat counterintuitive conclusion stating that the {\em classical} second FDT for a Brownian particle does not hold in the statistical equilibrium since we have an additional term \eqref{w0_Langevin} in the GLE memory kernel. This is because instead of using Hamiltonian dynamics, we used {\em stochastic} dynamics, i.e. Eqn \eqref{eqn:LE}, to simulate the equilibrium. On the other hand, the form of the second FDT strongly depends on the GLE under investigation. For other GLEs such as the nonlinear ones derived by Zwanzig-type projection operators \cite{lei2016data,Li2015,hudson2018coarse}, it is possible that the {\em classical} second FDT still holds for a Brownian particle generated by the Langevin dynamics.
\paragraph{Dissipative particle dynamics}
For a $d$-dimensional interacting particle system of $N$ particles, the SDE that governs the particle position $\bm q_i$ and momentum $\bm p_i$ in dissipative particle dynamics (DPD) is given by \cite{Hoogerbrugge_SMH_1992,espanol1995statistical}
\begin{equation}\label{eqn:DPD}
\begin{aligned}
\begin{dcases}
d\bm q_i&=\frac{\bm p_i}{m_i}dt\\
d\bm p_i&=\sum_{i\neq j}^N\bm F_{ij}^C(\bm q_{ij})dt
-\sum_{i\neq j}^N\gamma\omega(q_{ij})(\bm e_{ij}\cdot\bm v_{ij})\bm e_{ij}dt+\sum_{i\neq j}^N\sigma\omega^{1/2}(q_{ij})\bm e_{ij}d \W_{ij}(t)
\end{dcases}
\end{aligned}
\end{equation}
where $m_i$ is the mass of $i$-th particle, $\bm q_{ij}=\bm q_i-\bm q_j$, $q_{ij}=\|\bm q_i-\bm q_j\|$, $\bm e_{ij}=\bm q_{ij}/q_{ij}$, $\bm v_{ij}=\bm v_i-\bm v_j$, $\bm v_i=\bm p_i/m_i$ and $\bm F^C_{ij}(\bm q_{ij})$ is the conservative force exerted on particle $i$ by particle $j$. The dimensionless weight function $\omega(q_{ij})$ provides the range of interactions of the dissipative and random forces. The friction coefficient and the noise intensity are linked with each other through the fluctuation-dissipation relation $\sigma=(2\gamma/\beta)^{1/2}$, where $\beta=1/k_BT$. For the DPD model, the frictional
forces are applied in a pair-wise form, such that the sum
of thermostating forces acting on a particle pair equals
zero. Hence for $d\W_{ij}(t)=\bm \xi_i(t)dt$, we have $\xi_{ij}(t)=\xi_{ji}(t)$ and 
\begin{align*}
\langle \xi_{ij}(t)\xi_{i'j'}(t) \rangle =(\delta_{ii'}\delta_{jj'}+\delta_{ij'}\delta_{ji'})\delta(t-s).
\end{align*}
The Kolmogorov backward operator associated with the DPD model \eqref{eqn:DPD} is given by \cite{espanol1995hydrodynamics}
\begin{align*}
\K=\sum_{i=1}^N\frac{\bm p_i}{m_i}\cdot\partial_{\bm q_i}
+\sum_{i,j\neq i}^N\bm F_{i,j}^C\cdot\partial_{\bm p_i}
-\sum_{i,j\neq i}^N\gamma\omega(q_{ij})(\bm e_{ij}\cdot \bm v_{ij})\bm e_{ij}\cdot\partial_{\bm p_i}
+\frac{1}{\beta}\sum_{i,j\neq i}^N\gamma\omega(q_{ij})(\partial_{\bm p_i}\cdot\partial_{\bm p_i}-\partial_{\bm p_i}\cdot\partial_{\bm p_j}).
\end{align*}
Similarly, for the $x$-directional velocity of the $h$-th particle in the DPD model, the additional term $w(0)$ in the generalized second FDT can be calculated using \eqref{W_sigma} as:
\begin{align*}
w(0)&=\S\frac{p_{hx}}{m_h}\\
&=\frac{1}{\beta}\sum_{i,j}^N\frac{\sigma_i\sigma_j\delta_{hj}\delta_{ij}}{m_h\rho_{eq}}\partial_{p_{ix}}\rho_{eq}
+\frac{\sigma_i\sigma_j\delta_{hi}\delta_{ij}}{m_h\rho_{eq}}\partial_{p_{jx}}\rho_{eq}
-\frac{1}{\beta}\sum_{i,i\neq j}^N\frac{\sigma_i\sigma_j\delta_{hj}}{m_h\rho_{eq}}\partial_{p_{ix}}\rho_{eq}
+\frac{\sigma_i\sigma_j\delta_{hi}}{m_h\rho_{eq}}\partial_{p_{jx}}\rho_{eq}\\
&=-\sum_{i,j}^N\frac{\sigma_i\sigma_j\delta_{hj}\delta_{ij}}{m_hm_i}p_{ix}+\frac{\sigma_i\sigma_j\delta_{hj}\delta_{ij}}{m_hm_j}p_{jx}
+\sum_{i,i\neq j}^N\frac{\sigma_i\sigma_j\delta_{hj}}{m_hm_i}p_{ix}
+\frac{\sigma_i\sigma_j\delta_{hi}}{m_hm_j}p_{jx},
\end{align*}
where $\sigma_i\sigma_j=\gamma\omega(q_{ij})$ and we used the fact that $\partial_{\bm p_i}[\sigma_i\sigma_j]=0$. The resulting GLEs \eqref{gle_full}-\eqref{gle_projected} can be obtained accordingly. For the Langevin dynamics and the DPD model, the corresponding Kolmogorov operators are both degenerate elliptic operators since white noise is only imposed in the momentum space. Hence if we choose the tagged particle position $\bm q_{j}$ as the quantity of interest, then the classical second FDT holds for GLEs \eqref{gle_full}-\eqref{gle_projected} as claimed in Section \ref{sec:FDT}.  
\subsection{The generalized second FDT for nonequilibrium systems}
\begin{figure}[t]
\centerline{
\includegraphics[height=3cm]{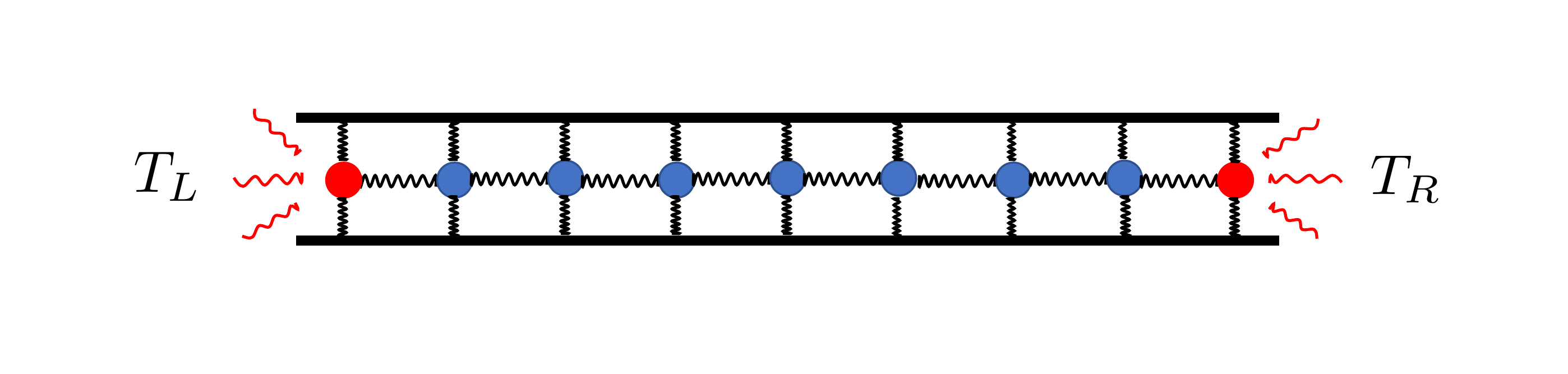}
}
\centerline{
\includegraphics[height=4.5cm]{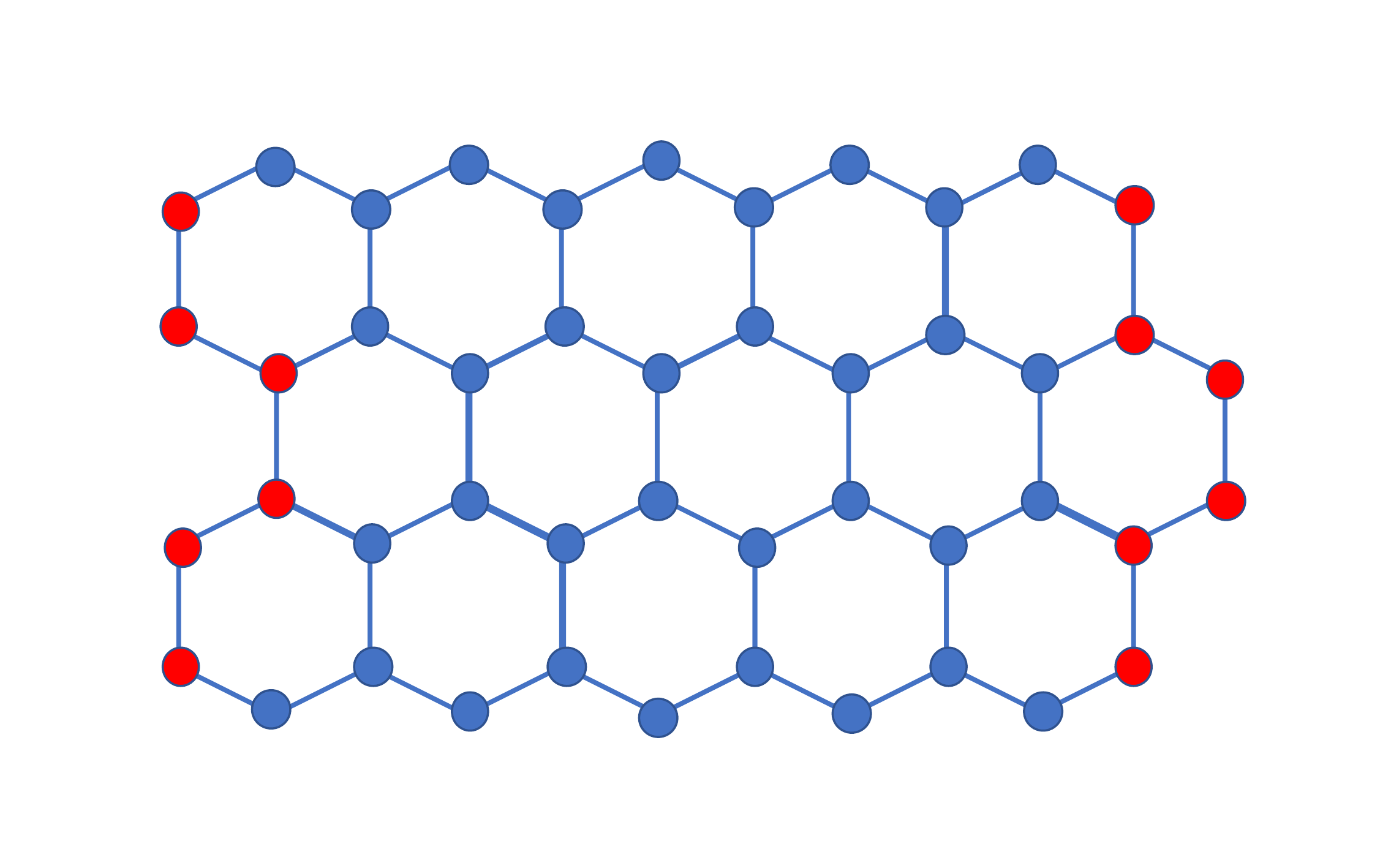}
}
\caption{Schematic of the $1$-D and $2$-D heat transport model \eqref{SDE:n_d_heat}. The onsite potential and boundary Langevin forces are shown in the top plot for the $1$-D model. Note that only the boundary oscillators, marked in red, interact with the heat bath.}
\label{fig:Heat_transport} 
\end{figure}
As an example of nonequilibrium systems, we consider an $d$-dimensional heat conduction model in \cite{cuneo2018non,lepri1998anomalous}. To this end, we consider a lattice $\G$ of interacting oscillators. For each oscillator $i\in \G$, the position and momentum are given respectively by $\bm q_i\in \R^d$ and $\bm p_i\in\R^d$. The phase space is therefore given by $\Omega=\R^{2|\G|d}$, where $|\G|$ gives the cardinality of the set $\G$ which corresponds to the total number of particles. These oscillators are interacting with the substrate and each other through the Hamiltonian 
\begin{align}\label{Hamiltonian_n}
\H(\bm p,\bm q)=\sum_{i\in \G}\left(\frac{\bm p_i^2}{2m_i}+U_{i}(\bm q_i)\right)+\sum_{e\in\E}V_{e}(\delta \bm q_e),
\end{align}
where $U_{i}$ and $V_e$ are the pining potential and interactive potential, respectivly, For $e=(i,i')\in \E=\G\times\G$ we have $\delta \bm q_{e}=\bm q_{i'}-\bm q_{i}\in \R^d$. The model $(\G,\E)$ can be viewed as an undirected graph with no on-site loop, i.e. self-interactions of the kind $V_0(\delta \bm q_e=0)$ are not allowed.  The graph is undirected in the sense that the interactive potential $V_e(\bm q_{i'}-\bm q_i)=V_e(\bm q_i-\bm q_{i'})$ appears only once in the Hamiltonian \eqref{Hamiltonian_n}. Without the loss of generality, we assume the uniform mass condition $m_i=m=1$.
We now choose a subset of the boundary oscillators $\B\subset\partial \G$ to impose thermal baths. For each $b\in\B$ we assume that a thermostat of the temperature $T_b>0$ is given, along with a coupling constant $\gamma_b>0$. With this setting, we obtain an $n$-dimensional heat conduction model given by the following system of stochastic differential equations:
\begin{equation}\label{SDE:n_d_heat}
\begin{dcases}
\begin{aligned}
d\bm q_i&=\bm p_idt\\
d\bm p_i&=-\partial_{\bm q_i}\H(\bm p,\bm q)dt-\gamma_i\bm p_idt+\sqrt{2k_BT_i\gamma_i}d\bm \W_i(t)
\end{aligned}
\end{dcases}
\qquad\qquad 
i\in\G
\end{equation}
where $\gamma_i=0$ for $i\in\G\setminus\B$ and $\gamma_i=\gamma_b$, $T_i=T_b$ for $b\in \B$. Since the Langevin forces only act on the boundary oscillators,  for a ``bulk" of oscillators that are away from the boundaries, the dynamics are kept deterministic and Hamiltonian. For modeling purposes, different boundary conditions can be specified. Typical choices are periodic, fixed or free boundaries. No matter which condition is used, the form of the Kolmogorov backward operator $\K$ corresponding to SDE \eqref{SDE:n_d_heat} can be written as: 
\begin{align}\label{Kolmogorov:n_d_heat}
\K=\sum_{b\in \B}(\gamma_bk_BT_b\partial_{\bm p_b}\cdot \partial_{\bm p_b}-\gamma_{b}\bm p_b\cdot \partial_{\bm p_b})+\sum_{i\in \G}(\bm p_i\cdot\partial_{\bm q_i}-\partial_{\bm q_i}\H(\bm p,\bm q)\cdot\partial_{\bm p_i}).
\end{align}
When all added thermostats have the same temperature $T_b=1/\beta, b\in\B$, then the system admits a unique invariant measure 
\begin{align}\label{eqil:heat}
d\mu_{\beta}=\rho_{eq}d\bm pd\bm q=\frac{1}{Z}e^{-\beta\H(\bm p,\bm q)}d\bm pd\bm q,
\end{align}
which is known as the Gibbs measure for thermal equilibrium. When the boundary temperatures are different, the Gibbs measure is no longer invariant and \eqref{SDE:n_d_heat} describes the dynamics of heat flowing from the higher-temperature thermostats to the lower-temperature thermostats. Under certain assumptions to the potential energy (see \cite{cuneo2018non}), it can be proved that the system approaches to a unique, steady state exponentially fast. In the literature, such a state is often called the nonequilibrium steady state (NESS). We further denote the steady state probability density as $\rho_{\text{NESS}}$. By introducing the Mori-type projection operator \eqref{Mori_P} in the Hilbert space $L^2(\R^{2|\G|d},\rho_{\text{NESS}})$, we can derive the GLE \eqref{gle_full}-\eqref{gle_projected}. According to Theorem \ref{thm:2nd-FDT}, we can obtain the following proposition:
\begin{prop}\label{prop:2nd_FDT_ndheat}
For an n-dimensional heat conduction model given by the SDE \eqref{SDE:n_d_heat}, if the potential energy $U_{i}(\bm q_i)$ and $V_{e}(\delta \bm q_e)$ satisfy certain conditions which ensure the smoothness and uniqueness of $\rho_{NESS}$, e.g. the one outlined in \cite{cuneo2018non}, then the following generalized second FDT holds for state space observable $u=u(\bm p,\bm q)$ in GLEs \eqref{gle_full}-\eqref{gle_projected}:
\begin{align*}
K(t)=-\frac{\langle f(0), f(t)\rangle_{\rho_{\text{NESS}}}}{\langle u^2(0)\rangle_{\rho_{\text{NESS}}}}+\frac{\langle w(0),f(t)\rangle_{\rho_{\text{NESS}}}}{\langle u^2(0)\rangle_{\rho_{\text{NESS}}}},
\end{align*}
where the additional term $w(0)$ is given by
\begin{align*}
w(0)=2\sum_{b\in \B}\gamma_bk_BT_b\partial_{\bm p_b}\cdot\partial_{\bm p_b}u(0)+\frac{2}{\rho_{\text{NESS}}}\sum_{b\in\B}\gamma_bk_BT_b\partial_{\bm p_b}\rho_{NESS}\cdot\partial_{\bm p_b}u(0).
\end{align*}
In particular, if observable $u$ is a function of bulk coordinates $\G\setminus\B$, i.e. $u=u(\bm p,\bm q)=u([\bm p_i]_{i=1}^{|\G\setminus\B|},[\bm q_i]_{i=1}^{|\G\setminus\B|})$, then the classical second FDT holds in the nonequilibrium steady state: 
\begin{align}\label{2nd-FDT_heat}
K(t)=-\frac{\langle f(0), f(t)\rangle_{\rho_{\text{NESS}}}}{\langle u^2(0)\rangle_{\rho_{\text{NESS}}}}.
\end{align}
\end{prop}
\begin{proof}
The Kolmogrov operator \eqref{Kolmogorov:n_d_heat} can be decomposed as 
\begin{align*}
\K=\underbrace{\sum_{b\in \B}\gamma_bk_BT_b\partial_{\bm p_b}\cdot\partial_{\bm p_b}}_{\D}-\underbrace{\sum_{b\in \B}\gamma_{b}\bm p_b\cdot \partial_{\bm p_b}+\sum_{i\in \G}(\bm p_i\cdot \partial_{\bm q_i}-\partial_{\bm q_i}\cdot\H(\bm p,\bm q)\partial_{\bm p_i})}_{\L}=\L+\D.
\end{align*}
It is easy to get the desired result using Theorem \ref{thm:2nd-FDT} and Corollary \ref{corollary_2ndFDT0}.
\end{proof}
For the heat conduction model \eqref{SDE:n_d_heat}, a physically meaningful observable $u=u(\bm q,\bm p)$ is the heat flux. If we consider an oscillator chain (one-dimensional case) with symmetric on-site and neighbourhood interaction potential energy. i.e. $U_{i}=U(q_{i})$ and $V_{e}(\delta q_e)=V(q_{i+1}-q_{i})$. The local and total heat flux of the system can be defined as \cite{lepri2003thermal,spohn2016fluctuating,lepri1998anomalous}: 
\begin{align}\label{def:local_global_flux}
J_{i}=p_iV'(q_{i+1}-q_i),\quad i\in\G\setminus\B, \qquad  J_{tot,N}=\sum_{j}p_iV'(q_{j+1}-q_j),\quad j\in \G\setminus\B
\end{align}
where $J_i$ is the local heat flux and $J_{tot,N}$ is the total one with $N=|\G\setminus\B|$. Proposition \ref{prop:2nd_FDT_ndheat} ensures the validity of the classical second FDT for observables of the bulk coordinates. This implies that for the local and total heat flux defined as \eqref{def:local_global_flux}, the classical second FDT holds in the NESS even though the explicit form of the steady state probability density $\rho_{NESS}$ is unknown. Note that our definition of the total heat flux excludes the heat flux at the chain boundary. Another frequently used definition of $J_{tot,N}(t)$ contains such contributions and can be decomposed as  
\begin{align}\label{def:global_flux_B}
J_{tot,N}(t)=J_{\G\setminus\B}(t)+J_{\B}(t),
\end{align}
where $J_{\G\setminus\B}(t)$ is the bulk contribution as \eqref{def:local_global_flux} and $J_{\B}(t)$ is the boundary contribution. When applying the generalized second FDT to \eqref{def:global_flux_B}, we have a non-zero additional term $w(0)$ which breaks the classical second FDT. However, with some weak assumptions, we can show (see \ref{sec:App1}) that the dynamics of the averaged heat flux $J_{av}(t)=J_{tot,N}(t)/N$ can be approximated by the bulk contribution $J_{\G\setminus\B}(t)/N$ in the thermodynamic limit as $N\rightarrow\infty$. Hence we conclude that the GLE \eqref{gle_full}-\eqref{gle_projected} model for the averaged heat flux satisfies the classical second FDT in the thermodynamic limit.

Lastly, we want to comment on the mathematical difficulty to get similar results on the second FDT for nonequilibrium systems generated by deterministic forces. Consider a similar heat conduction chain model driven by Nos\'e-Hoover thermostats. By assuming that $m_i=m=1$ and there are only two thermostats with temperature $T_L$ and $T_R$, the dynamics is described by the following equations of motion \cite{lepri2003thermal}: 
\begin{align}\label{Nose-Hover_heat}
\begin{dcases}
    \frac{dq_i}{dt}&=p_i\\ 
    \frac{dp_i}{dt}&=-\partial_{q_i}\H(\bm p,\bm q)-
    \begin{dcases}
    \gamma_L p_i,\qquad \text{if $i\in \B_{T_L}$}\\
    \gamma_R p_i,\qquad \text{if $i\in \B_{T_R}$}
    \end{dcases} 
\end{dcases},    
\end{align}
where $\B_{L,R}$ are the set of boundary oscillators which interact with thermostats at the temperature $T_{L,R}$. The cardinality of $\B_{L,R}$ are denoted as $|\B_{L,R}|$. The dynamics of the auxiliary variables $\gamma_{L,R}$ are given by:
\begin{align}\label{Nose-Hover_bath}
    \frac{d\gamma_{L,R}}{dt}=\frac{1}{\theta_{L,R}}\left(\frac{1}{k_BT_{L,R}|\B_{T_{L,R}}|}\sum_{n\in \B_{T_{L,R}}}p_n^2-1\right),
\end{align}
where $\theta_{L,R}$ are the thermostat response times. It is easy to check that the velocity field for the combined system \eqref{Nose-Hover_heat}-\eqref{Nose-Hover_bath} has divergence
\begin{align*}
    \nabla\cdot\bm F(\bm p,\bm q,\bm \gamma)=-\gamma_L(t)-\gamma_R(t)
\end{align*}
which changes back and forth between positive and negative values depending on the kinetic temperature of the boundary oscillators \cite{lepri2003thermal}. As a consequence, the whole system oscillates between energy dissipating state and increasing state. Mathematically speaking, it is very hard to define a proper probability measure to quantify such a nonequilibrium steady state. Even the SRB measure \cite{bonetto2000fourier,ruelle1999smooth}, which works for dissipative systems, is not applicable to this case. In practice, one may assume the ergodicity of the deterministic system so that the ensemble average $\langle \cdot \rangle_{\rho_{NESS}}$ can be replaced by the time average. Then due to the similarity of the dynamics, it is reasonable to expect the heat conduction model generated by deterministic thermostats shares some properties with the stochastic model \eqref{SDE:n_d_heat}, including the classical and the generalized second FDT. This conjecture, however, needs to be verified. 
%
%
%
\section{Application to reduced-order modeling}\label{sec:Numerical}
In this section, we apply the generalized second FDT to the reduced-order modeling of the heat conduction problem. We first propose suitable methods to approximate the memory kernel $K(t)$ and the fluctuating force $f(t)$ for a GLE satisfying the {\em classical} second FDT. The resulting stochastic integro-differential equation serves as the reduced-order model for Gaussian observables of the nonequilibrium system. This model enables us to numerically verify the generalized second FDT. Secondly, we propose a polynomial chaos expansion method to approximate the dynamics of non-Gaussian observables which satisfy the {\em generalized} second FDT. Applying these two methods to the averaged heat flux $J_{av}(t)$ leads to dynamical models which characterize the steady state heat transfer in nonequilibirum systems. We note that a similar approach was used in \cite{chu2017mori} for Hamiltonian systems and the resulting stochastic model is often referred to as the {\em fluctuating heat conduction model}.
\subsection{Methodology}\label{sec:Numerical_method}
Without loss of generality, we consider a one-dimensional GLE for scalar observable $u(t)$
\begin{align}\label{eqn:GLE_oned}
\frac{d}{dt}u(t)=\Omega u(t)+\int_0^tK(t-s)u(s)ds+ f_u(t).
\end{align}
In \eqref{eqn:GLE_oned}, the streaming constant $\Omega$ is easy to obtain using the definition \eqref{streaming}. However, evaluating the memory kernel $K(t)$ and the fluctuation force $f_u(t)$ from the first principle is rather challenging since it involves the approximation of the high dimensional orthogonal flow $e^{t\Q\K}$. To avoid such technical difficulties, here we adopt a data-driven method introduced in \cite{zhu2020hypoellipticity2} to approximate the memory kernel $K(t)$. As for the fluctuation force $f_u(t)$, it can be approximated by suitable series expansions of a stochastic process. The whole procedure can be described as follows.
First of all, we recall that the projected GLE yields (see Section \ref{sec:EMZE}) the evolution of the steady state correlation function $C(t)$ of $u(t)$:
\begin{align}\label{eqn:PGLE_heat}
\frac{d}{dt}C_u(t)=\Omega C_u(t)+\int_0^tK(t-s)C_u(s)ds, \qquad 
\text{where}
\qquad 
C_u(t)=\frac{\langle u(t),u(0)\rangle_{\rho}}{\langle u^2(0)\rangle_{\rho}}.
\end{align}
Hence with equation \eqref{eqn:PGLE_heat}, the memory kernel $K(t)$ can be represented formally using the inverse Laplace transform if we know $C_u(t)$. If we further assume that the dynamics of the observable $u(t)$ is a stationary Gaussian process, since GLE \eqref{eqn:GLE_oned} is a linear equation for $u(t)$, this implies the fluctuation force $f_u(t)$ is also a stationary Gaussian process. Then we can use the truncated Karhunen-Lo\'eve (KL) expansion series to approximate $f_u(t)$, namely
\begin{align}\label{KL_expansion_f}
f_u(t)\simeq \langle f_u(t)\rangle_{\rho}+\sum_{k=1}^{K}\xi_k\sqrt{\lambda_k}e_k(t)
=\bar{f}_u+\sum_{k=1}^{K}\xi_k\sqrt{\lambda_k}e_k(t).
\end{align}
In the steady state, the mean value of the stochastic process satisfies $\langle f_u(t)\rangle_{\rho}=\langle f_u(0)\rangle_{\rho}=\bar{f}_u$, which can be obtained by taking the ensemble average of the GLE \eqref{eqn:GLE_oned} and then evaluating it at $t=0$:
\begin{align}\label{f_t=0}
\bar{f}_u=\langle f(0)\rangle_{\rho}&=\langle u(0)\rangle_{\rho}- \Omega\langle u(0)\rangle_{\rho}=\langle u(0)\rangle_{\rho}
-\frac{\langle \dot{u}(0),u(0)\rangle_{\rho}}{\langle u^2(0)\rangle_{\rho}}\langle u(0)\rangle_{\rho}.
\end{align}
The KL expansion random coefficients $\{\xi_k\}_{k=1}^K$ are necessarily independent Gaussian random variables satisfying $\langle\xi_i\xi_j\rangle=\delta_{ij}$, and $\{\lambda_k,e_k\}_{k=1}^K$ are,
respectively, eigenvalues and eigenfunctions of the homogeneous Fredholm integral equation
of the second kind:
\begin{align}\label{Fredholm_eqn}
\int_0^T\langle f_u(t),f_u(0)\rangle_{\rho}e_k(s)ds=\lambda_ke_k(t),\qquad t\in[0,T],
\end{align}
where $T$ is a certain numerical integration time. If the classical second FDT holds for GLE \eqref{eqn:GLE_oned}, by substituting $\langle f_u(t),f_u(0)\rangle_{\rho}=-K(t)\langle u^2(0)\rangle_{\rho}$ into eqn \eqref{Fredholm_eqn} and solving for $\{\lambda_k, e_k(t)\}_{k=1}^K$, we can get the exact KL series representation \eqref{KL_expansion_f} for the fluctuation force $f_u(t)$. With all these terms available, we propose the following data-driven modeling diagram for an arbitrary Gaussian observable $u(t)$:
\begin{align}\label{diagram}
\langle u(t),u(0)\rangle_{\rho},\langle u(0)\rangle_{\rho}
\Rightarrow
\begin{cases}
\displaystyle K(t)=\mathfrak{L}^{-1}\left[s-\Omega-\frac{C_u(0)}{\tilde C_u(s)}\right]\\
\displaystyle f_u(t)\simeq  \bar{f}_{u}+\sum_{k=1}^{K}\xi_k\sqrt{\lambda_k}e_k(t)
\end{cases}
\Rightarrow \text{Eqn}\ \eqref{eqn:uni_GLE_model},
\end{align}
where $\mathfrak{L}[C_u(t)]=\tilde{C}_u(s)$ is the Laplace transform of $C_u(t)$ and $\mathfrak{L}^{-1}[\cdot](t)$ is the inverse transform. In flowchart \eqref{diagram}, the leftmost three terms are inputs of the diagram which can be obtained by solving numerically the SDE \eqref{eqn:sde} and then averaging samples collected from the steady state simulation. We note that solving \eqref{eqn:PGLE_heat} for $K(t)$ is a well-known inverse problem which is ill-conditioned. In this paper, we use the series expansion method \cite{zhu2020hypoellipticity2} and LASSO regression to approximate $K(t)$. By combining all these approximations and executing \eqref{diagram}, we can get the first reduced-order model for $u(t)$:
\begin{align}\label{eqn:uni_GLE_model}
\frac{d}{dt}u(t)=\Omega u(t)+\sum_{i=1}^I\int_{0}^tk_ib_i(t-s)u(s)ds+f(t,\bm \xi),
\end{align}
where $K(t)=\sum_{i=1}^Ik_ib_i(t)$ is the series expansion approximation for $K(t)$ and $f(t,\bm \xi)$ is the truncated KL expansion approximateing the fluctuation force $f_u(t)$. We want to emphasize that the key relation that rationalizes the whole algorithm is the classical second FDT.

%

The above reduced-order modeling diagram cannot be applied to non-Gaussian cases nor the case where the {\em generalized} second FDT holds. The main modeling difficulties stem from the fact that the steady state distribution of the fluctuation force $f(t)$ and the additional term $w(0)$ in the memory kernel $K(t)$ are generally unknown and cannot be easily constructed from MD simulation. However, we can use the polynomial chaos expansion to directly simulate $u(t)$. To this end, we propose the following modeling diagram for non-Gaussian $u(t)$:
\begin{align}\label{diagram2}
\langle u(t),u(0)\rangle_{\rho},\rho_u, \langle u(0)\rangle_{\rho}
\Rightarrow
\begin{cases}
\displaystyle K(t)=\mathfrak{L}^{-1}\left[s-\Omega-\frac{C_u(0)}{\tilde C_u(s)}\right]\\
\displaystyle \text{Eqn}\ \eqref{eqn:PGLE_heat}
\end{cases}
\Rightarrow u(t)=\sum_{i=1}^Mu_iH_i(\gamma(t,\xi)).
\end{align}
In \eqref{diagram2}, $u(t)=\sum_{i=1}^{M}u_iH_i(\gamma(t,\bm \xi))$ is the polynomial chaos expansion for a stationary non-Gaussian process, which can be constructed from the time-autocorrelation and the steady state probability density $\rho_u$. Specifically, the expansion coefficient $u_i$ and a Gaussian process $\gamma(t,\bm \xi)$ is calculated via a modified Sakamoto-Ghanem algorithm \cite{sakamoto2002polynomial}. In \ref{sec:App_poly}, we explain the procedure in detail. By directly simulating the non-Gaussian processes in the state space, we avoid the computation of the infinite Kramers-Moyal coefficients \cite{Risken} or the effective Fokker-Planck diffusion coefficient \cite{chu2017mori}.

For these two reduced-order modeling methods, by approximating the full GLE or using the polynomial chaos expansion method, we can construct a surrogate model for $u(t)$. Moreover, it also enables us to use a short-time MD simulation data to {\em predict} the long-time dynamics of $u(t)$. This part will be verified later via numerical simulations in the following section.
\subsection{Numerical result for a one-dimensional heat conduction model}
We now study numerically the one-dimensional heat conduction model \eqref{SDE:n_d_heat}. In particular, we will use reduced order models introduced in Section \ref{sec:Numerical_method} to build effective models for different phase space observables, from which we can verify the validity of the generalized second FDT and demonstrate the effectiveness of these reduced order models. Moreover, we will study in detail the model for the averaged heat flux $J_{av}(t)$ and discuss its usefulness in characterizing the heat transport intensity for systems in and out of the statistical equilibrium.   

To this end, we set the on-site potential energy in \eqref{SDE:n_d_heat} to be 0 and the neighbourhood interaction potential energy to be the Lennard-Jones (LJ) potential energy, i.e.:
\begin{align*}
V_{e}(\delta q_e)=V(q_{i+1}-q_{i})=4\epsilon\left[\left(\frac{\sigma}{q_{i+1}-q_i}\right)^{12}-2\left(\frac{\sigma}{q_{i+1}-q_i}\right)^{6}\right].
\end{align*}
The whole chain is linked with two thermostats with temperatures $T_L$ and $T_R$ which will be specified later for different cases. Free boundary conditions are imposed and the modeling parameters are set as follows: $N=|\G|=256$, $\epsilon=0.2$, $\sigma=1$, $\gamma_L=\gamma_R=1$. To solve \eqref{SDE:n_d_heat} numerically, we use the Euler-Maruyama scheme with step size $dt=10^{-5}$.  In Figure \ref{fig:Heat_transport}, we show the sample trajectories of selected observables of this stochastic system. 
\begin{figure}[t]
\centerline{
\includegraphics[height=4cm]{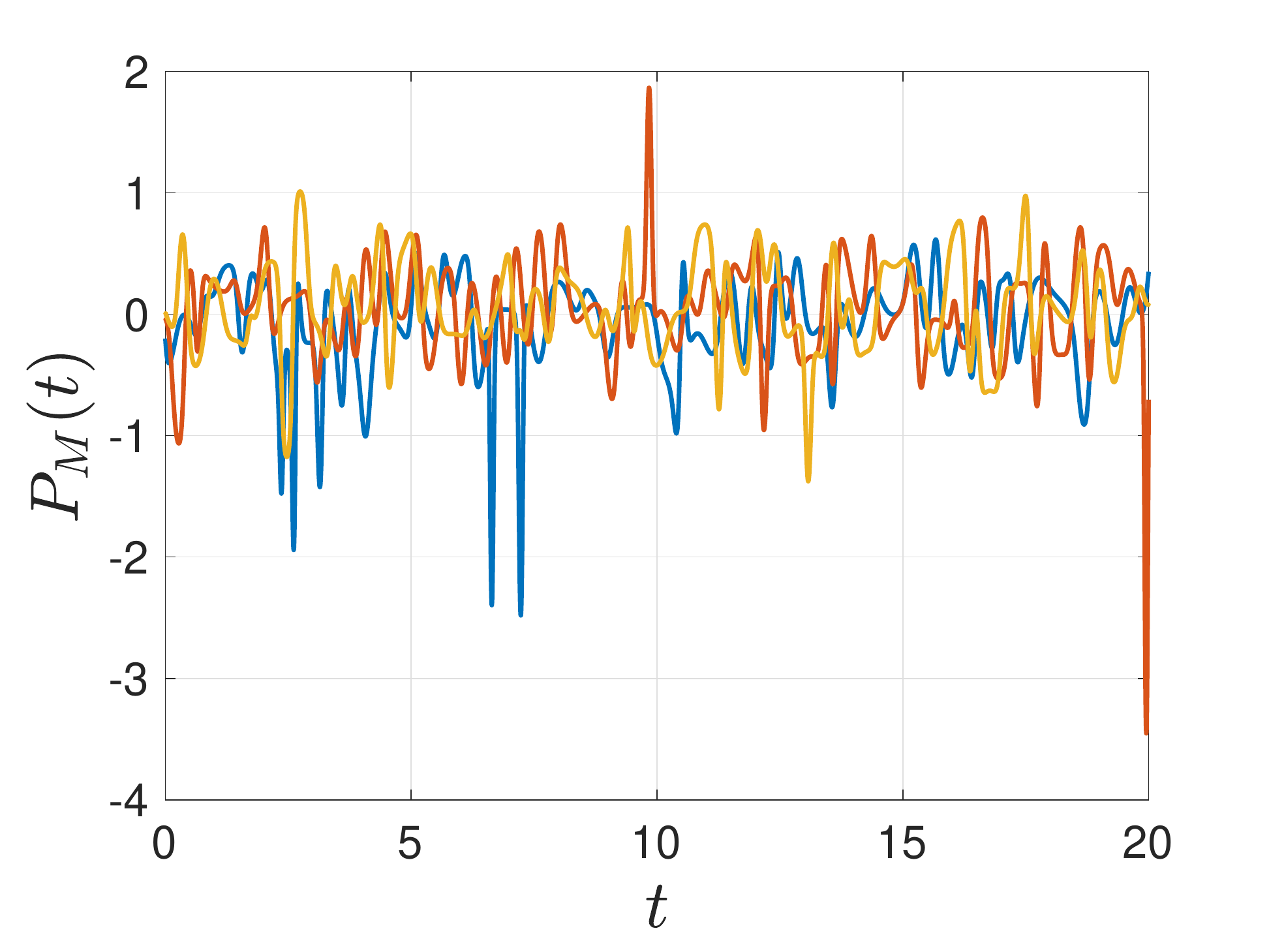}
\includegraphics[height=4cm]{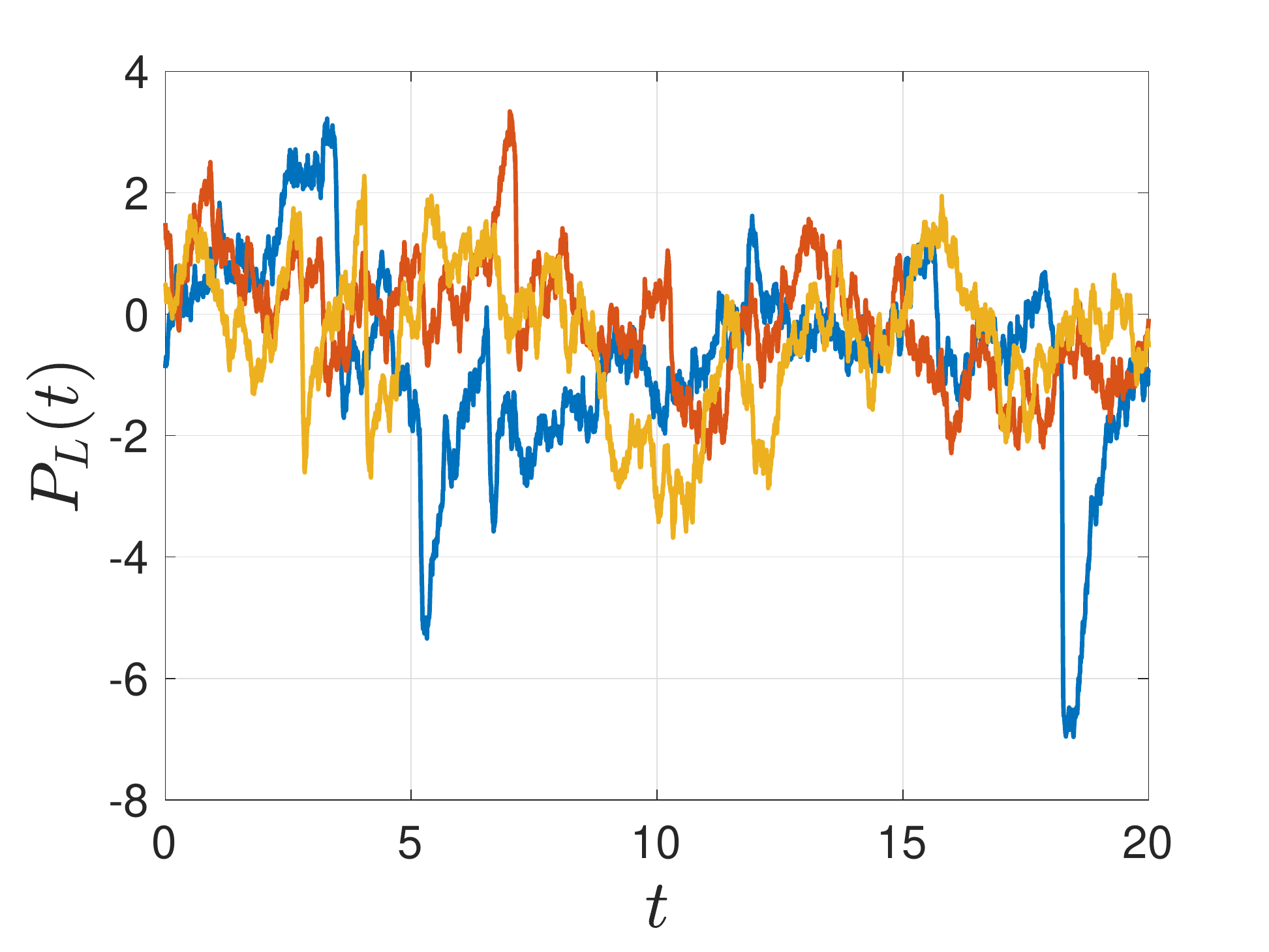}
\includegraphics[height=4cm]{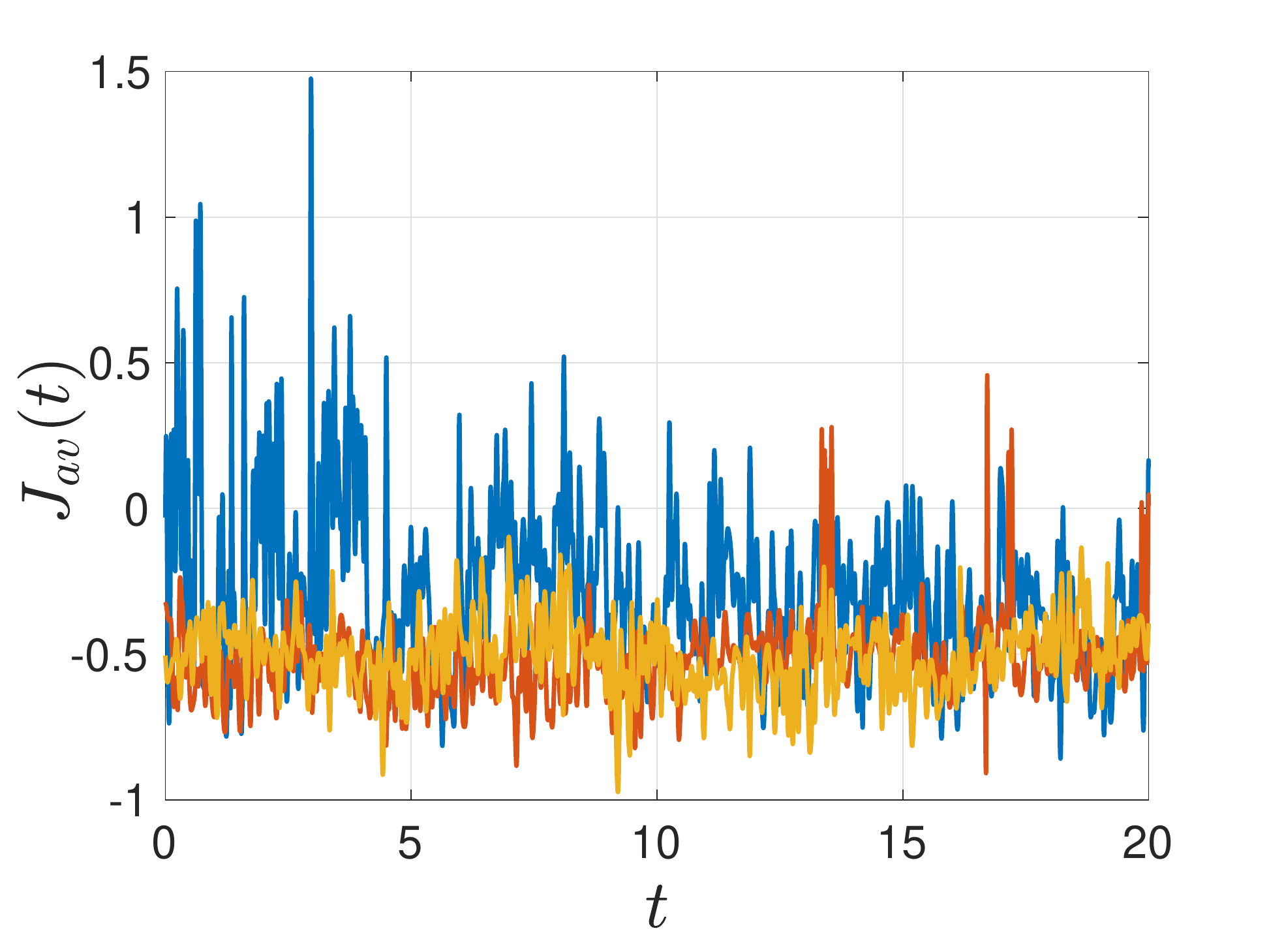}
}
\caption{Sample trajectories of the particle momenta $P_L(t)$, $P_M(t)$ and the averaged heat flux $J_{av}(t)$. The temperatures of the thermostats are set to be $T_L=1$ and $T_R=5$. The displayed time domain is $[80,100]$, where the system is verified to be in the NESS after the transient time $t=80$. }
\label{fig:Heat_transport_path} 
\end{figure}
\subsubsection{Verification of the generalized second FDT}
\label{sec:veri-2nd_FDT}
The reduced-order modeling method we introduced in Section \ref{sec:Numerical_method} enables us to numerically verify the generalized second FDT by hypothesis testing. We will first choose an observable for which the {\em classical} second FDT holds ($\Omega=0$ and $w(0)=0$) and show that the stochastic model \eqref{eqn:uni_GLE_model} gives correct statistics for observable $u(t)$. Then we will repeat the procedure for an observable for which the {\em generalized} second FDT holds ($\Omega\neq 0$ and $w(0)\neq 0$). Since the algorithm works {\em only} when the classical second FDT holds, stochastic model \eqref{eqn:uni_GLE_model} should give wrong statistics for observable $u(t)$. Throughout this subsection, the thermostat temperatures are set to be $T_L=1$ and $T_R=5$ which makes the system approaching an NESS as $t\rightarrow+\infty$. 
\paragraph{Observable $p_M$} We choose the momentum of the oscillator in midst as the observable. According to Proposition \ref{prop:2nd_FDT_ndheat}, the momentum $p_M$ satisfies the {\em classical} second FDT \eqref{2nd-FDT_heat} in the NESS. Figure \ref{fig:PLPM} shows that $p_M$ is a Gaussian variable and the marginal probability density function (PDF) satisfies $p_M\sim \N(0.2955,0.42^2)$ approximately. Therefore the standard KL expansion can be used to represent the the fluctuation term, which leads to the reduced-order model for $p_M(t)$:
\begin{align}\label{eqn:PM_full_GLE}
\frac{d}{dt}p_M(t)=\sum_{i=1}^I\int_0^tk_ib_i(t-s)p_M(s) ds+\sum_{i=1}^M\sqrt{\lambda_i}\xi_ie_i(t).
\end{align}
In this paper, the memory kernel expansion basis $b_i(t)$ are set to be the Laguerre polynomials, and a LASSO regression method is used to approximate the expansion coefficient $k_i$ \cite{zhu2020hypoellipticity2}. In \eqref{eqn:PM_full_GLE}, $I=20$ and $M=500$ (the same hereinafter). Since a stationary Gaussian process $p_M(t)$ is fully characterized by its marginal distribution and the time autocorrelation function, we solve \eqref{eqn:PM_full_GLE} numerically and compare these two statistics with the exact results obtained through MD simulation. In Figure \ref{fig:PLPM}, we can see that the solution of the reduced-order model \eqref{eqn:PM_full_GLE} reproduces the correct statistics of the observable $p_M(t)$. 
\begin{figure}[t]
\centerline{
\includegraphics[height=5cm]{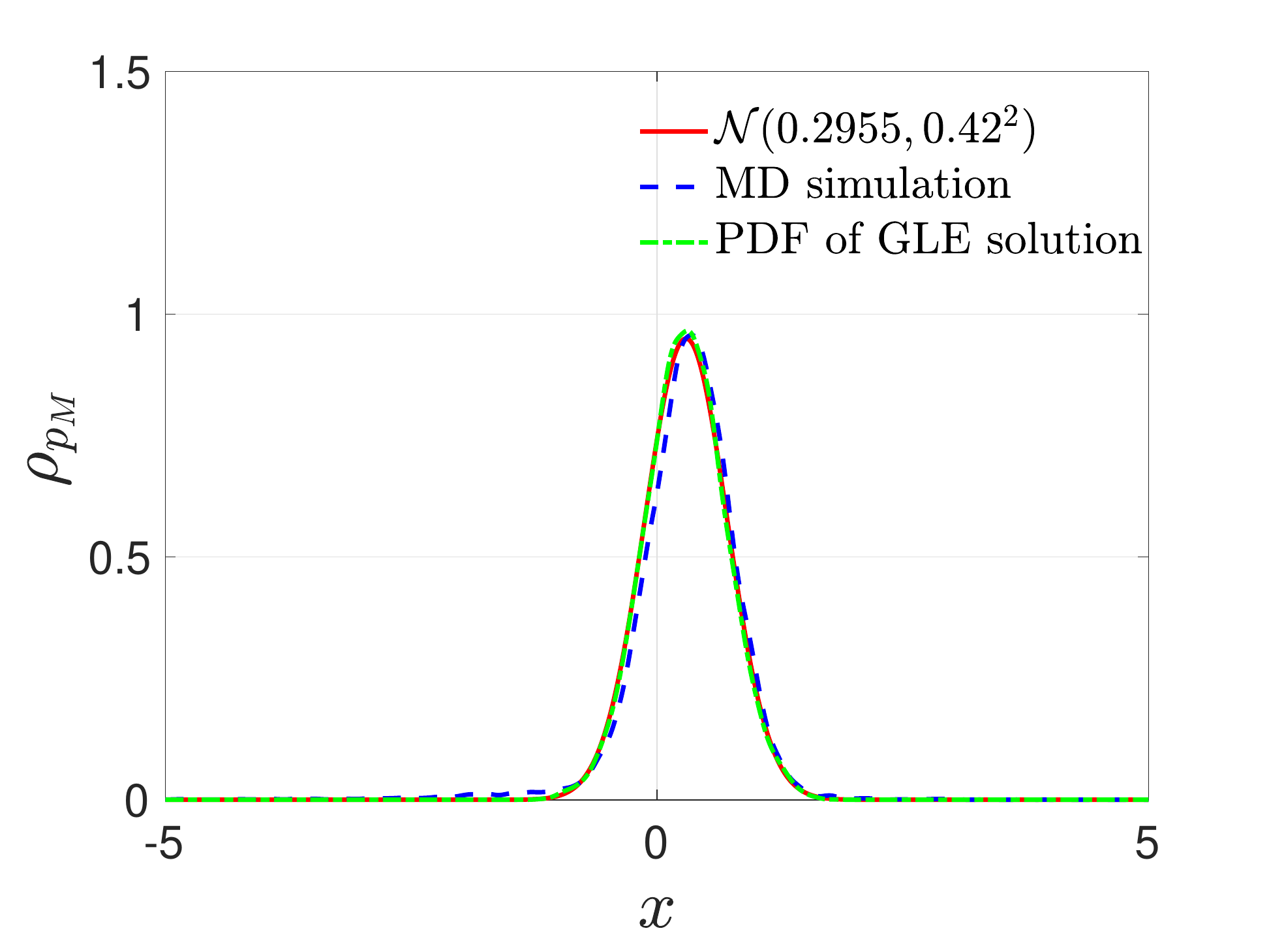}
\includegraphics[height=5cm]{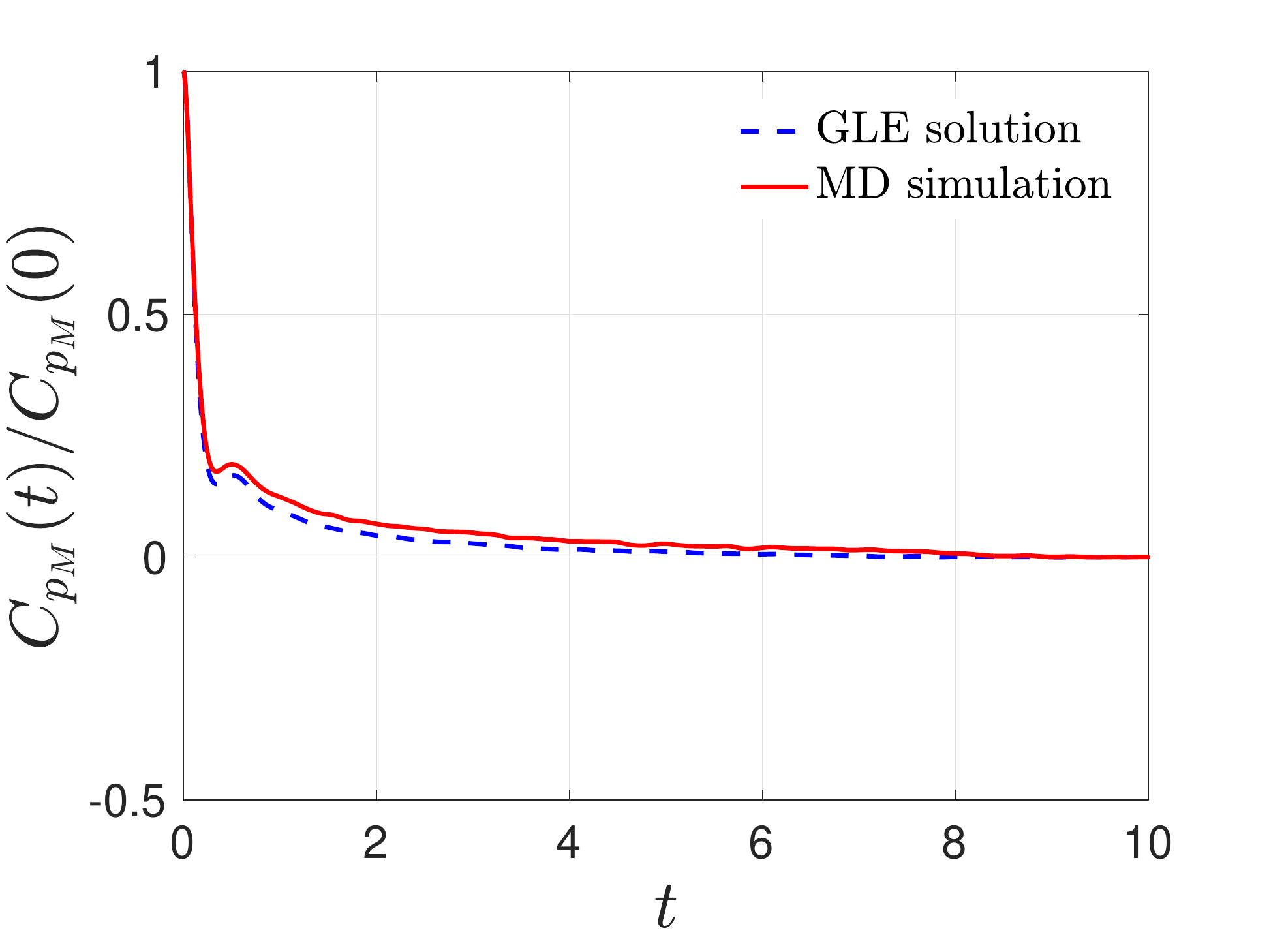}
}
\centerline{
\includegraphics[height=5cm]{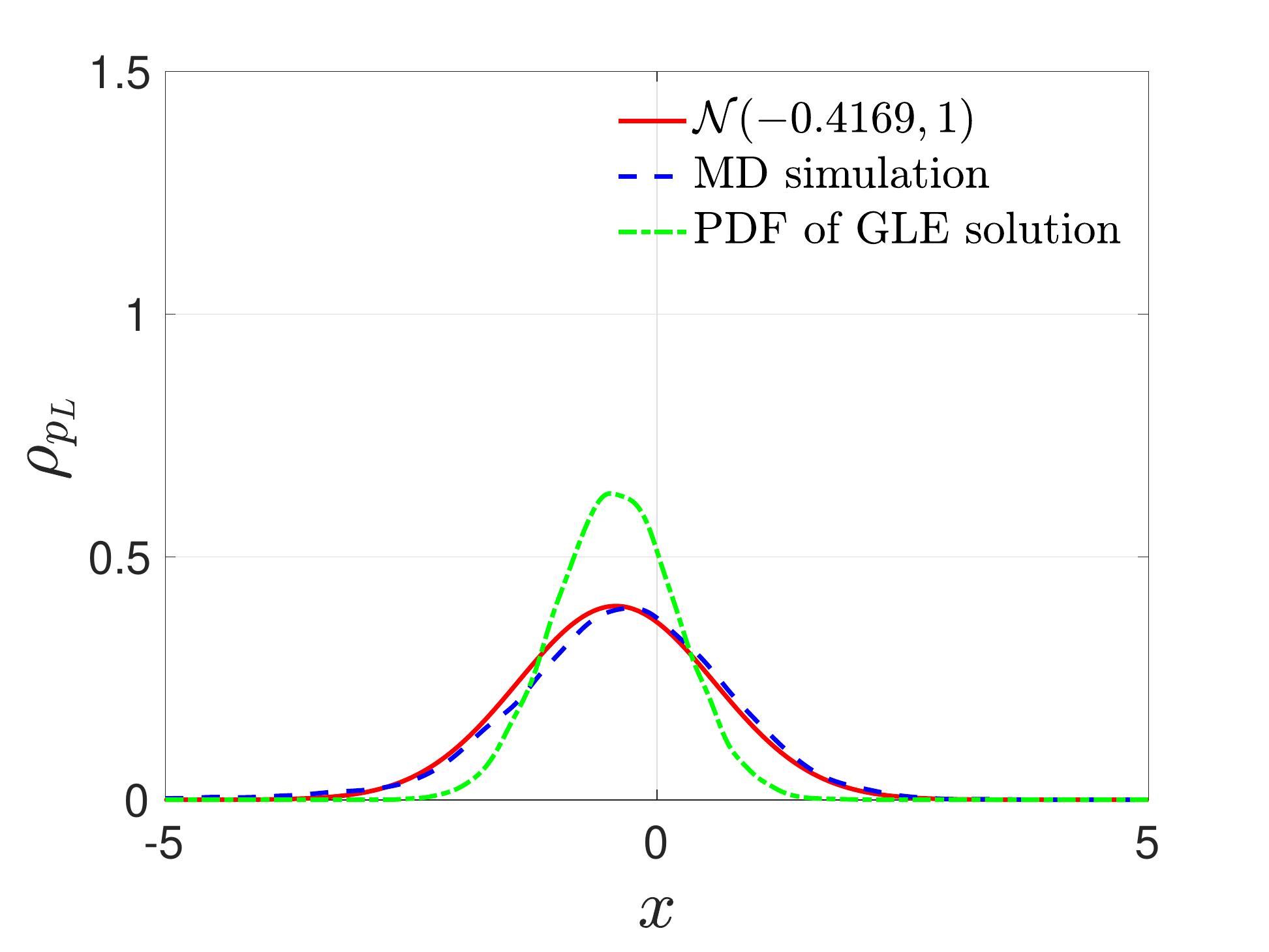}
\includegraphics[height=5cm]{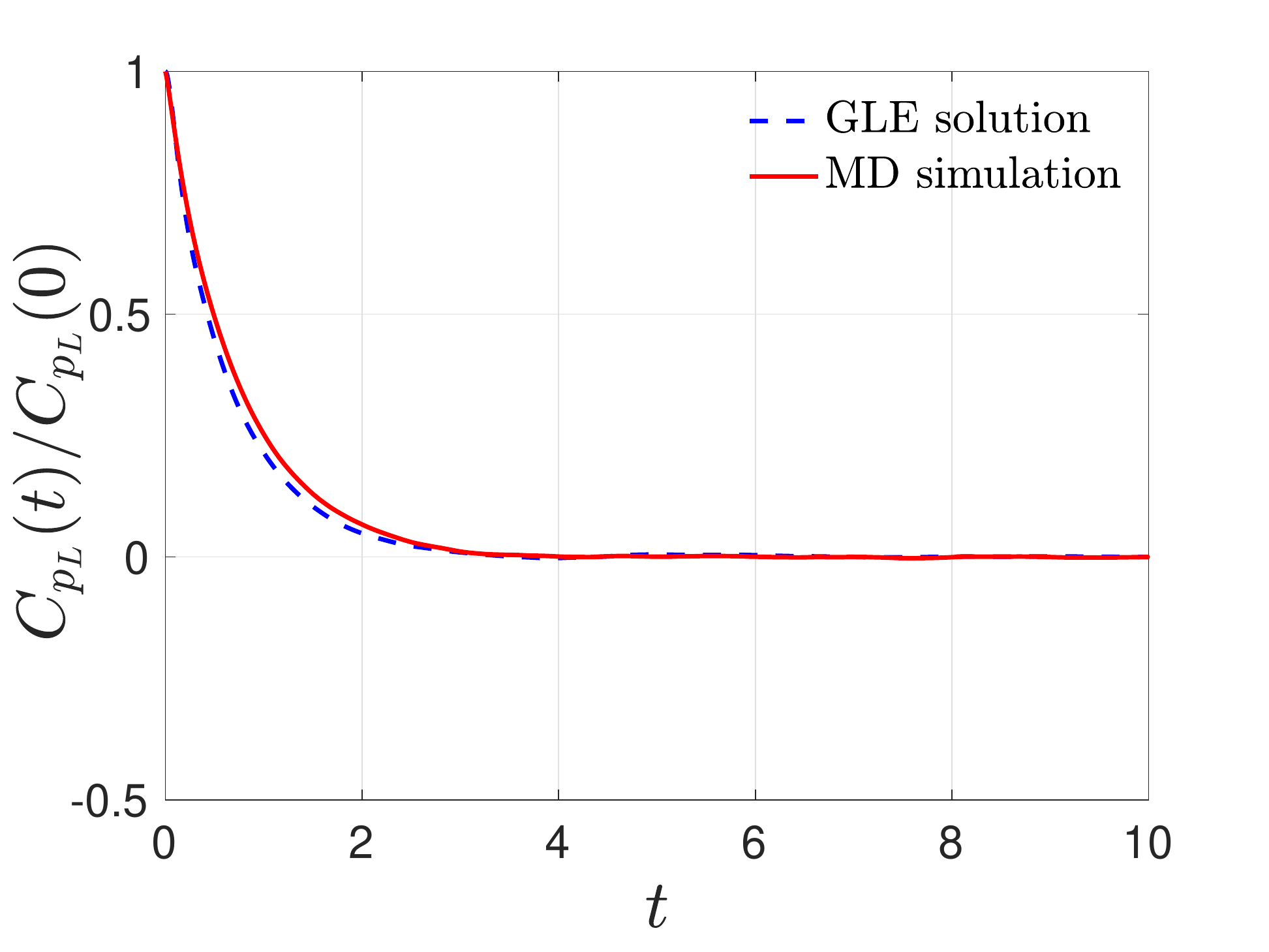}
}
\caption{
(First column) Marginal probability density for $p_M$ (Up-Left) and $p_L$(Down-Left); (Second column) Normalized time autocorrelation function $C(t)/C(0)$ for observable $p_M$ (Up-Right) and $p_L$ (Down-Right). The marginal PDF is obtained from the MD simulation data using kernel density estimation. The correlation function is obtained by averaging 5000 sample trajectories within the time domain $[90,100]$ while the system is in the NESS.}
\label{fig:PLPM}
\end{figure}
\paragraph{Observable $p_L$} We choose the momentum of the leftmost oscillator as the observable. According to Proposition \ref{prop:2nd_FDT_ndheat}, the momentum $p_L$ satisfies the {\em generalized} second FDT \eqref{2nd-FDT_heat} ($w(0)\neq 0$) in the NESS. Figure \ref{fig:PLPM} shows that $p_L$ is also a Gaussian variable with the marginal PDF satisfying $p_L\sim \N(-0.4169,1)$ approximately. If we assume that the {\em classical} second FDT holds, then the reduced-order model for $p_L(t)$ is given by
\begin{align}\label{eqn:PL_full_GLE}
\frac{d}{dt}p_L(t)=\Omega p_L(t)+\sum_{i=1}^I\int_0^tk_ib_i(t-s)p_L(s) ds+\sum_{i=1}^M\sqrt{\lambda_i}\xi_ie_i(t),
\end{align}
Using the data we obtained in the NESS we obtain $\Omega\approx -0.8622$. When comparing the marginal distribution and the time autocorrelation function with the exact MD simulation results, we find that the reduced-order model reproduces the correct autocorrelation function for observable $p_L(t)$, which is reasonable since the evolution equation for $C_{p_L}(t)$ is given by \eqref{eqn:PGLE_heat} and only depends on the memory kernel $K(t)$. However, since the {\em classical} second FDT does not hold ($w(0)\neq 0$) for the GLE of $p_L(t)$, the reduced-order model with the pre-assumed classical second FDT must reproduce the wrong NESS marginal distribution $\rho_{p_L}$. 
\paragraph{Remark}Combining these numerical simulation results for observable $p_M(t)$ and $p_L(t)$, we verify the existence of the additional term $w(0)$ in the memory kernel $K(t)$ that violates the classical second FDT. In this paper, we would not determine the specific form of $w(0)$ and leave it as an independent research topic. 
%
%
\subsubsection{Stochastic modeling of the averaged heat flux}
\label{sec:heat_red_mod}
In this subsection, we use the reduced-order technique to build dynamical models for the averaged heat flux $J_{av}(t)$ and show that \eqref{eqn:uni_GLE_model} is an effective, generic model for heat transport close to and far-from the statistical equilibrium. In addition, we show that \eqref{eqn:uni_GLE_model} with short MD simulation data can predict the correct long-term dynamics of the averaged heat flux $J_{av}(t)$. 

We firstly briefly review the classical Kubo's linear response theory for heat transport. For a stochastic model such as \eqref{SDE:n_d_heat}, the system is initially set to be in the statistical equilibrium \eqref{eqil:heat} and then perturbed by a small temperature difference acted on the boundary. This can be achieved by alternating the temperature of a thermostat as $T_{eq}+\Delta T$. If $\Delta T$ is sufficiently small, then the thermal conductivity $\kappa$ of the lattice system can be calculated using the first FDT, which is known as the Green-Kubo formula \cite{kubo2012statistical,lepri2003thermal,lepri1998anomalous,kundu2009green}:
\begin{align}\label{def:Green-Kubo}
\kappa=\frac{N}{k_BT_{eq}^2}\int_0^{\infty}\langle J_{av}(t),J_{av}(0)\rangle_{eq}dt,
\end{align}
where the ensemble average $\langle\cdot\rangle_{eq}$ is taken with respect to the equilibrium measure \eqref{eqil:heat} with temperature $T_{eq}$ and $N$ is the total particle number. Hence the Green-Kubo formula \eqref{def:Green-Kubo} links the {\em equilibrium} time autocorrelation function of the flux with the transport coefficient {\em near} the statistical equilibrium. We note that many low-dimensional heat conduction models exhibit a violation of Fourier's law with $\kappa=\kappa(N)$ depending on $N$. This, however, will not invalidate the Green-Kubo formula since the $N$-dependence for such special systems is embedded in the time correlator $\langle J_{av}(t),J_{av}(0)\rangle_{eq}$. A more difficult case for some low-dimensional systems is the anomalous long-time tail of the heat flux autocorrelation function since it scales as $t^{-d}$, $0<d<1$. This anomaly is normally associated with a chain-length dependent conductivity scaling $\kappa(N)\propto N^{\alpha}$ with $\alpha>1$, which will lead to divergent Green-Kubo integral \eqref{def:Green-Kubo} and an infinite conductivity $\kappa$. This phenomenon is rather system-dependent and strongly related to the boundary conditions. The LJ system under our investigation does not exhibit such an anomaly, which agrees with the finding in \cite{savin2014thermal} with fixed boundary conditions.

\paragraph{Equilibrium case} We first focus on the equilibrium case and show that the stochastic model we introduced in Section \ref{sec:Numerical_method} yields a correct prediction of the equilibrium time autocorrelation function $\langle J_{av}(t),J_{av}(0)\rangle_{eq}$. Note that the equilibrium case is a special case of the nonequilibrium model with $T_L=T_R=T_{eq}$, and we already proved that the GLE for the averaged heat flux $J_{av}(t)$ in the NESS satisfies the classical second FDT. If the equilibrium distribution for $J_{av}(t)$ is Gaussian, then the full stochastic model can be constructed via \eqref{diagram}, which generates a simulated sample trajectory of $J_{av}(t)$ as the solution of \eqref{eqn:uni_GLE_model} with $\Omega=0$. Figure \ref{fig:J_ave_Eqn} shows that equilibrium marginal distribution for $J_{av}(t)$ at low temperature $T_{eq}=1$ is approximately Gaussian with $J_{av}(t)\sim \N(0,0.21^2)$. Hence the corresponding fluctuation force $f_J(t)$ in \eqref{eqn:uni_GLE_model} can be approximated by a truncated KL series, which leads to the reduced-order model:
\begin{align}\label{eqn:E_J_full_GLE}
\frac{d}{dt}J_{av}(t)=\sum_{i=1}^I\int_0^tk_ib_i(t-s)J_{av}(s) ds+\sum_{i=1}^M\sqrt{\lambda_i}\xi_ie_i(t).
\end{align}
The heat flux has a longer correlation time scale when comparing to observables such as $p_L,p_M$. When solving \eqref{eqn:E_J_full_GLE} numerically, we use a short MD trajectory data for $t\in[0,10]$ to construct the memory kernel $K(t)$. From Figure \ref{fig:J_ave_Eqn}, we can see that the stochastic model \eqref{eqn:E_J_full_GLE} predicts the long time tail of the correlation function. 

%
The situation is different when increasing the system temperature to $T_{eq}=5$. In particular, the estimated PDF of $J_{av}(t)$ is clearly non-Gaussian and has a long tail. For modeling such heat flux, we adopt the second model \eqref{diagram2} where $J_{av}(t)$ is approximated by a polynomial chaos expansion 
\begin{align}\label{J_eq_Saka}
   J_{av}(t)=\sum_{i=1}^{M}J_iH_i(\gamma(t,\bm \xi)). 
\end{align}
The simulation result is shown in the second row of Figure \ref{fig:J_ave_Eqn}. We find that with short-term data, the generated stochastic process \eqref{J_eq_Saka} has a target distribution which agrees with the MD simulation result. Moreover, the extrapolated correlation function of \eqref{J_eq_Saka} predicates the correct long-time tail of the exact result. 
\paragraph{Remark}The reduced-order modeling we introduced for $J_{av}(t)$ is only semi-analytical. Namely, while the second FDT induced GLE is closed and derived from the first principle, the calculation of the memory kernel is extrapolated using a short-term data-driven method. This is less satisfying from a theoretical point of view. Here we want to mention the technical difficulty of developing a pure analytical method to get similar results for the system under our investigation. We note that the time autocorrelation function of the heat flux has multiple time scales. Specifically, there are a short-time scale, $\delta$-function alike fast decaying of the correlation close to $t=0$ and a long-time scale, slow decaying of the correlation as $t\rightarrow+\infty$. Generally speaking, it is hard to find an analytical method which generates such two-time-scale dynamics of the correlation simultaneously. Some recently developed analytical methods such as the nonlinear fluctuating hydrodynamics \cite{spohn2016fluctuating,mendl2015current} used mode-coupling theory to approximate the memory kernel and successfully predicted the long-time tail of the correlation function. But since the short-time scale dynamics cannot be captured within this framework, which is important for accurately evaluating the Green-Kubo integral \eqref{def:Green-Kubo}, it is hard to get the correct transport coefficient based on purely analytical calculations.   

\begin{figure}[t]
\centerline{
\includegraphics[height=4cm]{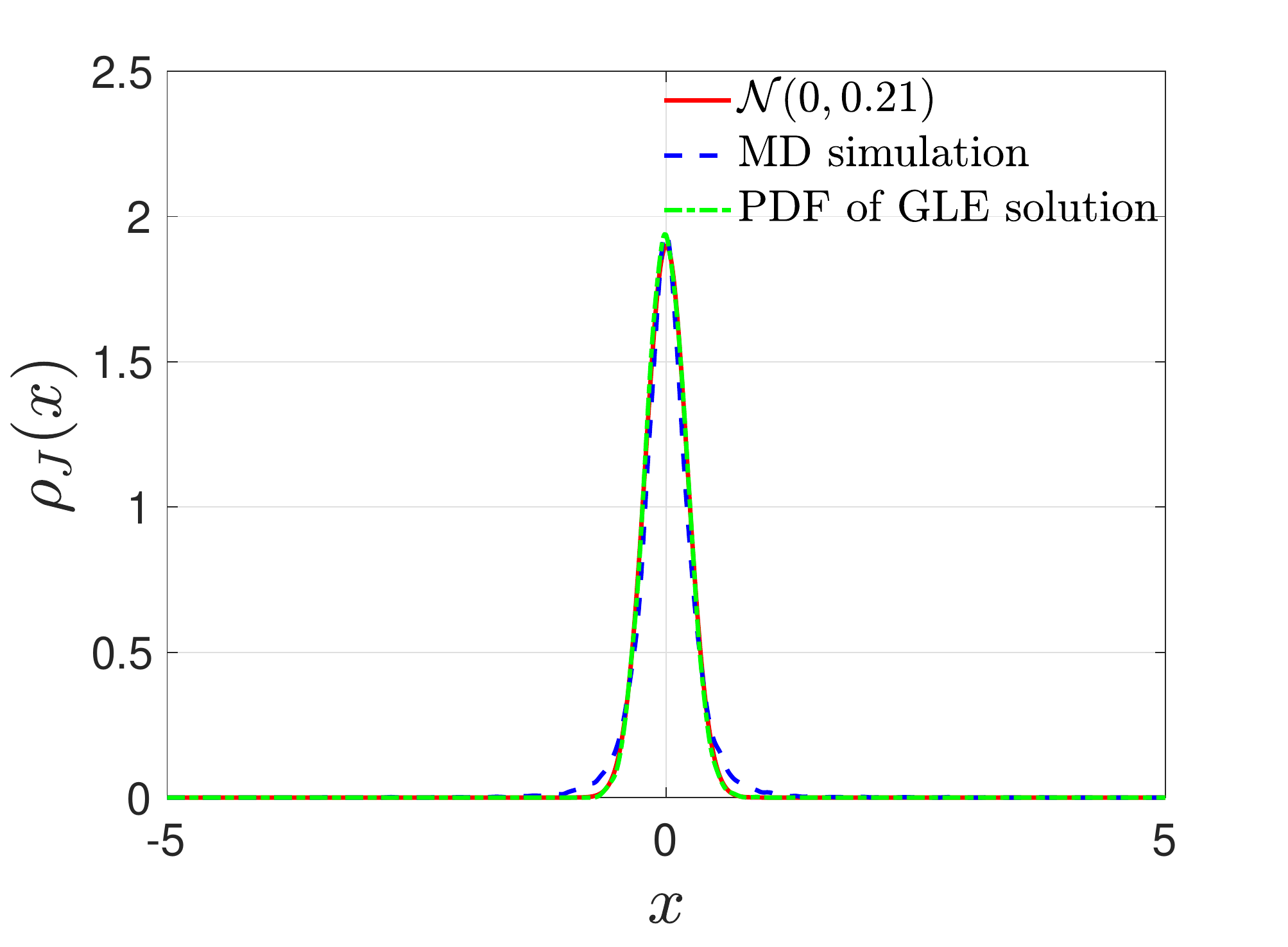}
\includegraphics[height=4cm]{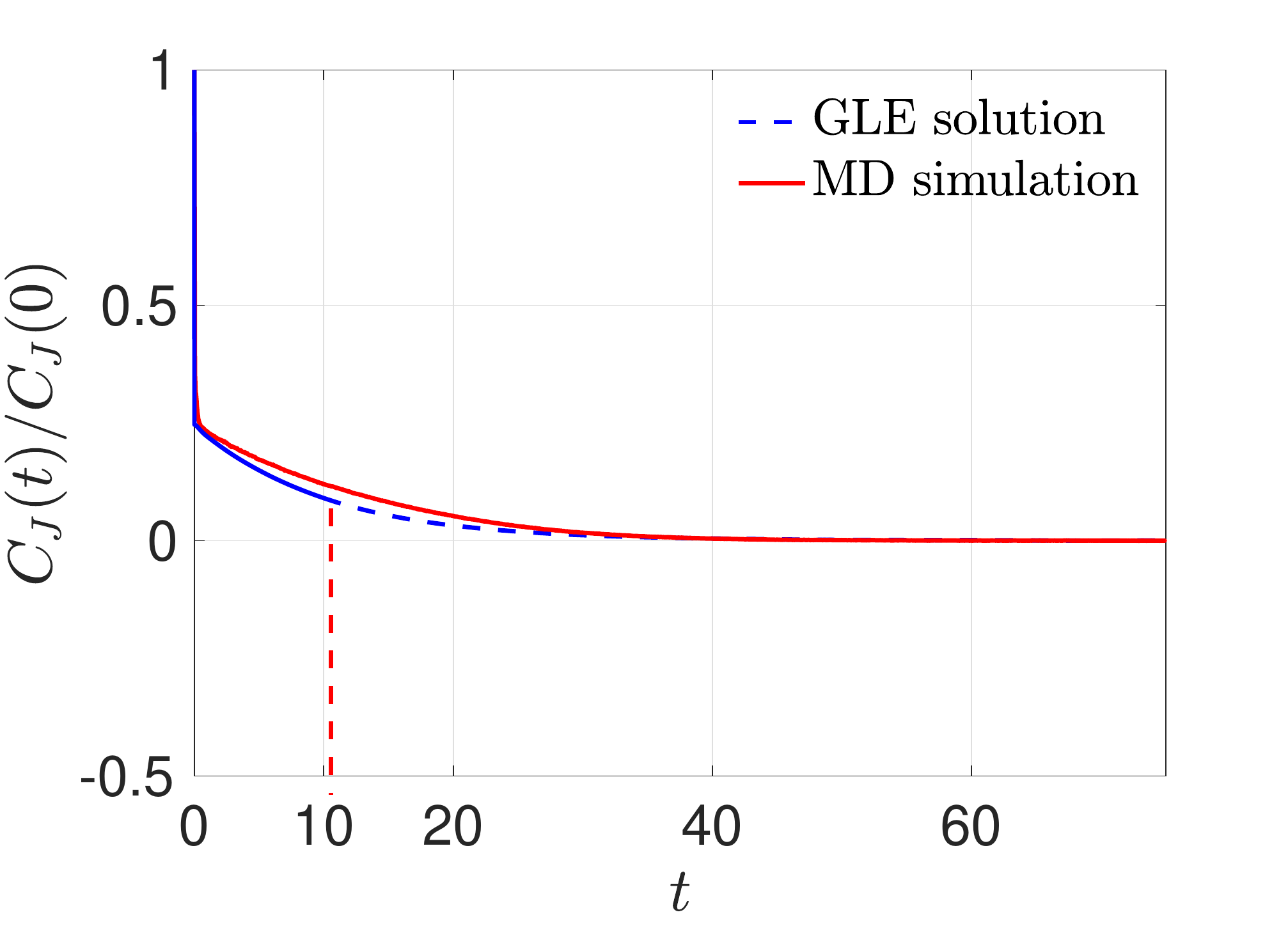}
\includegraphics[height=4cm]{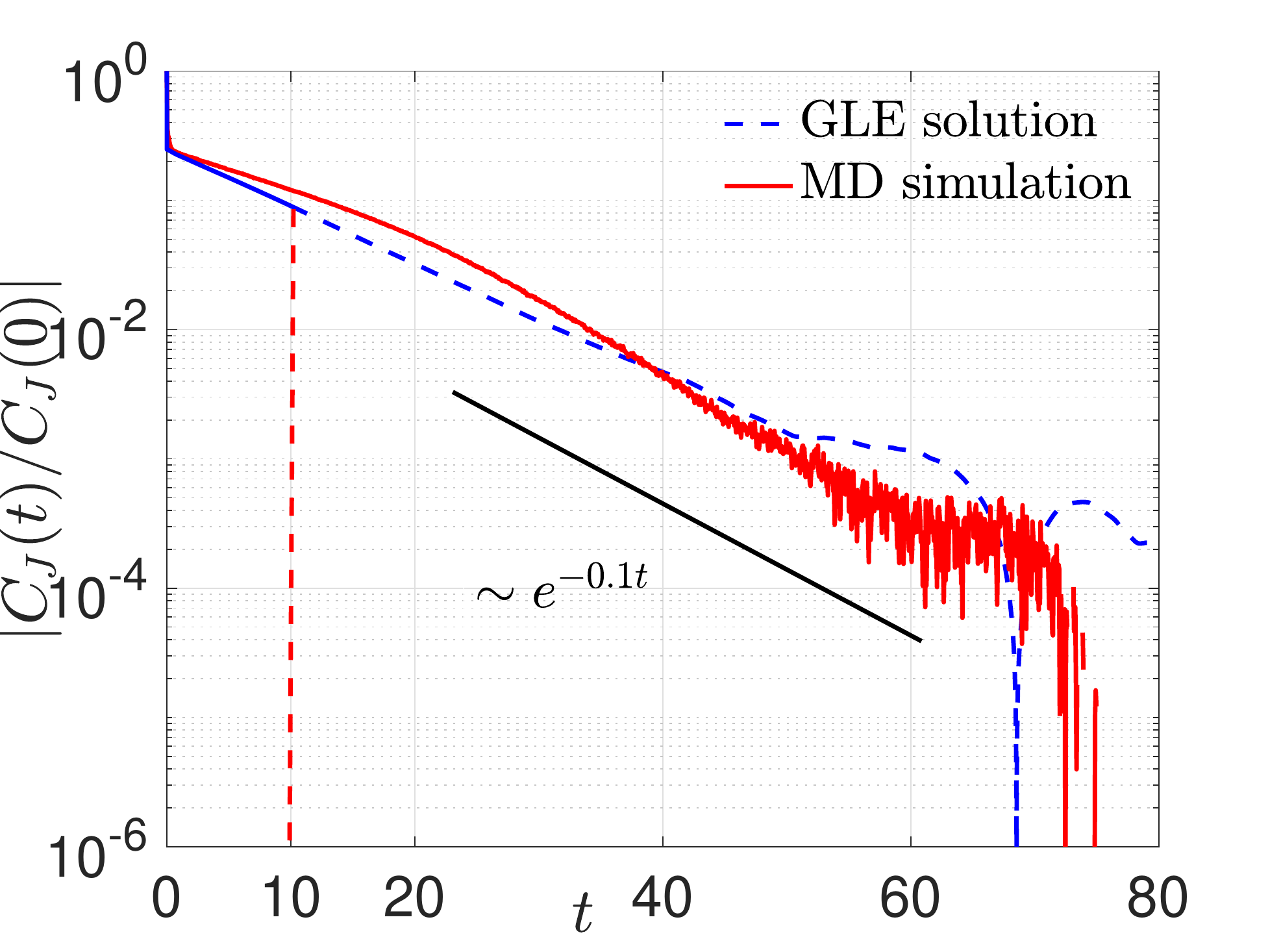}
}
\centerline{
\includegraphics[height=4cm]{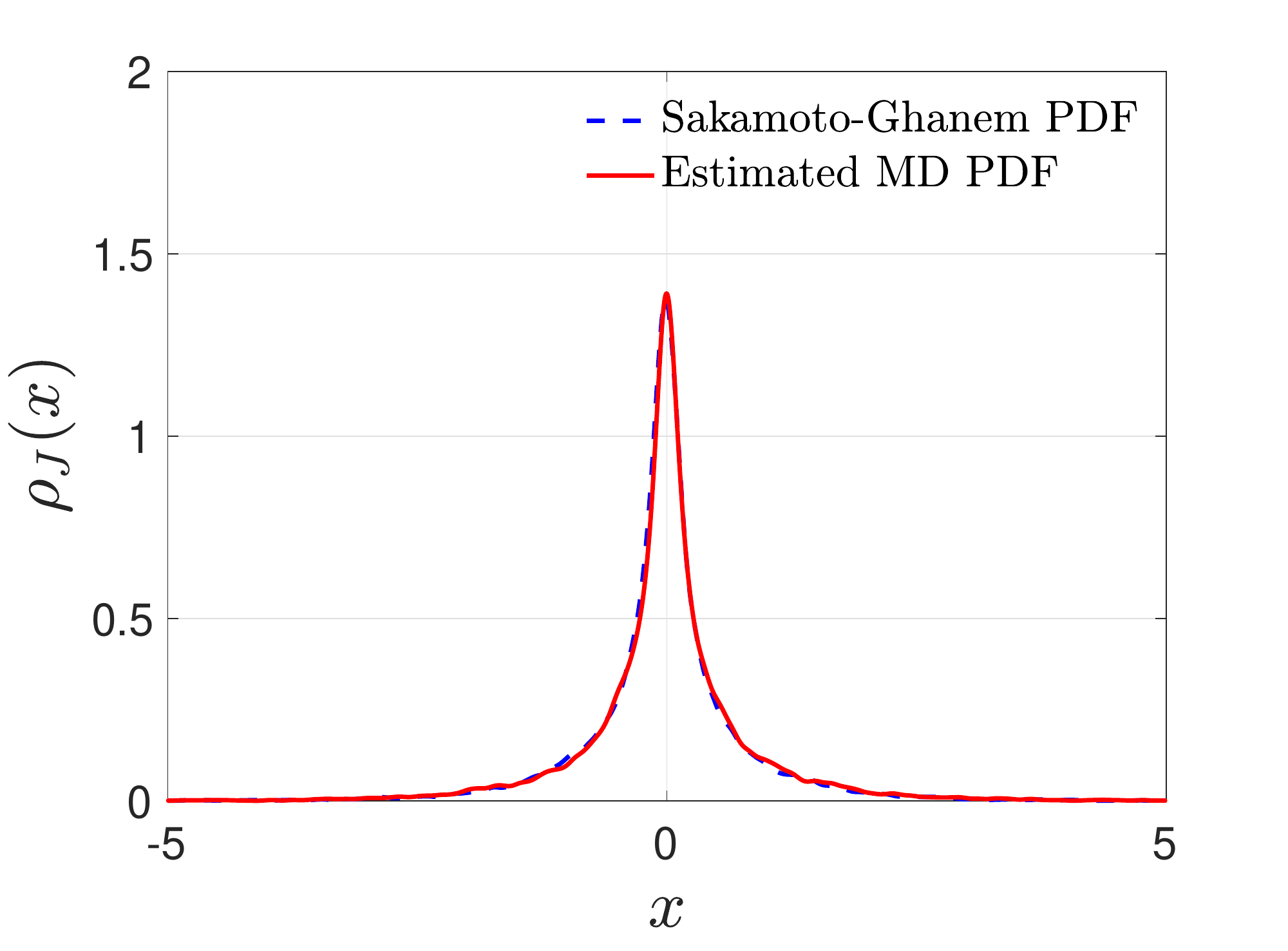}
\includegraphics[height=4cm]{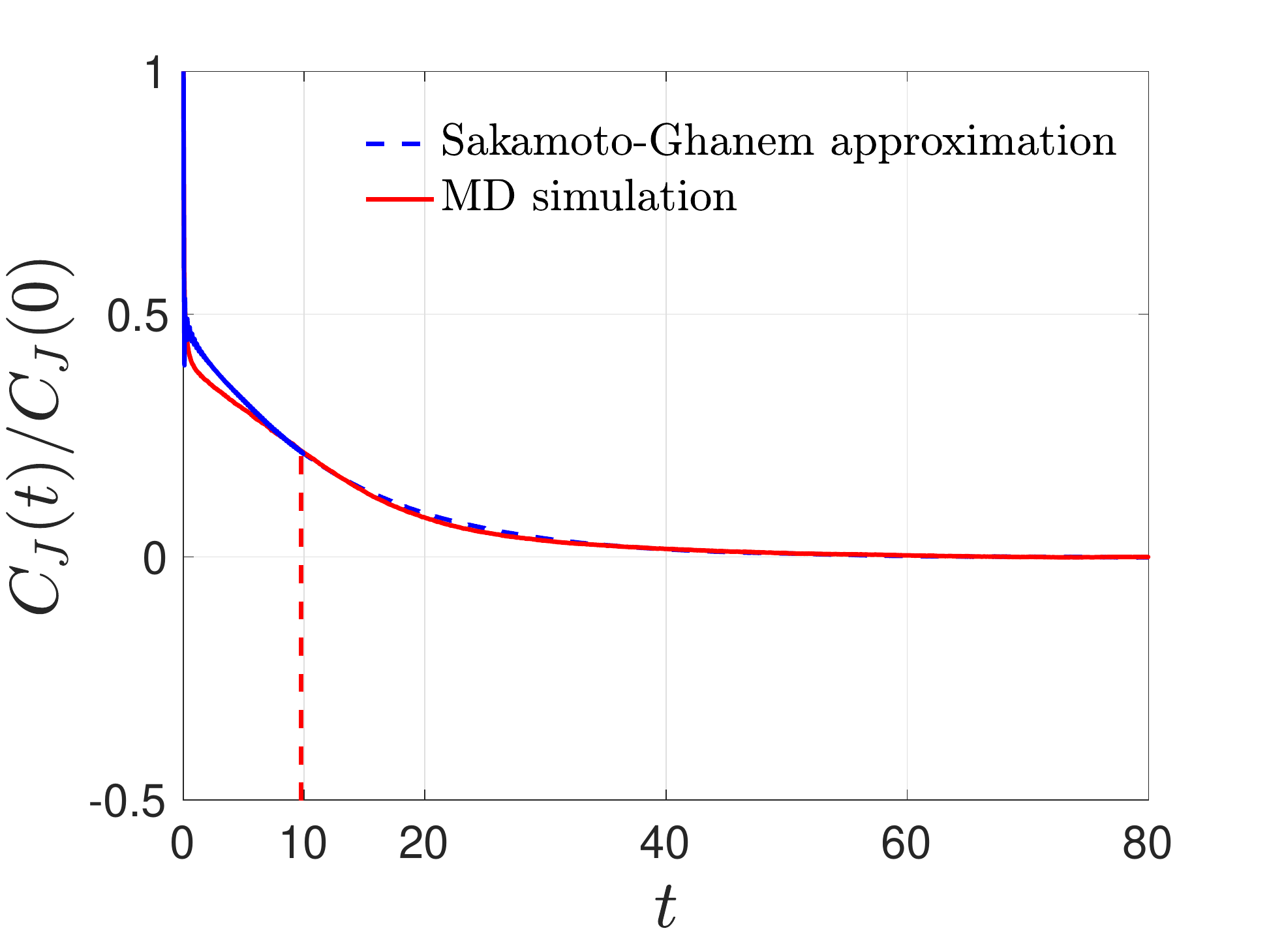}
\includegraphics[height=4cm]{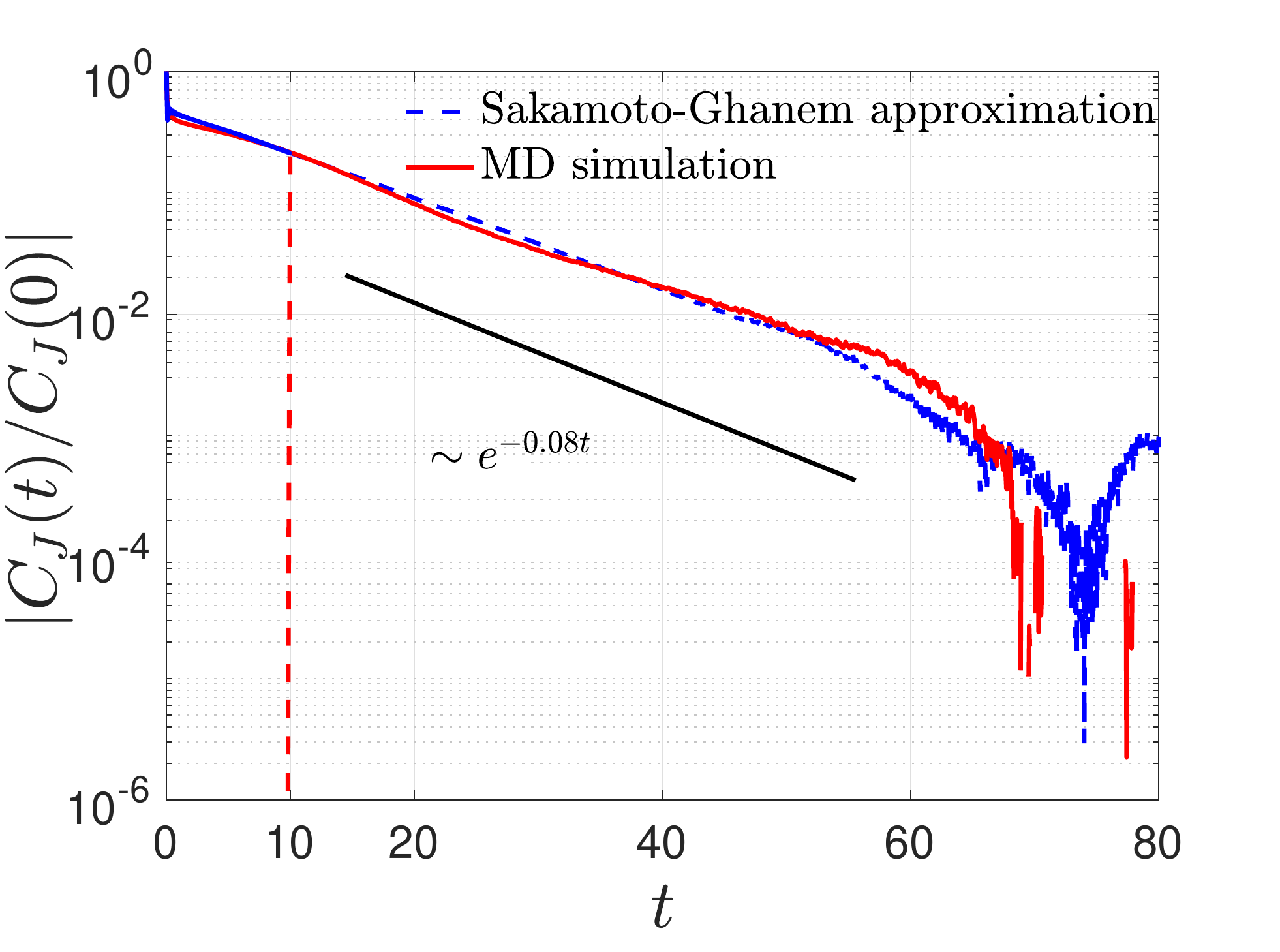}
}
\caption{Marginal probability density $\rho_J(x)$ and the normalized time autocorrelation function $C_J(t)/C_J(0)$ for the averaged heat flux $J_{av}(t)$ in the statistical equilibrium. The presented numerical simulation time domain is $[50,130]$ to ensure the system is already in the stationary state (equilibrium). Note that the short-time ($t\in [0,10]$) MD simulation for the correlation is displayed as a blue solid line and the extrapolation result based on the stochastic modeling is shown as a blue dashed line, with a vertical red dashed line separating these two. The presented simulation results are for systems with temperature $T_{eq}=1$ (First row) and $T_{eq}=5$ (Second row).}
\label{fig:J_ave_Eqn} 
\end{figure}
\paragraph{Close-to-Equilibrium case}
We now consider the close-to equilibrium case and use the classical Green-Kubo theory to verify the validity of stochastic heat conduction model. To this end, we choose $T_L=1$ and $T_R=1.1$ in the MD simulation of the SDE \eqref{SDE:n_d_heat}. By collecting the data from the nonequilibrium steady state, we find $J_{av}(t)$ is also approximately Gaussian with $J_{av}\sim \N(0.0170,0.24^2)$. Hence the reduced-order model for the close-to-equilibrium averaged heat flux $J_{av}(t)$ is given by 
\begin{align}\label{eqn:near_J_full_GLE}
\frac{d}{dt}J_{av}(t)=\sum_{i=1}^I\int_0^tk_ib_i(t-s)J_{av}(s) ds+\sum_{i=1}^M\sqrt{\lambda_i}\xi_ie_i(t).
\end{align}
In Figure \ref{fig:J_ave_Near_Eqn}, we compare the MD simulation result of the normalized time-autocorrelation function: 
\begin{align*}
    C_J(t)=\langle J_{av}(t)-\langle J_{av}\rangle_{\Delta T}, J_{av}(0)-\langle J_{av}\rangle_{\Delta T}\rangle_{\Delta T},
\end{align*}
where $\langle\cdot\rangle_{\Delta T}$ is the ensemble average with respect to the NESS probability density with $T_R-T_L=\Delta T$. Similar to the previous case, we used the NESS simulation data for $t\in[0,10]$ to construct the memory kernel. The long-term prediction of the correlation function is shown to be accurate. In addition, we calculate the averaged long-term energy accumulation:
\begin{align*}
\Delta E(t)=\left\langle\int_0^t J_{av}(s)ds\right\rangle_{\Delta T},
\end{align*}
where $J_{av}(t)$ is the MD sample path for system in the near-equilibrium steady state. For stochastic model \eqref{eqn:near_J_full_GLE}, this quantity can be calculated by averaging the solution with respect to the probability measure introduced by the i.i.d Gaussian random variables $\{\xi_i\}_{i=1}^M$ , i.e. 
\begin{align*}
\Delta E(t)=\left\langle\int_0^t J_{av}(s)ds\right\rangle_{\bm \xi=\{\xi_i\}_{i=1}^M}.
\end{align*}
Here $J_{av}(t)$ is the numerical solution of \eqref{eqn:near_J_full_GLE}. On the other hand, the Green-Kubo formula can also be used to calculate the energy accumulation as:
\begin{align}\label{eqn:Delta_E}
\Delta E(t)=\kappa\frac{dT}{dx}t\approx\kappa\frac{\Delta T}{aN}t=\frac{t}{k_BT^2_{eq}a}\int_0^{\infty}\langle J_{av}(s),J_{av}(0)\rangle_{eq}ds,
\end{align}
where $a$ is the equilibrium position spacing of the chain. For our example $a=\sigma=1$. $\langle J_{av}(s),J_{av}(0)\rangle_{eq}$ can be obtained from the equilibrium MD simulation result or the solution of \eqref{eqn:E_J_full_GLE}. Note that the second approximation in \eqref{eqn:Delta_E} is valid since normally a linear temperature profile is formed for the heat conduction model \cite{lepri2003thermal}. The simulation result verifies that both the Green-Kubo formula and the reduced-order model \eqref{eqn:near_J_full_GLE} predicate rather accurately the long-term energy accumulation. In fact, stochastic model \eqref{eqn:near_J_full_GLE} gives a different definition of thermal conductivity $\kappa$ based on {\em the second FDT}:
\begin{align}\label{model_GK}
\kappa(T,\Delta T,N):=\lim_{t\rightarrow +\infty} a\frac{1}{\Delta T}\frac{\langle\int_0^{t} J_{tot,N}(s)ds\rangle_{\Delta T}}{t}
\approx
\lim_{t\rightarrow +\infty} a\frac{1}{\Delta T}\frac{\langle\int_0^{t} J_{tot,N}(s)ds\rangle_{\bm \xi}}{t}.
\end{align}
For our example, $\kappa(T=1,\Delta T=0.1,N=256)\approx\kappa=0.0178$, where $\kappa$ is the Green-Kubo transport coefficient \eqref{def:Green-Kubo}. 
We also consider a high temperature case. By choosing $T_L=5$ and $T_R=5.25$, we found the steady state distribution for $J_{av}(t)$ is similar to the result we obtained for the high-temperature equilibrium case. Because of the non-Gaussian feature of $J_{av}(t)$, we need to use the aforementioned polynomial chaos expansion method to generate the stochastic model. The simulation result is shown in Figure \ref{fig:J_ave_Near_Eqn}, and we see that the constructed fluctuating heat conduction model faithfully captures and predicates the static and dynamical properties of $J_{av}(t)$. When comparing the accumulated energy $\Delta E(t)$, the calculated result is more accurate than the predication of the Green-Kubo formula, with an estimated conductivity (given by \eqref{model_GK}) $\kappa(T=5,\Delta T=0.25,N=256)\approx 0.0271$. This might stems from the fact that for temperature difference $\Delta T=0.25$, the system is already out of the linear response regime (see explanations below).  
\begin{figure}[t]
\centerline{
\includegraphics[height=4cm]{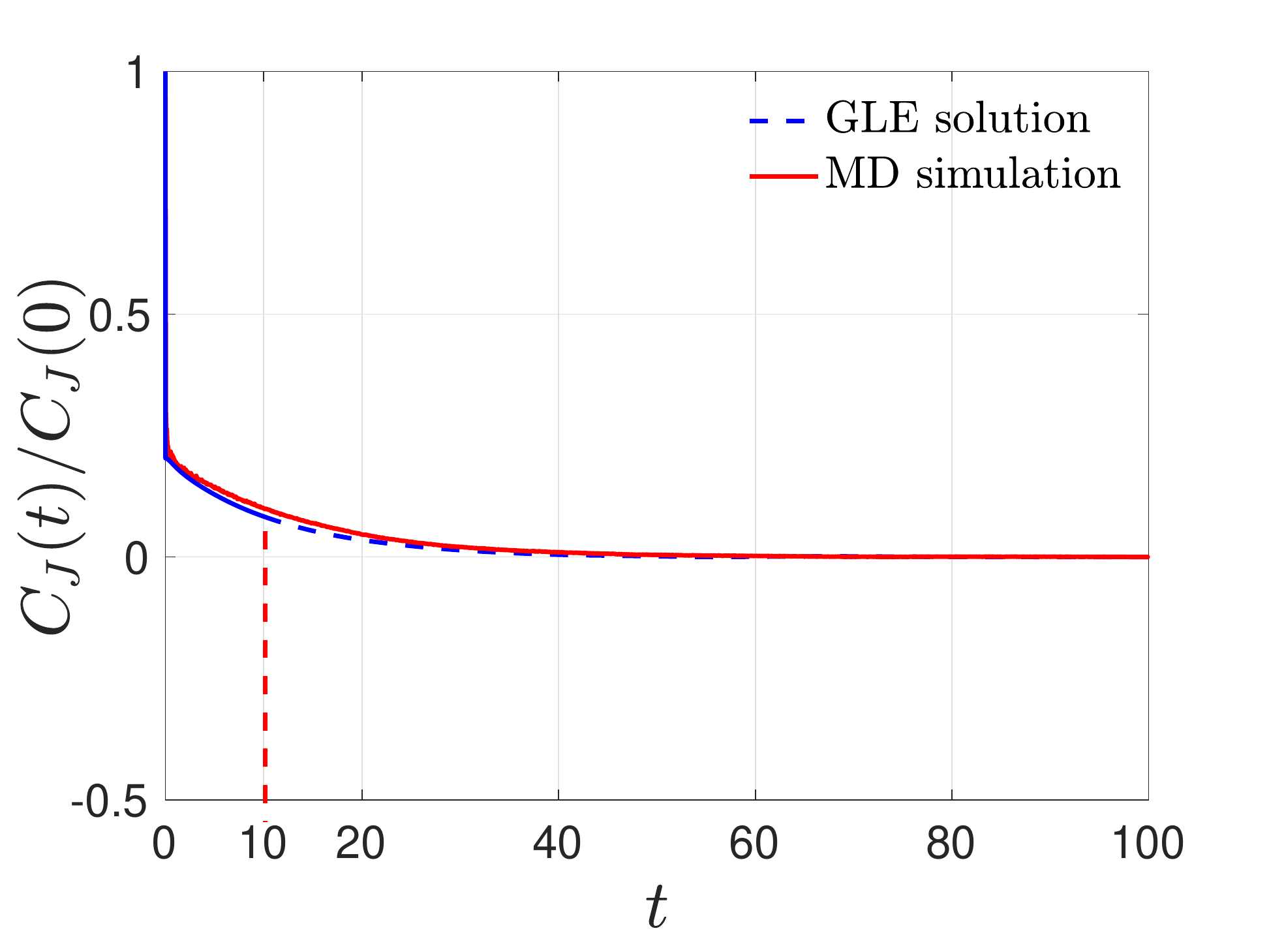}
\includegraphics[height=4cm]{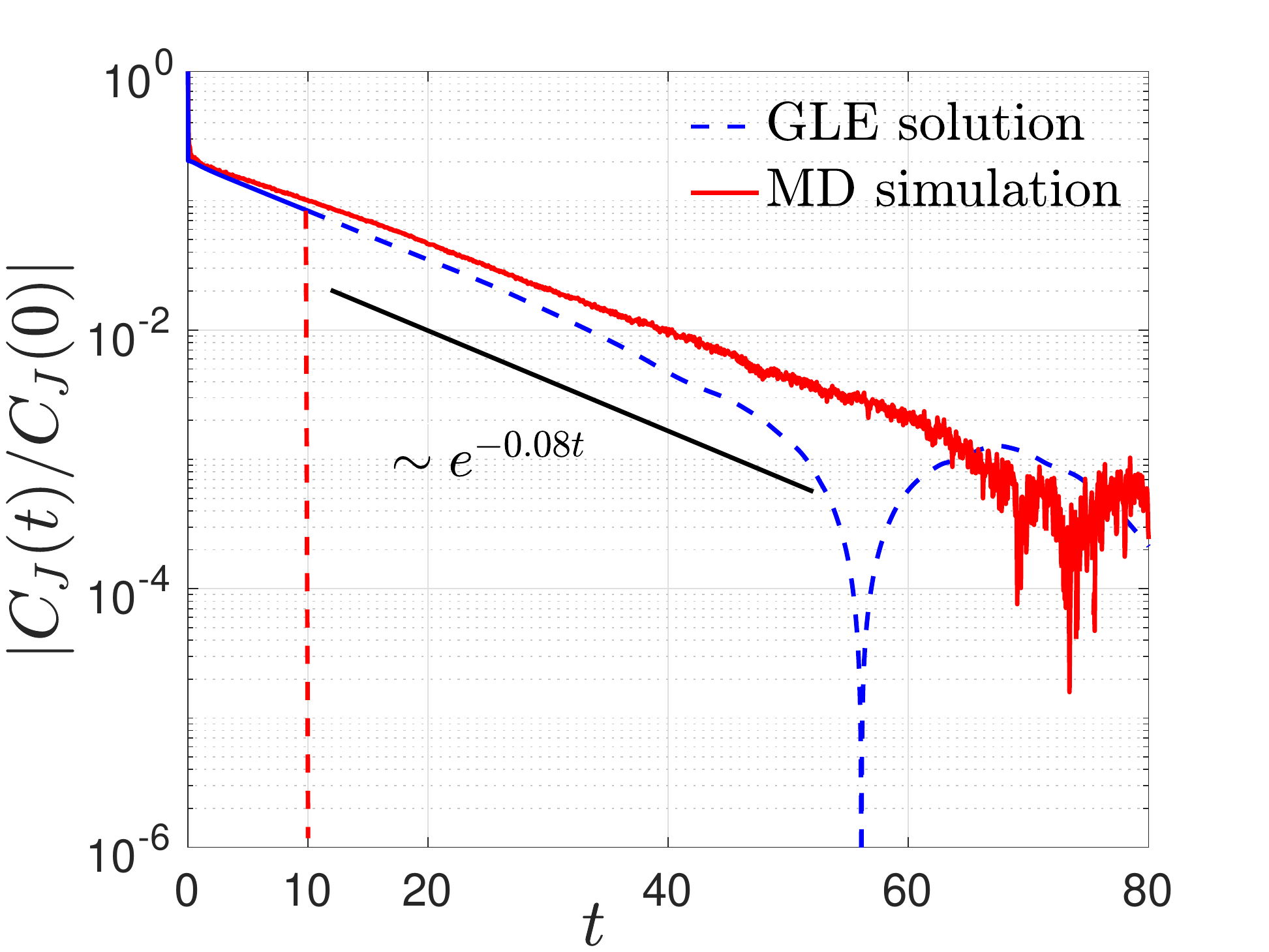}
\includegraphics[height=4cm]{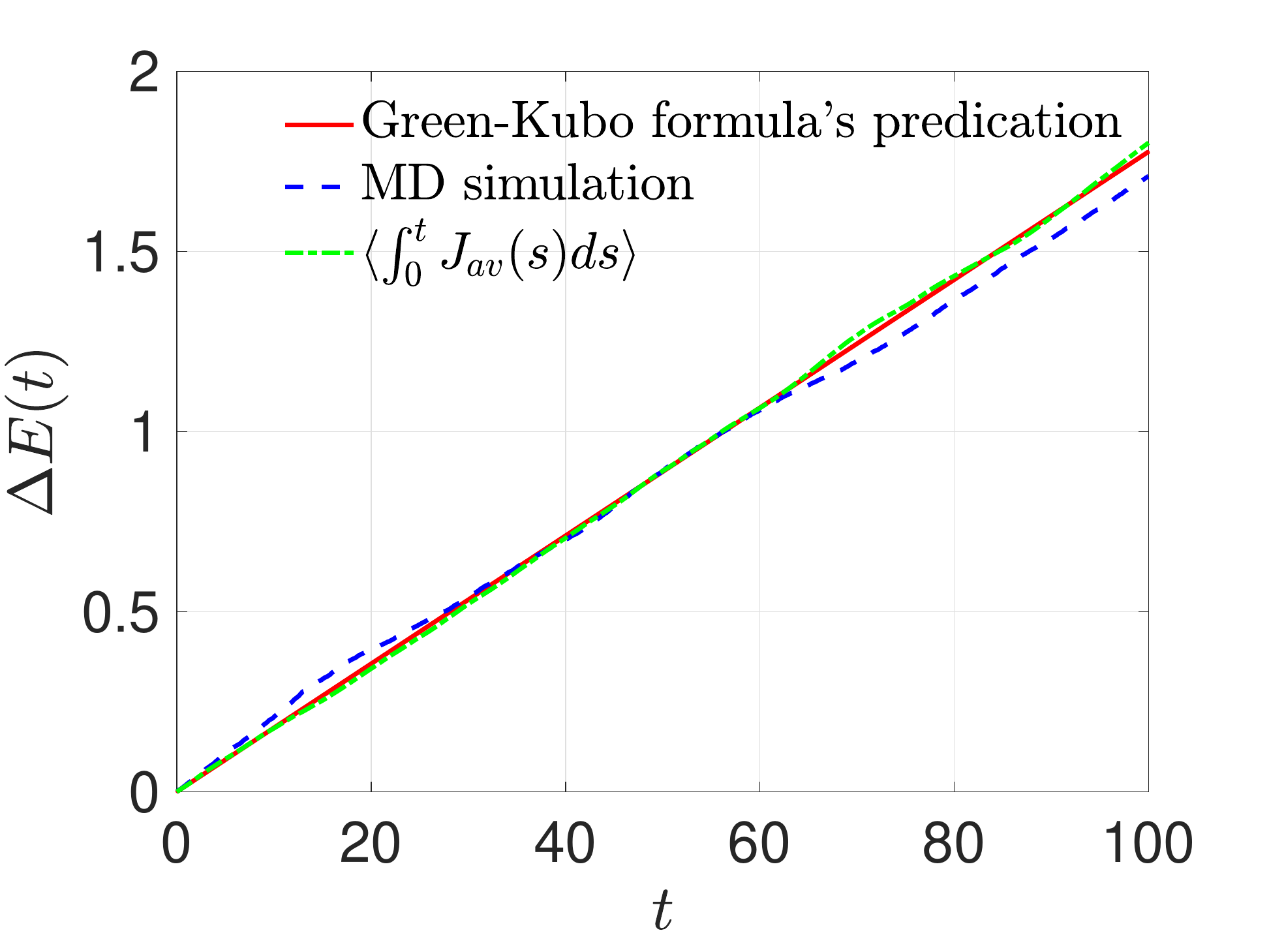}
}
\centerline{
\includegraphics[height=4cm]{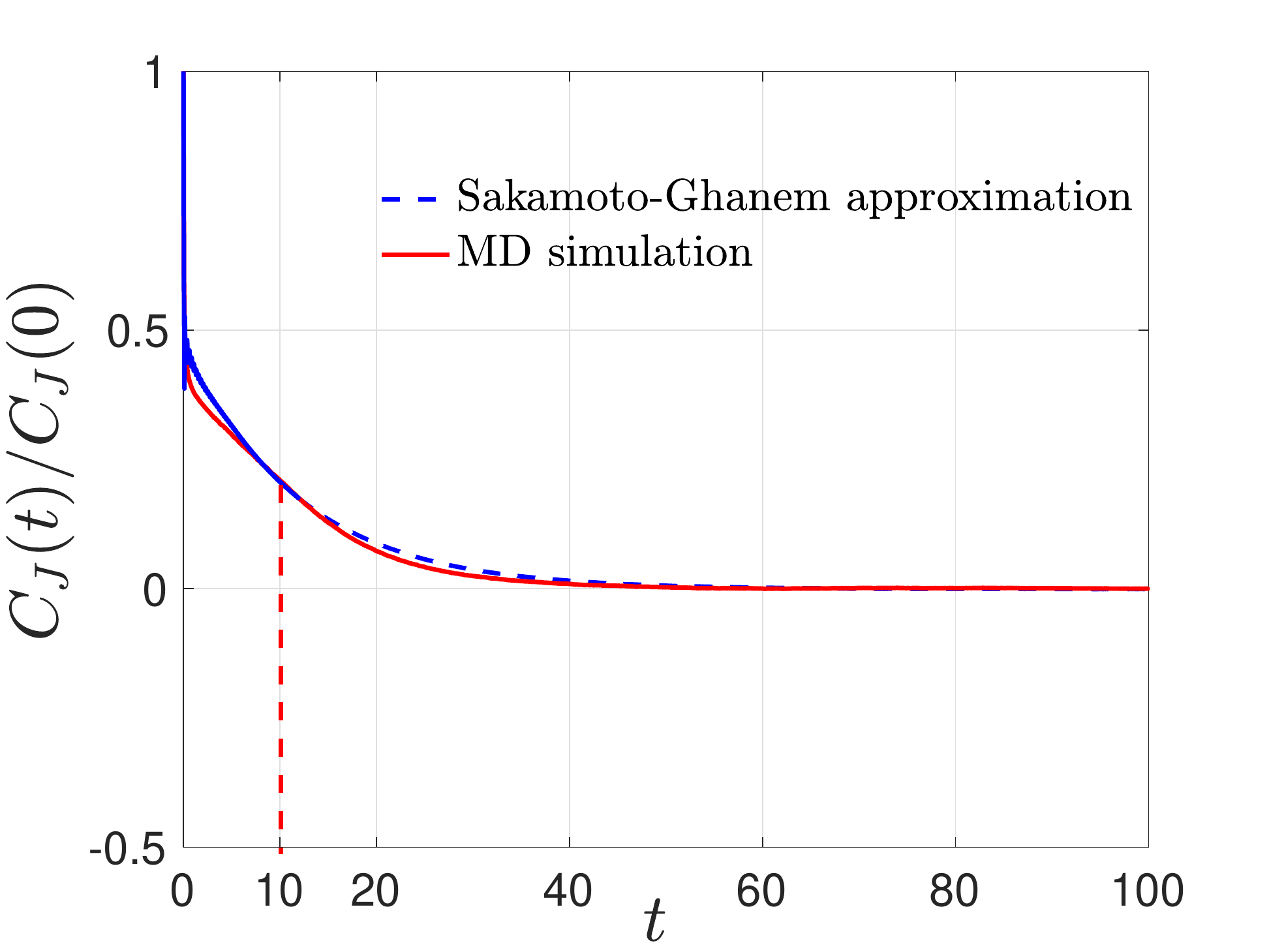}
\includegraphics[height=4cm]{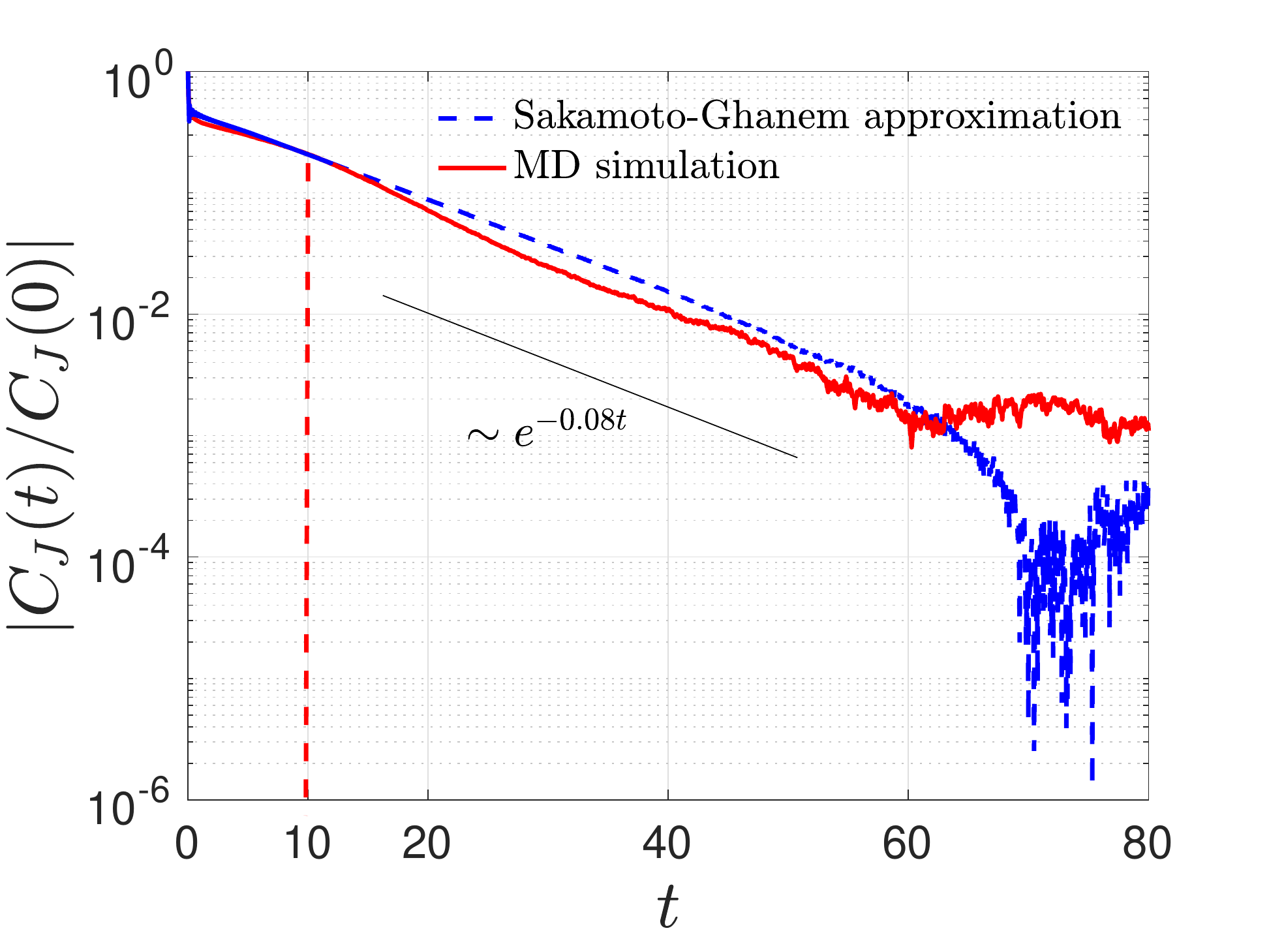}
\includegraphics[height=4cm]{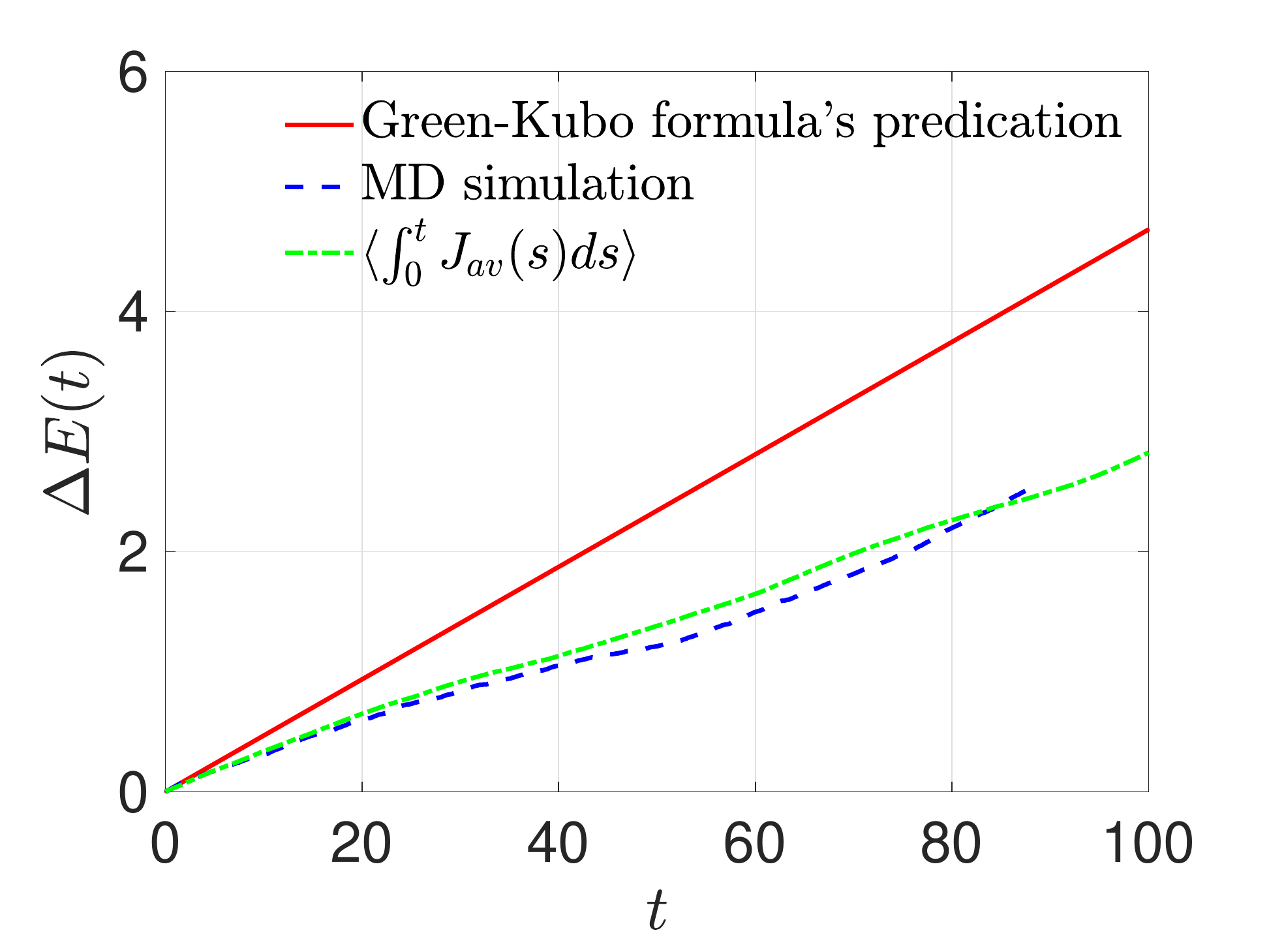}
}
\caption{Marginal probability density $\rho_J(x)$ and the normalized time autocorrelation function $C_J(t)/C_J(0)$ for the averaged heat flux $J_{av}(t)$; The presented simulation results are for close-to-equilibrium systems with temperatures $T_{L}=1,T_R=1.1$ (First row) and $T_{L}=5, T_{R}=5.25$ (Second row). Other settings are the same as in Figure \ref{fig:J_ave_Eqn}.}
\label{fig:J_ave_Near_Eqn} 
\end{figure}
\paragraph{Far-from-equilibrium case}
When a large temperature gradient $\Delta T$ is imposed to the system through the boundary thermostats, the system is outside of the linear response regime. For such far-from-equilibrium cases, the validity of the Green-Kubo formula \eqref{def:Green-Kubo} is questionable since it is based on perturbation theory. However, we can still use methods introduced in Section \ref{sec:Numerical_method} to build reduced-order models for $J_{av}(t)$ and use it to quantify the intensity of heat transfer. Form the MD simulation result presented in Figure \ref{fig:J_ave_NESS}, we find that the heat flux probability density in the far-from equilibrium steady state is non-Gaussian, strongly asymmetric and has long tails. Similar results were reported for Hamiltonian systems \cite{mendl2015current}, where the PDF was shown to approximate the Baik-Rains distribution. Here we did not pre-assume the form of the PDF and directly used the numerically evaluated density function to build a model. Specifically, since the heat flux is a non-Gaussian process, we adopt the second methodology \eqref{diagram2} and repeat what we have done for the high temperature near-equilibrium case to calculate $J_{av}(t)$ using $J_{av}(t)=\sum_{i=1}^{M}J_iH_i(\gamma(t,\bm \xi))$. Figure \ref{fig:J_ave_NESS} clearly shows that the simulated $J_{av}(t)$ predicts the energy accumulation more accurately than the Green-Kubo formula (with $T_{eq}=1,5$), as expected. The estimated conductivity for the two cases we considered are $\kappa(T=1,\Delta T=4,N=256)\approx 0.28$ and $\kappa(T=5,\Delta T=4,N=256)\approx 0.29$.

\paragraph{Remark} We may compare the classical Green-Kubo transport theory with the second FDT-induced transport theory from the information theory point of view. The Green-Kubo theory uses the equilibrium information (correlation function) of the flux to predict the near-equilibrium transport. In our method, we used the short-time information of flux to predict the long-time transport. These two methods have their own merits and drawbacks. Specifically, the Green-Kubo theory is uniformly valid for equilibrium systems with various perturbations, such as thermostats with different temperature gradients $\Delta T$. However, it only applies to near-equilibrium systems. The second FDT-induced transport theory has larger range of applicability and can be used to predict far-from-equilibrium transport. But up to this point, we can only reply on data-driven methods to calculate the GLE, which means that the heat transport for systems with different temperature gradients has to be handled on a one-to-one basis.    
\begin{figure}[t]
\centerline{
\includegraphics[height=4cm]{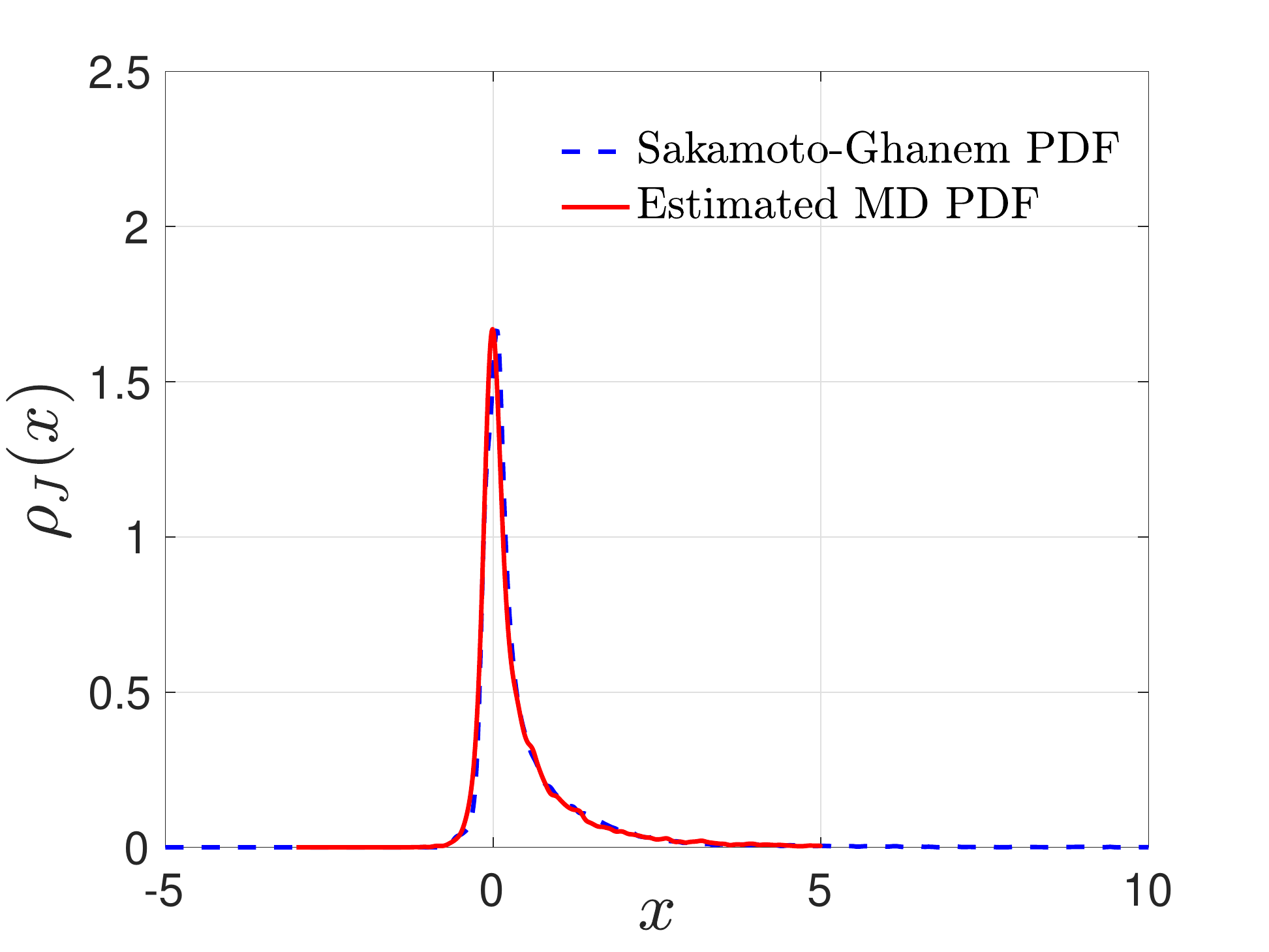}
\includegraphics[height=4cm]{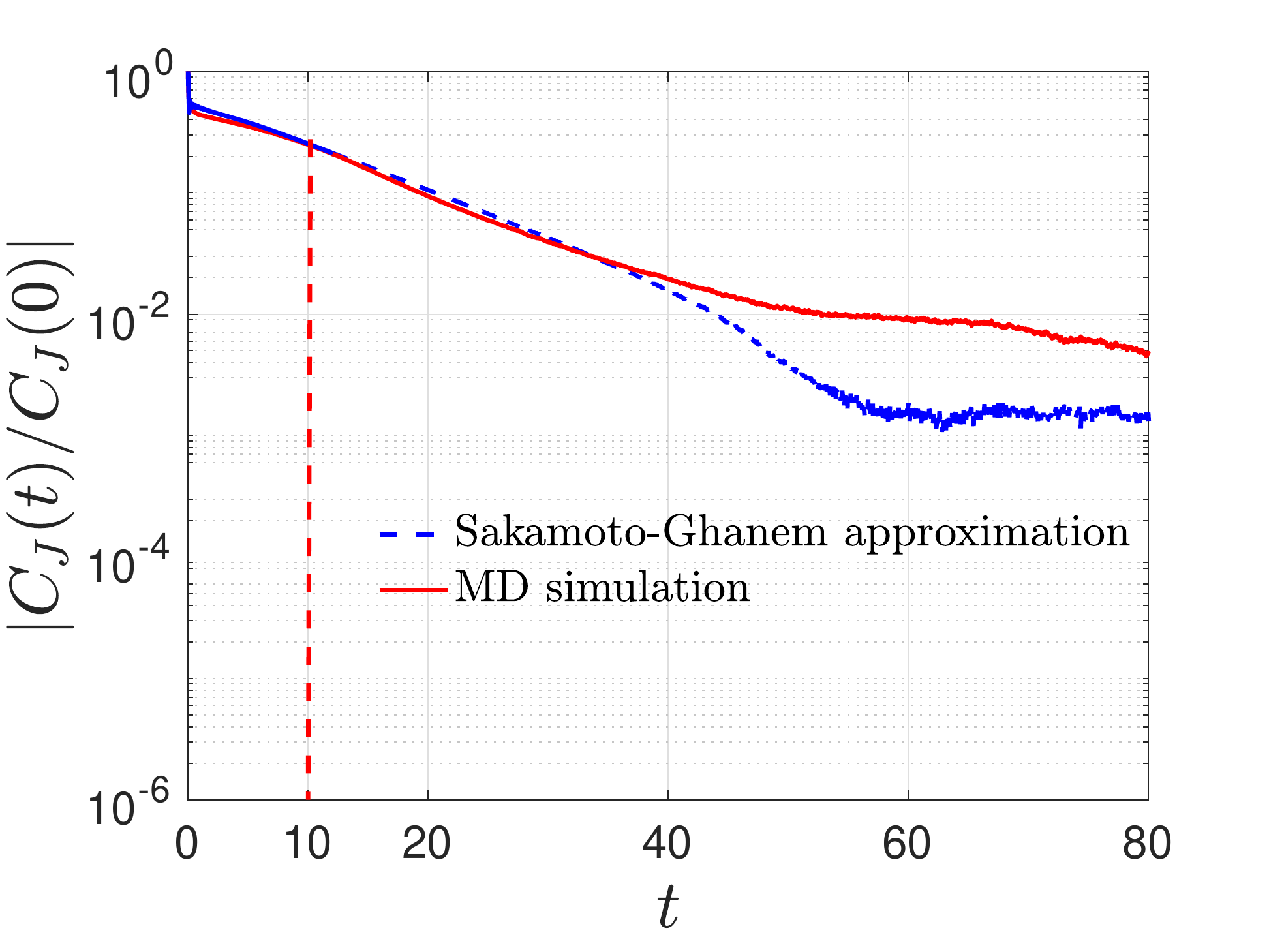}
\includegraphics[height=4cm]{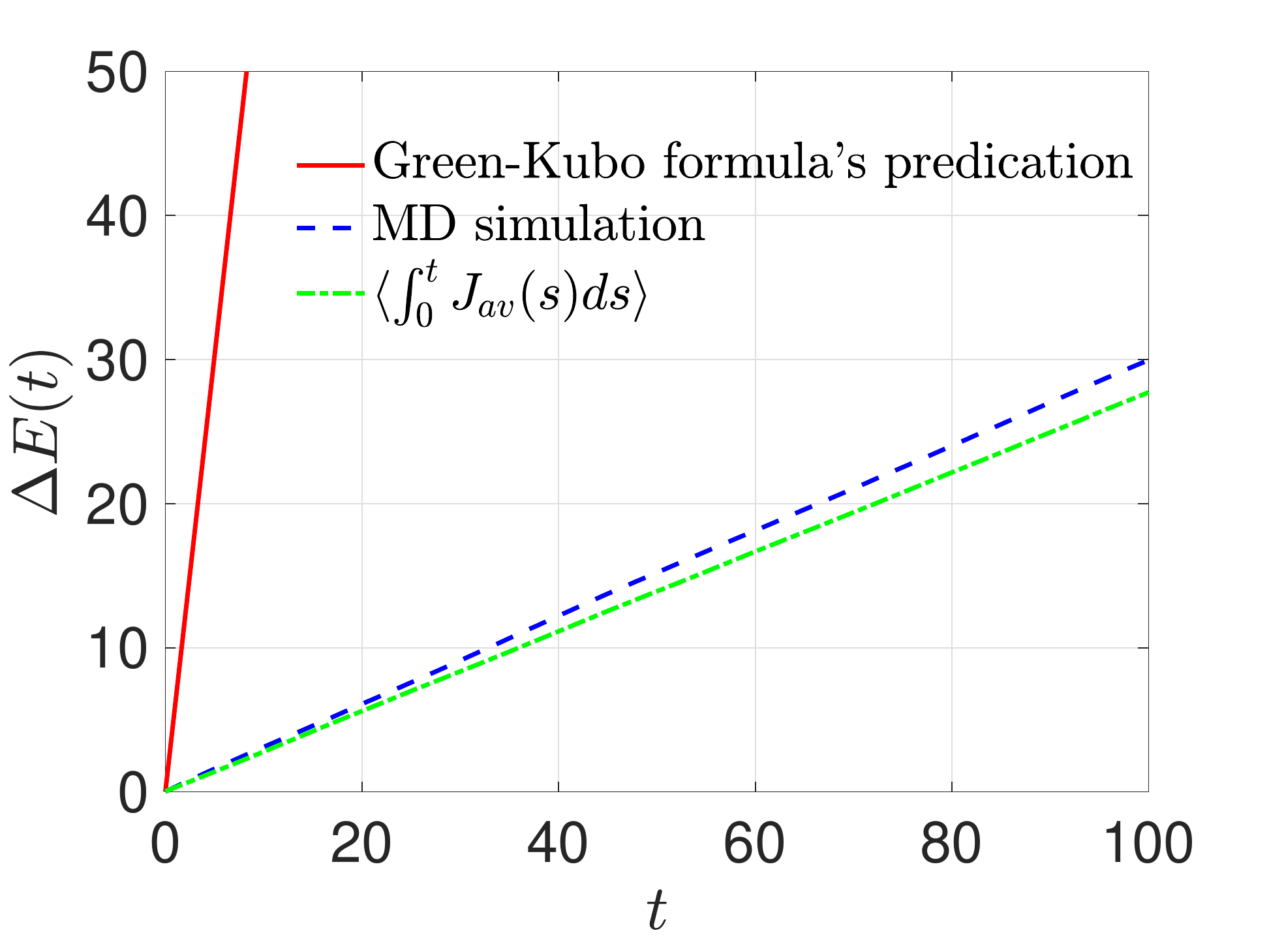}
}
\centerline{
\includegraphics[height=4cm]{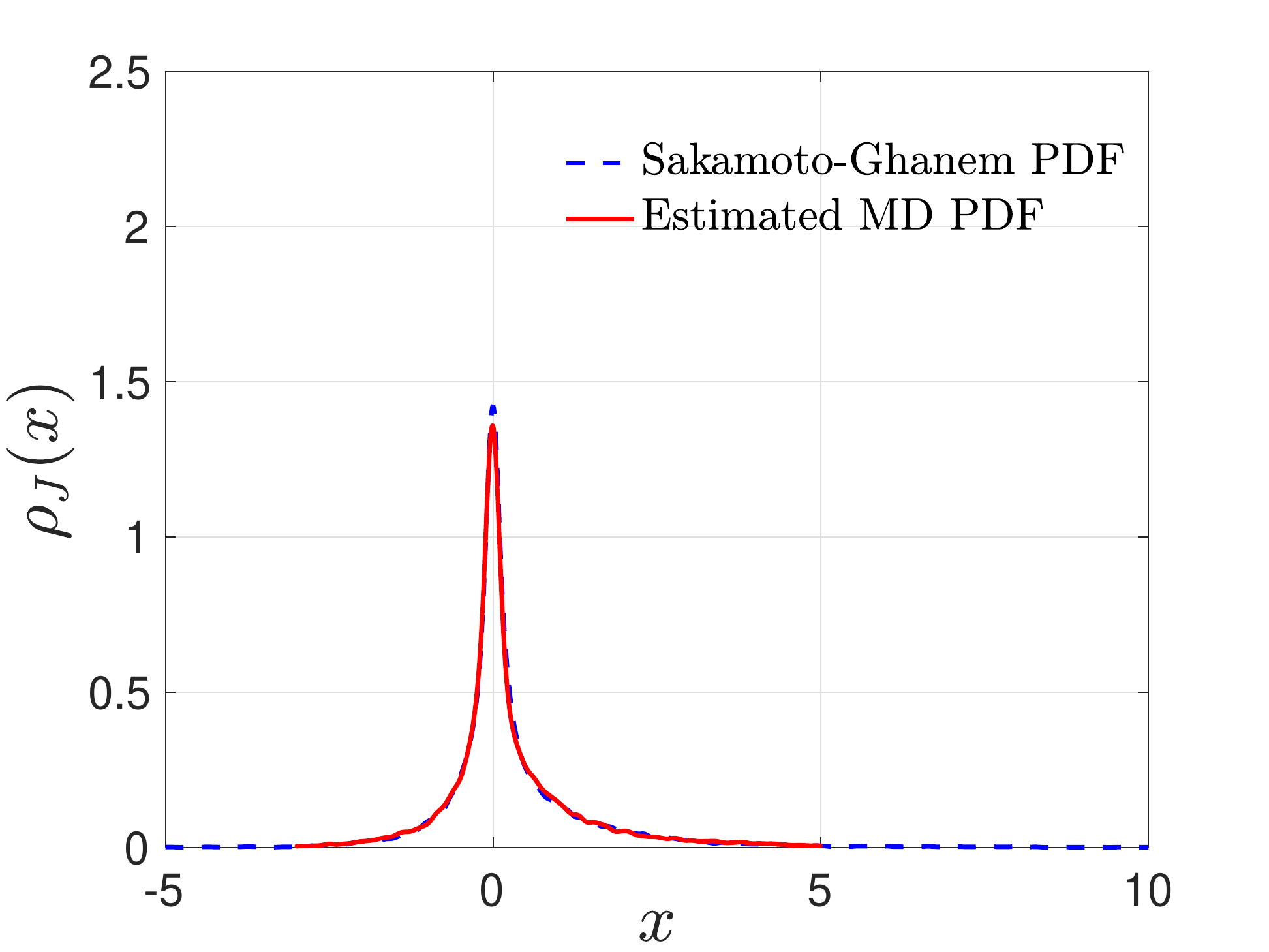}
\includegraphics[height=4cm]{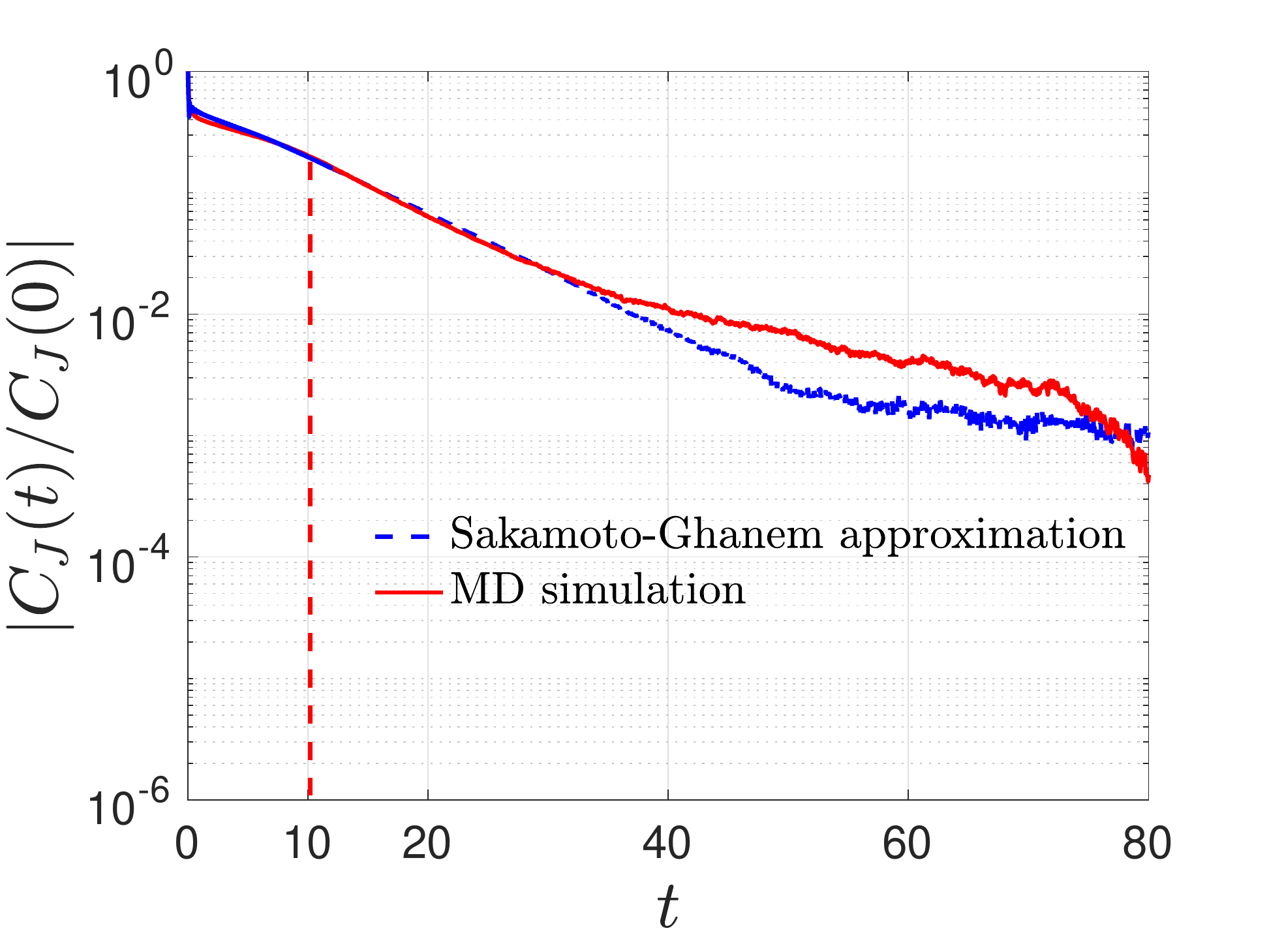}
\includegraphics[height=4cm]{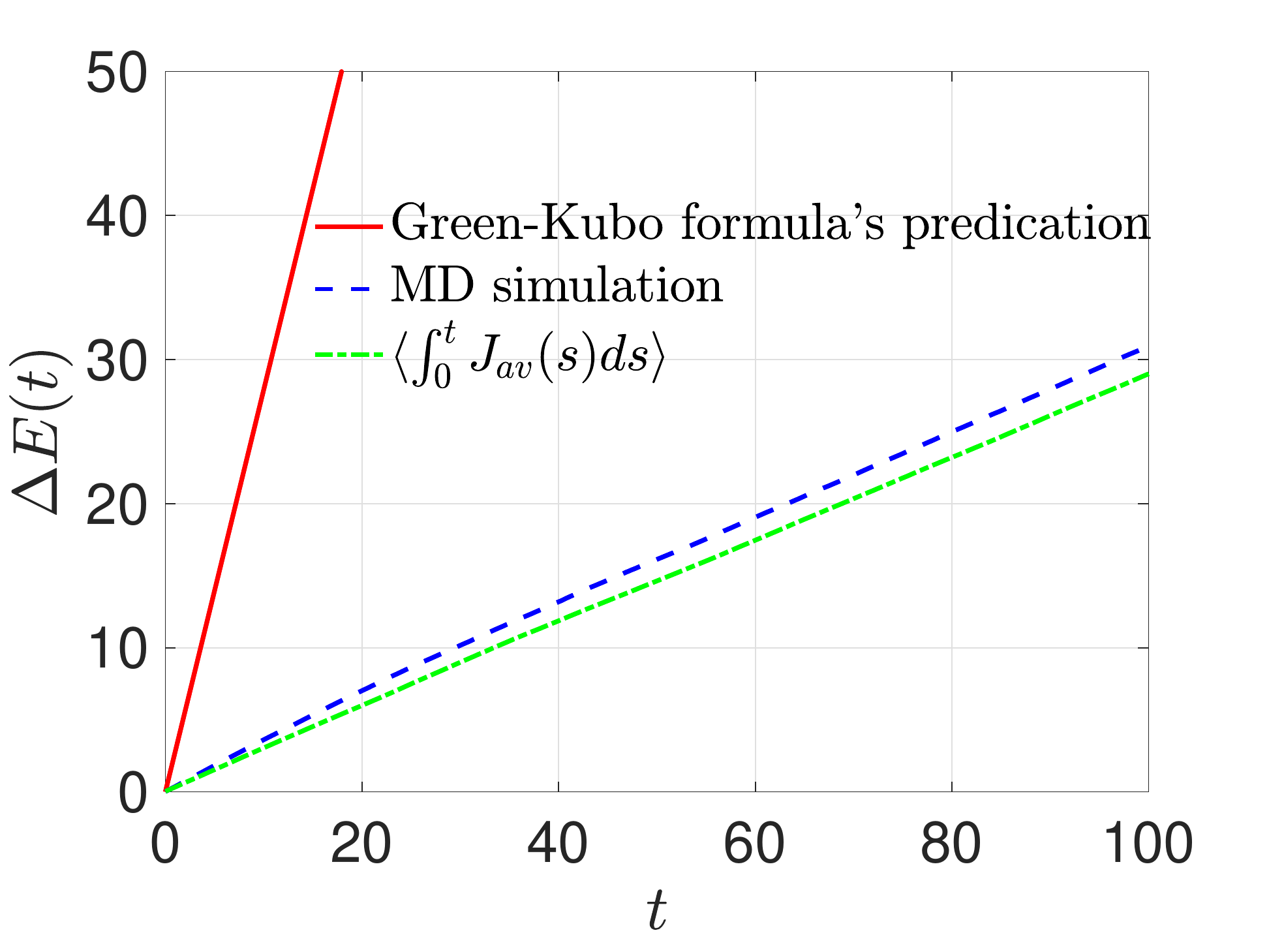}
}
\caption{Marginal probability density $\rho_J(x)$ and the normalized time autocorrelation function $C_J(t)/C_J(0)$ for the averaged heat flux $J_{av}(t)$; The presented simulation results are for far-from-equilibrium systems with temperature $T_{L}=1,T_R=5$ (First row) and $T_{L}=5, T_{R}=9$ (Second row). Other settings are the same as in Figure \ref{fig:J_ave_Eqn}.}
\label{fig:J_ave_NESS} 
\end{figure}

\section{Summary}
\label{sec:conclusion}
In this paper, we discussed the generalization of the second fluctuation-dissipation theorem (FDT) for stochastic dynamical systems driven by white noise and its application to the heat conduction model. Following Kubo-Mori's methodology, the derivation of this new FDT only uses the property of the Komogorov operator $\K$, and hence generally holds for observables in both equilibrium and nonequilibrium steady states (NESS). We also note a surprising fact that for observables in the degenerate coordinate of the Kolmogorov operator, the {\em classical} second FDT holds in the NESS
even when the steady state probability measure is generally unknown. The generalized second FDT was applied to various statistical physics models in and out of equilibrium such as the Langevin dynamics and the DPD model. In particular, we focused on a low-dimensional heat conduction model and proved the validity of the classical second FDT for the averaged heat flux. We also introduced two reduced-order modeling methods based on the second FDT. The first one works for Gaussian observables which satisfy the classical second FDT, whereas the other method is based on the polynomial chaos expansion approach which is generally applicable to observables satisfying the generalized second FDT. With these two methods, we were be able to numerically verify our main theoretical result, Theorem \ref{thm:2nd-FDT}. Moreover, we derived a fluctuating heat conduction model and introduced a new definition of thermal conductivity which holds for both near-equilibrium and far-from-equilibrium heat transport. For a nanoscale open system out of the linear response regime, this second-FDT-induced conductivity yields more accurate prediction of heat transfer than the classical Green-Kubo theory. We conclude by emphasizing that the presented work can be  generalized to other transport processes. Further applications of the proposed methodology can be expected.  

\vspace{0.3cm}
\noindent 
{\bf Note added in proof}. 
We note that similar conclusions on the validity of the classical second FDT for far-from-equilibrium systems was discovered independently by Jung and Schmid \cite{jung2021fluctuation} soon after we submitted this manuscript to arXiv. In their work, more numerical evidence from the DPD simulation of an colloid particle immersed in fluids is provided which supports our claims, in particular, the result (II) announced in Section \ref{sec:short_sum}. Moreover, a mathematical argument based on the deterministic Volterra equation is shown to provide a different proof of the classical second FDT in the nonequilibrium. In fact, their derivation shows that the necessary condition of Corollary \ref{corollary_2ndFDT} is also sufficient which means $\Omega=0$ if and only if $w(0)=0$. We believe our methodology provides a solid mathematical validation of Jung and Schmid's result using the stochastic system Mori-Zwanzig equation, and it is constructive to compare the theoretical and numerical results in these two manuscripts. 

\appendix
\section{Proof of the generalized first FDT}\label{sec:APP1-FDT_proof}
For the perturbed stochastic system \eqref{eqn:sde_p} evolving from the steady state $\rho(0)=\rho$, the evolution equation of the probability density is given by the Fokker-Planck equation:
\begin{align}\label{eqn:pert_rho}
\frac{\partial}{\partial t}\rho(t)=\K^*(t)\rho(t),
\end{align}
where $\K^*(t)$ is the Kolmogorov forward (Fokker-Planck) operator of the form
\begin{align*}
\K^*(t)=\K^*+\K^*_{ext}(t)=\K^*+\left[\sum_{i}\delta G_i(t)\tilde{F}_i(\bm x)\partial_{x_i}+\sum_i\delta H^2_i(t)\tilde{\sigma}_{i}^2(\bm x)\partial_{x_i}^2\right]^*.
\end{align*}
Since the added forces are of the order $O(\delta)$, which is assumed to be a small quantity, the solution of equation \eqref{eqn:pert_rho} can be written as a perturbative expansion of $\rho(t)$, i.e.
\begin{align}\label{pert_series}
\rho(t)=\rho+\delta\rho_1(t)+\delta^2\rho_2(t)+\cdots,
\end{align}
where the steady state solution satisfies $\K^*\rho=0$. By substituting the perturbation series \eqref{pert_series} into equation \eqref{eqn:pert_rho} and retaining only the first order term $O(\delta)$, we get the formal expression of $\delta \rho_1(t)$:
\begin{align*}
\delta\rho_1(t)=\int_0^tdse^{(t-s)\K^*}\K^*_{ext}(s)\rho.
\end{align*}
Hence in the phase space, for observable $u=u(\bm x)$, the change of ensemble average induced by the perturbation can be expressed as:
\begin{align*}
\la u(t)\rangle_{\rho_{\delta}}-\la u(t)\ra&=\left\langle\int_0^te^{(t-s)\K^*}\K^*_{ext}(s)ds\rho u(0)\right\rangle+O(\delta^2)\\
&=\left\langle\rho\int_0^t\K_{ext}(s)e^{(t-s)\K}dsu(0)\right\rangle+O(\delta^2)=\int_0^t\la\K_{ext}(s)u(t-s)\ra ds+O(\delta^2),
\end{align*} 
where $\rho_{\delta}=\rho_{\delta}(t)$ is the formal solution of the Fokker-Planck equation \eqref{eqn:pert_rho}. Using the integration by parts formula to simplify $\langle\K_{ext}(s)u(t-s)\rangle_{\rho}$, we obtain the generalized linear-response formulas \eqref{Generalized_1st_FDT}-\eqref{Generalized_1st_FDT2}.
\section{Proof of the generalized second FDT \eqref{W_sigma}}
\label{sec:APP-2FDT_proof}
Following the proof of Theorem \ref{thm:2nd-FDT}, we have 
\begin{align*}
    \L+\L_{\rho}^*=-\nabla\cdot\bm F(\bm x)-\frac{1}{\rho}\L\rho.
\end{align*}
Using integration by parts twice and a shorthand notation $\D=\sum_{i,j=1}^N\sigma_i(\bm x)\sigma_j(\bm x)\partial_{x_ix_j}^2 =\sum_{i,j=1}^N\sigma_i\sigma_j\partial_{x_ix_j}^2$, we have
\begin{align*}
    \D_{\rho}^*&=\sum_{i,j=1}^N\partial_{x_i}\partial_{x_j}[\sigma_i\sigma_j(\cdot)]+\frac{1}{\rho}\partial_{x_i}\rho\partial_{x_j}[\sigma_i\sigma_j(\cdot)]+\frac{1}{\rho}\partial_{x_j}\rho\partial_{x_i}[\sigma_i\sigma_j(\cdot)]+\frac{1}{\rho}\sigma_i\sigma_j\partial_{x_ix_j}^2\rho\\
    &=\D+\frac{1}{\rho}\D\rho+\frac{1}{\rho}\partial_{x_i}\rho\partial_{x_j}[\sigma_i\sigma_j(\cdot)]+\frac{1}{\rho}\partial_{x_j}\rho\partial_{x_i}[\sigma_i\sigma_j(\cdot)]+\partial_{x_i}[\sigma_i\sigma_j]\partial_{x_i}+\partial_{x_i}[\sigma_i\sigma_j]\partial_{x_i}+\partial_{x_ix_j}^2[\sigma_i\sigma_j]\\
    &=\D+\frac{1}{\rho}\D\rho+\A,
\end{align*}
where $\A$ is a second-order differential operator. Then it is straightforward to obtain
\begin{align*}
    \D_{\rho}^*+\D&=2\D+\frac{1}{\rho}\D\rho+\A, \\
    \K^*_{\rho}+\K&=-\nabla\cdot\bm F(\bm x)-\frac{1}{\rho}\L\rho+\frac{1}{\rho}\D\rho+2\D+\A.
\end{align*}
Since $\K^*\rho=\L^*\rho+\D^*\rho=0$, $\L+\L^*=-\nabla\cdot\bm F(\bm x)$ and $\D^*=\D+\B$, where operator $\B$ is defined as 
\begin{align*}
    \B=\sum_{i,j=1}^N\partial_{x_i}[\sigma_i\sigma_j]\partial_{x_j}+\partial_{x_j}[\sigma_i\sigma_j]\partial_{x_i}+\partial^2_{x_ix_j}[\sigma_i\sigma_j],
\end{align*}
we now obtain
\begin{equation}
\left.\begin{aligned}
\L\rho+\L^*\rho&=-\rho\nabla\cdot\bm F(\bm x)\\
\D^*\rho&=\D\rho+\B\rho
\end{aligned}\right\}
\Rightarrow
-\frac{1}{\rho}\L\rho+\frac{1}{\rho}\D\rho-\nabla\cdot\bm F(\bm x)=-\frac{1}{\rho}\B\rho.
\end{equation}
By substituting the above equality into $\K_{\rho}^*+\K$, we can get the desired result:
\begin{equation}
\begin{aligned}   
   \K^*_{\rho}+\K&=-\frac{1}{\rho}\B\rho+2\D+\A\\
   &=2\D+\sum_{i,j=1}^N\frac{1}{\rho}\sigma_i\sigma_j\partial_{x_i}\rho\partial_{x_j}+\frac{1}{\rho}\sigma_i\sigma_j\partial_{x_j}\rho\partial_{x_i}+\partial_{x_i}[\sigma_i\sigma_j]\partial_{x_j}+\partial_{x_j}[\sigma_i\sigma_j]\partial_{x_i}\\
   &=\S,
\end{aligned}
\end{equation}
where $\S$ is defined as \eqref{W_sigma}.

\section{The generalized second FDT for nonlinear GLEs}\label{sec:APP_Zwanzig}
The nonlinear GLEs are normally derived using the Zwanzig-type projection operator \cite{zwanzig1961memory,chorin2002optimal,zhu2018estimation}:
\begin{align}\label{Zwanzig_P}
    (\P f)(\hat {\bm x})=\frac{\int f(\bm x)\rho(\bm x)d\tilde {\bm x}}{\int\rho(\bm x)d\tilde {\bm x}},\qquad d\bm x=\underbrace{dx_1dx_2,\cdots dx_m}_{d\hat{\bm x}}\underbrace{dx_{m+1}dx_{m+2},\cdots dx_N}_{d\tilde{\bm x}}=d\hat{\bm x}d\tilde{\bm x}.
\end{align}
Zwanzig's projection \eqref{Zwanzig_P} is a conditional expectation operator satisfying $\P^2=\P$. It is also a symmetric operator in $L^2(M,\rho)$, i.e. $\langle\cdot,\P\cdot\rangle_{\rho}=\langle\P\cdot,\cdot\rangle_{\rho}$. Using \eqref{Zwanzig_P} to derive the Mori-Zwanzig equation for stochastic system \eqref{eqn:sde}, we obtain a nonlinear GLE:
\begin{align}\label{intro_MZ_Z}
   \frac{\partial}{\partial t}\underbrace{e^{t\K}\bm u(0)}_{\bm u(t)}
=\underbrace{e^{t\K}\P\K\bm u(0)}_{\bm F(\bm u(t))}
+\underbrace{\int_0^te^{s\K}\P\K
e^{(t-s)\Q\K}\Q\K \bm u(0)ds}_{\int_0^t \bm R(t-s,s) ds}+\underbrace{e^{t\Q\K}\Q\K \bm u(0)}_{\bm f(t)},
\end{align}
where $\bm F(\bm u(t))$ is a nonlinear vector function of $\bm u(t)$. The memory integral $\int_0^t\bm R(t-s,s)ds$ is not directly given by a convolution form such as the one in \eqref{gle_full}. To simplify this equation, we note that in the the conditional projection \eqref{Zwanzig_P} is an infinite-rank operator \cite{chorin2002optimal,zhu2018estimation}, therefore can be formally written as \cite{hislop2012introduction}:
\begin{align*}
    \P=\sum_{j=1}^{\infty}\langle \cdot, \psi_j(0)\rangle_{\rho}\phi_j(0),
\end{align*}
where $\{\psi_j(0)\}_{j=1}^{\infty}$ and $\{\phi_j(0)\}_{j=1}^{\infty}$ are set of phase space functions which are linearly independent in $L^2(M,\rho)$. As a result, the Zwanzig-equation memory kernel is given by:  
\begin{align*}
\int_0^t e^{s\K}\P\K e^{(t-s)\Q\K}\Q\K u_i(0)ds
&=\int_0^t e^{s\K}\sum_{j=1}^{\infty}\langle\psi_j(0),\K e^{(t-s)\Q\K}\Q\K u_i(0)\rangle_{\rho}\phi_j(0)ds\\
&=\int_0^t\underbrace{\sum_{j=1}^{\infty}\langle\psi_j(0),\K e^{(t-s)\Q\K}\Q\K u_i(0)\rangle_{\rho}}_{K_{ij}(t-s)}\phi_j(s)ds.
\end{align*}
Hence using the proof of Theorem \ref{thm:2nd-FDT}, the memory kernel can be rewritten as:
\begin{align}
K_{ij}(t)=\sum_{j=1}^{\infty}\langle \psi_j(0),\K e^{t\Q\K}\Q\K u_i(0)\rangle_{\rho}
&=\sum_{j=1}^{\infty}-\langle\K\psi_j(0),\K e^{t\Q\K}\Q\K u_i(0)\rangle_{\rho}+\langle\S\psi_j(0),\K e^{t\Q\K}\Q\K u_i(0)\rangle_{\rho}\nonumber\\
&=\sum_{j=1}^{\infty}-\langle\Q\K\psi_j(0),e^{t\Q\K}\Q\K u_i(0)\rangle_{\rho}+\langle\S\psi_j(0),e^{t\Q\K}\Q\K u_i(0)\rangle_{\rho}\nonumber\\
&=\sum_{j=1}^{\infty}-\langle g_j(0),f_i(t)\rangle_{\rho}+\langle h_j(0),f_i(t)\rangle_{\rho},
\label{nonlinear_GLE_FDT}
\end{align}
where $g_j(0)=\Q\K\psi_j(0)$ and $h_j(0)=\S\psi_j(0)$. After all these steps, we obtain a nonlinear GLE:
\begin{align}\label{nonlinear_Z_GLE}
\frac{\partial}{\partial t}u_i(t)
=F_i(\bm u(t))+\sum_{j=1}^{\infty}\int_0^tK_{ij}(t-s)\phi_j(s)ds+f_i(t),
\end{align}
where the memory kernel and the fluctuation force satisfy the generalized second FDT \eqref{nonlinear_GLE_FDT}. The above derivation is of course quite formal. However such mathematical rigorous discussion clearly indicates for stochastic systems the classical second FDT has to be generalized by including more observables such as $g_j(0)$ and $h_j(0)$ in the time correlations with the noise $f_i(t)$. At the current stage, it seems difficult or even unnecessary to get the explicit expressions for terms such as $\psi_j(0)$, $\phi_j(0)$ and $g_j(0)$, $h_j(0)$. However, nonlinear GLE \eqref{nonlinear_Z_GLE} and the generalized second FDT \eqref{nonlinear_GLE_FDT} may be used as the {\em ansatz} for proper reduced-order modellings for stochastic systems, such as the ones used in \cite{lei2016data}.  
\section{Validity of the classical second FDT for the averaged heat flux}\label{sec:App1}
It is sufficient to prove the result for $d=1$ system and the generalization to multi-dimensional cases is direct. The local heat flux at the boundaries have to be handled independently. Since the Langevin forces act on the momentum coordinates of the boundary oscillators, the heat reservoir only exchanges energy with the oscillators through the kinetic part \cite{lepri2003thermal}. From this observation we may define the boundary heat flux as the non-Hamiltonian contribution of the time derivative of the kinetic energy i.e., 
\begin{align*}
J_b&=\gamma_bk_BT_b\partial_{p_b}^2\frac{p_b^2}{2}-\gamma_b p_b\partial_{p_b}\frac{p_b^2}{2}=\gamma_bk_BT_b-\gamma_bp_b^2,\qquad b\in\B.
\end{align*}
In the steady state, a stable kinetic temperature profile is often obtained through numerical experiments and the boundary oscillators admit a kinetic temperature $T_b^k\neq T_b$ \cite{lepri2003thermal}. We further assume the marginal distribution for the momentum of the boundary oscillators $p_{b}$, $b\in \B$ is Gaussian, i.e. 
\begin{align}\label{eqn:6}
\rho_{p_b}=\int\rho\  d\bm q d\bm p\setminus\{p_b\}\propto e^{-\beta_b\frac{p_b^2}{2}},
\end{align}
where $\bm p\setminus\{p_b\}$ represents all momenta $\bm p$ but $p_b$, $\beta_b=1/k_BT_b^k$ and $\rho=\rho_{NESS}$ is the probability density corresponding to the NESS. To be noticed that this assumption is verified numerically in Section \ref{sec:veri-2nd_FDT} (see Figure \ref{fig:PLPM}). In the stationary regime, we get the boundary heat flux ensemble average 
\begin{align*}
\la J_b(t)\ra&=\gamma_bk_B(T_b-T_b^k),\\
\la J_b^2(t)\ra&=\gamma_b^2k_B^2(3(T_b^k)^2+T^2_b-2T_b^kT_b).
\end{align*}
Hence for the total heat flux we have
\begin{align*}
\la J_{tot,N}(t)\ra=\la J_{\G\setminus\B}(t)\ra+\la J_{\B}(t)\ra=\la  J_{\G\setminus\B}(t)\ra+\sum_{b\in\B}\gamma_bk_B(T_b-T_b^k),
\end{align*}
where $J_{\G\setminus\B}$ and $J_{\B}$ denote respectively the bulk and boundary contributions to the total heat flux. As it is shown in \cite{lepri2003thermal}, in the stationary regime the average bulk heat flux $\la J_{\G\setminus\B}(t)\ra$ is an extensive quantity which scales as the order of volume for the lattice system i.e. $O(L^d)$, where $L$ is the length of the lattice system. For fixed coupling constant $\lambda_b$ and thermal bath temperature $T_b$, the average boundary heat flux $\la J_{\B}(t)\ra$ scales as the order of surface area for the lattice system. i.e. $O(L^{d-1})$. Naturally since $N\propto L^d$, we have
\begin{align}\label{est_1}
\lim_{N\rightarrow\infty}\frac{1}{N}\la J_{tot,N}(t)\ra=\frac{1}{N}\la J_{\G\setminus\B}(t)\ra.
\end{align}
The above limit can be seen as the first-order moment estimation of the total heat flux in the stationary regime. We also need to verify whether the above estimation hold in the second order.  To this end, we use the triangle inequality to get 
\begin{align}\label{eqn:ine}
 \|J_{\G\setminus\B}(t)\|_{L^2(\rho)}-\|J_{\B}(t)\|_{L^2(\rho)}\leq \|J_{tot,N}(t)\|_{L^2(\rho)}\leq \|J_{\G\setminus\B}(t)\|_{L^2(\rho)}+\|J_{\B}(t)\|_{L^2(\rho)}.
\end{align}
Then by combining another triangle inequality 
\begin{align}\label{eqn:ine2}
\|J_{\B}(t)\|_{L^2(\rho)}\leq \sum_{b}\|J_b(t)\|_{L^2(\rho)}=\sum_b[\gamma_b^2k_B^2(3(T_b^k)^2+T^2_b-2T_b^kT_b)]^{1/2},
\end{align}
we have 
\begin{align*}
 \|J_{\G\setminus\B}(t)\|_{L^2(\rho)}-\sum_b\|J_{b}(t)\|_{L^2(\rho)}\leq \|J_{tot,N}(t)\|_{L^2(\rho)}\leq \|J_{\G\setminus\B}(t)\|_{L^2(\rho)}+\sum_b\|J_{b}(t)\|_{L^2(\rho)}.
\end{align*}
Using the embedding inequality associated with $L^2(\rho)\rightarrow L^1(\rho)$ and Cauchy-Schwartz inequality, we get that 
\begin{align}\label{eqn:ine3}
C|\la J_{\G\setminus\B}(t)\ra|\leq C \|J_{\G\setminus\B}(t)\|_{L^1(\rho)}\leq  \|J_{\G\setminus\B}(t)\|_{L^2(\rho)},
\end{align}
where $C$ is the embedding constant. Inequality \eqref{eqn:ine2} also implies $\|J_{\B}(t)\|_{L^2(\rho)}$ scales at most of the order $O(L^{d-1})$, while \eqref{eqn:ine3} indicates $\|J_{\G\setminus\B}(t)\|_{L^2(\rho)}$ scales at least of the order $O(L^d)$. Hence dividing by $N$ and then taking the limit $N\rightarrow\infty$ in both sides of inequality \eqref{eqn:ine}, we have
\begin{align}\label{est_2}
\lim_{N\rightarrow\infty}\frac{1}{N}\|J_{tot,N}(t)\|_{L^2(\rho)}=\frac{1}{N}\|J_{\G\setminus\B}(t)\|_{L^2(\rho)}.
\end{align}
Using the similar technique and the stationary condition $\|J_{\G\setminus\B}(t)\|_{L^2(\rho)}=\|J_{\G\setminus\B}(0)\|_{L^2(\rho)}$, $\|J_{\B}(t)\|_{L^2(\rho)}=\|J_\B(0)\|_{L^2(\rho)}$, it is easy to obtain the following estimate
\begin{align}\label{eqn:ine4}
\la J_{\G\setminus\B}(t),J_{\G\setminus\B}(0)\ra-2&\|J_{\G\setminus\B}(t)\|_{L^2(\rho)}\|J_{\B}(0)\|_{L^2(\rho)}-\|J_{\B}(0)\|^2_{L^2(\rho)}\leq 
\la J_{tot,N}(t),J_{tot,N}(0)\ra  \nonumber \\ 
&\leq
\la J_{\G\setminus\B}(t),J_{\G\setminus\B}(0)\ra+2\|J_{\G\setminus\B}(t)\|_{L^2(\rho)}\|J_{\B}(0)\|_{L^2(\rho)}+\|J_{\B}(0)\|^2_{L^2(\rho)}.
\end{align}
For the same reason, dividing by $N^2$ and then taking the limit $N\rightarrow\infty$ in both sides of inequality \eqref{eqn:ine4}, we obtain
\begin{align}\label{est_J_tot}
\lim_{N\rightarrow\infty}\frac{1}{N^2}\la J_{tot,N}(t),J_{tot,N}(0)\ra=\frac{1}{N^2}\la J_{\G\setminus\B}(t),J_{\G\setminus\B}(0)\ra.
\end{align}
Estimates \eqref{est_1} and \eqref{est_J_tot} imply that the averaged total heat flux $J_{av,N}(t)$, as a stationary second-order process, can be approximated by its bulk contribution $J_{\G\setminus\B,N}(t)/N$ in the $L^2$ sense.
\section{Polynomial chaos expansion for stationary non-Gaussian processes}\label{sec:App_poly}
For strongly non-Gaussian stochastic processes such as the far-from equilibrium heat flux $J_{av}(t)$, the KL expansion is no longer suitable to represent such processes since the random coefficients $\xi_k$ are not i.i.d Gaussian and cannot be easily determined. Some methods were proposed to address the approximation problem of a non-Gaussian process. For instance, Chu and Li \cite{chu2017mori}
suggested to use a Gaussian multiplicative noise to approximate the random force in the extended stochastic system. Zhu and Venturi \cite{zhu2019generalized} used Phoon's algorithm \cite{phoon2002simulation,phoon2005simulation} to generate a sample-based, iterated KL expansion to represent the non-Gaussian process. In this paper, we adopt a modified Sakamoto-Ghanem \cite{sakamoto2002polynomial} method to approximate $J_{av}(t)$. The numerical merits of the new algorithm are highlighted in two aspects. First, it is generally appliable to non-Gaussian stochastic process $u(t)$ with {\em arbitrary} steady state distribution $\rho_{u}$ and correlation function $C(t_1,t_2)=\langle u(t_1),u(t_2)\rangle$. Secondly, it works much more efficiently when comparing to similar approaches such as Phoon's algorithm \cite{phoon2002simulation,phoon2005simulation}, which enables us to rapidly generate sample trajectories from the GLE. Here we only briefly explain the main steps of Sakamoto-Ghanem algorithm and our modification of it. More details and extension of such a method can be found in \cite{sakamoto2002polynomial}.

Suppose $u(t)$ is a second order, stationary non-Gaussian stochastic process with an arbitrary steady state distribution $\rho_{u}$ and the stationary correlation function $C_u(t+s,s)=C_u(t,0)=\langle u(t),u(0)\rangle$. Then the following polynomial chaos expansion approximates $u(t)$ in the $L^2$ sense as $M\rightarrow+\infty$ and $K\rightarrow+\infty$:
\begin{align}\label{saka_process_u}
u(t)=\sum_{i=1}^{M}u_iH_i(\gamma(t,\bm \xi)), \qquad \text{where} \quad 
\gamma(t,\bm \xi)=\sum_{j=1}^K\sqrt{\lambda_j}\xi_je_j(t). 
\end{align}
Here $H_i$ is the $i$-th order probabilist's Hermite polynomials. Being represented as a truncated KL expansion which can be determined by solving the Fredholm equation \eqref{Fredholm_eqn}, $\gamma(t,\bm \xi)$ has steady correlation function $C_{\gamma}(t)=\langle \gamma(t,\bm \xi),\gamma(0,\bm \xi)\rangle$ satisfying the following algebraic equation:
\begin{align}\label{saka_process_alge_eqn}
C_u(t)=\sum_{j=1}^{M}\frac{u_j^2j!}{\sum_{k=1}^Mu_k^2k!}C^j_{\gamma}(t),
\end{align}
where $u_j$ are the coefficients of the Hermite-chaos expansion of the random variable $u$ with probability density $\rho_u$. Specifically, $u_i$ can be obtained (see more details in \cite{xiu2002wiener}, Section 6) using Gaussian quadrature by evaluating the integral: 
\begin{align}\label{saka_process_ui}
u_i=\frac{1}{\langle H_i^2\rangle}\int_0^1U^{-1}(x)H_i(G^{-1}(x))dx,
\end{align}
where $U^{-1}(x)$ and $G^{-1}(x)$ are, respectively, the inverse of the cumulative distribution function (CDF) for an arbitrary random variable $u\sim \rho_u$ and a Gaussian random variable $g\sim \N(0,1)$. In the original version of the Sakamoto-Ghanem algorithm \cite{sakamoto2002polynomial}, algebraic equation \eqref{saka_process_alge_eqn} is solved {\em exactly} for $C_u(t)$ and $C_{\gamma}(t)$ in a discrete lattice $\{0,\Delta t,2\Delta t,\cdots,T\}$ by assuming that $C_u(t)>0$ for $t\in [0,T]$. Here we generalize the algorithm by solving \eqref{saka_process_alge_eqn} {\em approximately} in the same lattice for the part $C_u(t)<0$. i.e. For the lattice point $i\Delta t$ such that $C_u(i\Delta t)<0$, we find an approximated solution $C_{\gamma}(i\Delta t)$ in $[-1,1]$ for the following algebraic equation:
\begin{align}\label{saka_process_alge_eqn_e}
C_u(i\Delta t)=\sum_{j=1}^{M}\frac{u_j^2j!}{\sum_{k=1}^Mu_k^2k!}C^j_{\gamma}(i\Delta t).
\end{align}
To apply the modified Sakamoto-Ghanem algorithm in the stochastic modeling of the far-from equilibrium heat flux $J_{av}(t)$, we only need to input the NESS probability density $\rho_J$ and correlation function $C_J(t)$ which can be obtained numerically via the MD simulation. We finally obtain the polynomial chaos expansion of $J_{av}(t)$
\begin{align}\label{saka_process_J}
J_{av}(t)=\sum_{i=1}^{M}J_iH_i(\gamma(t,\bm \xi)).
\end{align}

\bibliographystyle{plain}
\bibliography{2nd_FDT_NESS}
\end{document}